\documentclass{article}
\usepackage[a4paper]{geometry}

\usepackage{amsmath,verbatim,amsthm,amssymb,amsfonts,amscd, graphicx,natbib}
\usepackage{graphics,xcolor}
\usepackage{multicol}
\usepackage{tikz,enumerate}
\usepackage[framemethod=TikZ]{mdframed}
\usepackage{tikz}
\usepackage[framemethod=TikZ]{mdframed}

\usepackage[margin=1cm,font=small]{caption}
\usepackage{url}

\usetikzlibrary{arrows,automata,decorations.markings}
\usetikzlibrary{patterns,decorations}
\usetikzlibrary{positioning}
\usetikzlibrary{calc}

\tikzstyle{vertex}=[circle, draw, fill, inner sep=0pt, minimum size=0.15cm]

\newcommand\X{\mathbf{X}}
\newcommand\Y{\mathbf{Y}}
\newcommand\A{\mathbf{A}}
\newcommand\Z{\mathbf{Z}}
\newcommand\Ptrain{\mathbb{P}_{\text{train}}}
\newcommand\Etrain{\mathbb{E}_{\text{train}}}

\DeclareMathOperator*{\argmin}{argmin}

\usetikzlibrary{arrows}
\newdimen\arrowsize
\pgfarrowsdeclare{arcsq}{arcsq}
{
  \arrowsize=0.2pt
  \advance\arrowsize by .5\pgflinewidth
  \pgfarrowsleftextend{-4\arrowsize-.5\pgflinewidth}
  \pgfarrowsrightextend{.5\pgflinewidth}
}
{
  \arrowsize=1.5pt
  \advance\arrowsize by .5\pgflinewidth
  \pgfsetdash{}{0pt} \pgfsetroundjoin   \pgfsetroundcap    \pgfpathmoveto{\pgfpoint{0\arrowsize}{0\arrowsize}}
  \pgfpatharc{-90}{-140}{4\arrowsize}
  \pgfusepathqstroke
  \pgfpathmoveto{\pgfpointorigin}
  \pgfpatharc{90}{140}{4\arrowsize}
  \pgfusepathqstroke
}

\usepackage{fullpage}

\newtheorem{theorem}{Theorem}
\newtheorem{example}{Example}

\newtheorem{lemma}{Lemma}
\newtheorem{proposition}{Proposition}

\title{Anchor regression: heterogeneous data meet causality}

\author{Dominik Rothenh\"ausler, Nicolai Meinshausen, Peter B\"uhlmann and Jonas Peters}

\begin{document}

\maketitle
\abstract{
We consider the problem of predicting a response variable from a set of covariates on a data set that differs in distribution from the training data. Causal parameters are optimal in terms of predictive accuracy if in the new distribution either many variables are affected by interventions or only some variables are affected, but the perturbations are strong. If the training and test distributions differ by a shift, causal parameters might be too conservative to perform well on the above task. This motivates anchor regression, a method that makes use of exogeneous variables to solve a relaxation of the “causal” minimax problem by considering a modification of the least-squares loss. The procedure naturally provides an interpolation between the solutions of ordinary least squares and two-stage least squares. We prove that the estimator satisfies predictive guarantees in terms of distributional robustness against shifts in a linear class; these guarantees are valid even if the instrumental variables assumptions are violated. If anchor regression and least squares provide the same answer (“anchor stability”), we establish that OLS parameters are invariant under certain distributional changes. Anchor regression is shown empirically to improve replicability and protect against distributional shifts.

}

\section{Introduction}
A substantial part of contemporaneous datasets are not collected under carefully designed experiments.
  Furthermore, data collected
from different sources are
often heterogeneous due to, e.g., changing
circumstances, batch effects, unobserved confounders or time-shifts in the
distribution. These heterogeneities or “perturbations” make it difficult to
gain actionable knowledge that generalizes well to new data
sets. Approaches to deal with inhomogeneities include robust methods
\citep{huber1964robust,huber1973robust}, mixed effects models
\citep{pinheiro2000linear}, time-varying coefficient models
\citep{hastie1993varying,fan1999statistical} and maximin effects
\citep{meinshausen2015maximin}.

On the other hand there is a growing literature on causal inference under
various types of assumptions and different frameworks, with applications
ranging from public health to biology and economics
\citep{lauritzen1988local,Bollen1989,greenland1999causal,Spirtes2000,robins2000marginal,dawid2000causal,rubin2005causal,Pearl2009,Peters2017}.
Often the goal is to find the
causes of some response variable $Y$ among a given set of covariates $X$ or
to quantify the causal relationships between a set of variables.  There are
two main reasons why one is
interested in
the identification and
quantification of causal
effects. On one hand, it answers questions of the type ``what happens to
variable $Y$ if we intervene on variable $X$'', perhaps being the
  classical viewpoint of causality.
On the other hand, predictions based on a causal model, that is, using the conditional mean of $Y$ given all its causal predictors, 
will in general work
equally well under arbitrary perturbations (interventions) on the covariates
  and thus, this provides an answer to the problem of generalization to new
  data sets mentioned above. The  invariance property
for prediction across interventions
  or perturbations has recently been exploited for causal inference
\citep{peters2016causal} and a form of invariance plays a
  crucial role here as well.

In causal inference, one often considers so-called hard interventions that set some covariates to a certain value. In this paper, we instead consider interventions that shift the distribution of a target variable, which corresponds to an intervention on a variable that enters the target equation linearly.
Using causal concepts for prediction under
heterogeneous data seems attractive due to invariance guarantees under
arbitrary shifts. In practice, however, exact invariance guarantees may be too conservative and can come with
 a price of
subpar predictive performance on observational and moderately shifted data. We propose a balanced approach for trading off predictive
performance on observational data and predictive
performance on perturbed (shifted)
new data, with rigorous
optimality guarantees under specific sets of perturbations or
  interventions. This can be cast as a form of distributional robustness, as discussed
next.
We consider being robust to
interventional shifts in a particular class of models where identifying functionals yield a
particularly elegant modification of OLS loss, with future work possibly allowing
more general types of robustness to be developed. In addition to
  distributionally robust prediction, we will also consider the problem
  of distributionally robust variable selection. In this context,
  distributionally robust variable selection refers to the question
  whether a statistical parameter is invariant under certain distribution
  changes. Distributionally robust prediction and variable selection are
  closely related, as we will see below.

\subsection{Distributionally robust prediction and variable selection}
In a linear setting, the goal of distributionally robust prediction can be expressed as the optimization problem
\begin{equation} \label{eq:optimrob}
\min_{b \in \mathbb{R}^{d}} \max_{F \in \mathcal{F}} \mathbb{E}_F[(Y-X^\intercal b)^{2}],
\end{equation}
where $X$ is a $d$-dimensional vector of covariates, $Y$ is the
target variable of interest, $\mathcal{F}$ is a class of distributions, and $\mathbb{E}_F$ takes the expectation w.r.t.\ $F
\in \mathcal{F}$.
Choosing different classes $\mathcal{F}$ results in
estimators with different properties, see for example
\citet{sinha2017certifiable,gao2017wasserstein,meinshausen18robust}.
We  first discuss two well-known choices of $\mathcal{F}$ and the corresponding estimators.

\subsubsection{No perturbations and ordinary least squares}
If
$\mathcal{F}$ contains only the training (or observational) distribution, we write $\Etrain$ and the optimization problem~\eqref{eq:optimrob} becomes ordinary least squares,
$$
b_{\text{OLS}} =
\argmin_{b} \Etrain[(Y-X^\intercal b)^{2}].
$$
This does not take into account any
distributional robustness.
The sample version substitutes
  $\Etrain$ by the sample mean over the observed data
  resulting in ordinary least squares estimation. We discuss in
  Section~\ref{sec:relatedwork} that $\ell_{2}$- and $\ell_1$-norm regularized regression can
  also be derived from a sample version of~\eqref{eq:optimrob} for a
  suitable class $\mathcal{F}$.

\subsubsection{Intervention perturbations and causality}
Assume now that the distribution $(X,Y)$ is induced by an (unknown) linear causal model, e.g., a linear structural causal model,
an example of which we will see in
Section~\ref{sec:setting-notation}.If the class $\mathcal{F}$ contains all interventions on subsets of variables not including $Y$,
then the
optimizer of~\eqref{eq:optimrob}
is the vector of causal coefficients
\citep[e.g.,][Theorem~1]{rojas2015causal}.
That is,
\begin{eqnarray}\label{eq:33b}
  b_{\text{causal}} = \argmin_{b} \max_{F \in {\cal F}} \mathbb{E}_{F}[(Y-X^\intercal b)^{2}],
\end{eqnarray}
for ${\cal F}$ containing all interventions on (components) of $X$. 
Similarly, the causal parameters are optimal if in all distributions $F \in \mathcal{F}$
there are hard interventions on all parents and children of $X$ (here, the interventions do not  need to be arbitrarily strong). Both of these results are direct implications of well-known invariance properties of causal models \citep{haavelmo1944,aldrich1989,Pearl2009}.

In this spirit, a causal model can be seen as a prediction mechanism that works best under interventions on subsets of $X$
that are arbitrarily strong or affect many variables.
Under the training distribution, however, this solution is usually not as good as $b_{\text{OLS}}$,
\begin{equation}
 \Etrain[(Y-X^\intercal b_{\text{causal}})^{2}] \ge \min_{b} \Etrain[(Y-X^\intercal b)^{2}] = \Etrain[(Y-X^\intercal b_{\text{OLS}})^{2}],
\end{equation}
with a potentially large difference.
Hence in many cases, estimating the causal parameter leads to conservative predictive performance compared to standard prediction methods.
The OLS solution 
on the other hand,
can have arbitrarily high predictive error when the test distribution is obtained under an
intervention.

This paper suggests a trade-off between these two estimation principles.
Several relaxations of the problem in equation~\eqref{eq:33b} are possible.
Instead of protecting against arbitrarily strong interventions
one can protect against interventions up to a certain size (norm).
Also, perturbations
in some directions
may be more important than in other directions.
Alternatively, instead of protecting against interventions on all subsets of variables $X_{1},\ldots, X_{d}$, one can attempt to find out which variables $S \subseteq \{X_{1},\ldots, X_{d}\}$ are likely to be perturbed in the future. Then one can protect against interventions on the variables in $S$. For example, we might know (e.g., through background knowledge) that shifts in the distribution of $X_{1}$ are
more likely than shifts in the distribution of $X_{2}$ on future data
sets, which may be included in the class
$\mathcal{F}$.

In
this paper, we
propose a new estimation principle, called \emph{anchor
regression}, see~\eqref{eq:anchor}.
We will see that under a linearity assumption,
the proposed estimator
can be written
as a solution to~\eqref{eq:optimrob},
where the class ${\cal F}$ consists of certain shift interventions,
  i.e.,
  interventions that shift numerical variables by
  a certain amount, which then propagate through the system.

  \subsubsection{Distributional replicability}

Distributional replicability aims to understand whether a statistical parameter is stable under certain distributional changes. Replicability in this sense is distinctly different from statistical uncertainties due to finite samples, but closely related to the concepts of invariance and distributionally robust prediction.   In the case of ordinary least-squares, it can be formalized as follows. The goal is to investigate whether
\begin{equation*}
\argmin_{b \in \mathbb{R}^{d}}  \mathbb{E}_F[(Y-X^\intercal b)^{2}] \approx \argmin_{b \in \mathbb{R}^{d}}  \mathbb{E}_{F'}[(Y-X^\intercal b)^{2}],
\end{equation*}
for all $F,F' \in \mathcal{F}$, where $\mathcal{F}$ is some set of distributions. For example, two researchers may collect data about the same research question in different locations. Due to different circumstances, the data may come from two distributions $F \neq F'$. Even if the researchers use the same OLS model, they might get different estimates if the estimator is sensitive to small distributional changes.

We will see that \emph{anchor regression} can be used to assess distributional replicability of OLS parameters across a certain set of distributions $\mathcal{F}$. 

\subsection{Our contribution}\label{sec:our-contribution}
We propose an estimator that regularizes ordinary least
squares with a penalty encouraging some form of invariance as
  mentioned above. The setting relies on the presence of exogenous variables which generate heterogeneity.
We denote by $A \in \mathbb{R}^{q}$
such exogeneous variables and call them
``anchors''. If $A$ is discrete, dummy encoding can be used in a pre-processing step to obtain $A \in \mathbb{R}^q$.
Let $X$ and $Y$ be predictors and target variable, and assume that all variables are centered and have finite variance. Let further $P_A$ denote the $L_2$-projection on the linear span from the components of $A$
     and write $\mathrm{Id}(Z) := Z$.
We then define, for $\gamma >0$,
the solution $b^{\gamma}$ to the population version of \emph{anchor
regression} as
\begin{equation}\label{eq:anchor}
  b^{\gamma} := \argmin_{b}  \Etrain[((\mathrm{Id}- \mathrm{P}_{A})(Y-X^\intercal b))^{2}] + \gamma  \Etrain[( \mathrm{P}_{A}(Y-X^\intercal b))^{2}],
\end{equation}
where $\Etrain$ denotes the expectation over the observational or training distribution.

Turning to the finite-sample case, let $\X \in \mathbb{R}^{n \times d}$ be a matrix containing observations of $X$. Analogously, the matrix containing observations of $A$ is denoted by $\A \in \mathbb{R}^{n \times q}$, and the vector containing the observations of $Y$ is denoted by $\Y \in \mathbb{R}^n$. We recommend a simple plug-in estimator for the \emph{anchor regression} coefficient $b^\gamma$:
\begin{equation}
  \hat b^{ \gamma} = \argmin_{b} \|(\mathrm{Id}-\Pi_{\A})( \Y- \X b) \|_{2}^{2} +  \gamma \|  \Pi_{\A} ( \Y- \X b)\|_{2}^{2},
\end{equation}
where $\Pi_{\A} \in \mathbb{R}^{n \times n}$ is the matrix that projects on the column space of $\A$, i.e., if $  \A^{\intercal} \A$ is invertible, then $\Pi_{\A} := \A ( \A^{\intercal} \A)^{-1} \A^{\intercal}$. 
Martin Emil Jakobsen realized that the family of finite sample estimators of anchor regression coincides with what is known as $k$-class estimators. These estimators have been suggested to improve IV-type estimation of structural parameters \citep{theil1958economic,nagar1959bias}.
 In the high-dimensional case where $d > n$, an $\ell_1$-penalty can be added to encourage sparsity. Computation of $\hat b^\gamma$ is
simple as it can be obtained by running a least-squares regression of $
\tilde \Y := (\mathrm{Id} + (\sqrt{\gamma}-1) \Pi_\A) \Y$ on $\tilde \X := (\mathrm{Id} + (\sqrt{\gamma}-1) \Pi_\A) \X$. More details
on finite-sample \emph{anchor regression} can be found in Section~\ref{sec:finite-sample}. 

For $\gamma = 1$ we obtain the least squares solution, while for $\gamma > 1$
  the \emph{anchor regression} concept enforces that the projection of
  the residuals onto the linear space spanned by $A$ is small (``near
  orthogonality''); the latter is related to the framework of instrumental
  variables regression.
We will prove that the penalty term
 corresponds to the maximal change in
expected loss
under certain shift
interventions.
In particular, we show that the solutions on the regularization path are
optimizing a worst case risk under shift-interventions up to a given
strength. In addition, we show that if  \emph{anchor regression}  and ordinary least squares provide the same answer, the coefficients have a causal interpretation and are stable under certain distributional changes. More specifically, in this case the \emph{anchor regression} coefficients are equal to OLS coefficients under certain perturbed distributions.

Under instrumental variables assumptions \citep{didelez2010assumptions}, $b^{\infty} = b_{\text{causal}}$, i.e.\ one endpoint of anchor regression corresponds to the solution of equation~\eqref{eq:33b}. Our framework substantially relaxes the assumptions from the
instrumental variables (IV) setting:
in particular, we allow that the
  exogeneous anchor variables $A$ are invalid instruments, as they are
  allowed now to directly influence (i.e., being direct causes of) $Y$ or some
  hidden confounders~$H$. The price to be paid for such cases is that the causal
  parameters are not identifiable any more. However, one can still exploit
  some invariance properties and obtain robust predictions in the sense of
  distributional robustness over a class ${\cal F}$ as introduced
  before. In addition, under the assumptions of instrumental variables, one
  can identify the causal parameters as the
  procedure naturally interpolates between the solution to
  ordinary least squares and two-stage least squares. One can also abandon
  causal and structural equation models and prove that the proposed anchor
  regression procedure minimizes quantiles of a conditional mean squared
  error.

The main benefits of the proposed \emph{anchor regression} concept are
robust predictions and replicability of variable selection on test data
sets when the training data set can be grouped according to some exogeneous
categorical variable (the ``anchor'') such as
different circumstances, time-spans, experiments or experimental batches,
or when certain numerical exogeneous variables are only available on the
training, but not on the test data set. The anchor variable can either be used to encode heterogeneity ``within'' a data set or heterogeneity  ``between'' data sets. More specifically, within one data set, each level of the anchor variable encodes a homogeneous group of observations of $(X,Y)$. Alternatively, the anchor variable can be an indicator of data sets, where each data set is an homogeneous set of observations of $(X,Y)$. In principle, it is possible to develop the theory for the case where the anchor is deterministic. However, for simplicity of exposition in this paper we will model the anchor variable as random.

Our \emph{anchor regression} framework allows to quantitatively relate causality,
invariance, robustness and replicability, under weaker assumptions than what is necessarily
required to infer causal effects. Our work seems to be the first attempt to
achieve this, with a practical procedure which is easy to compute and use
in practice.

\subsection{Related work} \label{sec:relatedwork}

The considered perturbations from the class ${\cal F}$ are modeled
by interventions in an underlying structural equation model
\citep{Pearl2009}. Furthermore, as the proposed procedure interpolates
between the solution to ordinary least squares and the instrumental
variables (two-stage least squares) approach, there are obvious connections
to the IV literature, see e.g.,
\citep{Wright1928,bowden1990instrumental,didelez2010assumptions}. 
$K$-class estimators have the same algebraic form as anchor regression. 
The former are used to estimate structural parameters 
and often possess 
improved
statistical properties
compared to two-stage least squares, 
for example
\citep{theil1958economic,nagar1959bias}. In
  \citet{leamer1978least} and \citet{klepper1984consistent} the authors
  show how backwards regressions can be used to bound the regression
  coefficients for errors-in-variables models. It is similar to our work in
  the sense that the considered model class forms a convex set, a
  structure  which can be explored by modified linear
  regressions.

As mentioned above, predictive invariance in causal models has been exploited in \citet{peters2016causal} for the purpose of learning direct causal effects.
However in this work, the main goal is not to learn causal parameters, but
to obtain predictive stability under perturbations.  The goals of achieving
robustness and learning causal parameters can be
different, as shown by the example discussed in Section~\ref{sec:example}.
In a different line of work, \citet{pearl2014external} have
  developed a formal language to treat the problem of generalizability of
  causal effects across environments or populations, assuming that the causal structure is known.

There exists a plethora of work on transfer learning in the machine
learning literature, which focuses on knowledge transfer across different
domains of the data \citep{pan2010survey}.
Furthermore, there is work on distributional robustness, which explores bounded
  distributional perturbations, e.g., 
in a Wasserstein ball \citep{sinha2017certifiable} or under noise scaling \citep{heinze2017guarding}.
In \citet{rojas2015causal} and \citet{magliacane2017causal}, the authors
propose to use the best predictive model under all invariant models.
In general, these methods
do not allow for
interventions on the target variable Y
and concentrate on strong perturbations.
Unlike pre-specifying the class ${\cal F}$, we aim to learn it from
   the training data: it has then the interpretation of an estimated class
   ${\cal F}$ which is generated from a structural equation model.
\citet{Pfister2018}
show for ODE based models that
by trading off predictability and invariance
under different experimental conditions in a
similar way as \emph{anchor regression},
one may still learn models that generalize better to unseen experiments. \citet{yu2019three} expand traditional statistical uncertainty considerations by adding new notions of stability to improve reliability and reproducibility of knowledge extraction from data.

In \citet{entner2013data} the authors derive two rules that are
  sound and complete for inferring whether a given variable has a causal
  effect or not. The first rule uses (conditional) instruments to deduce
  the presence of a causal effect. While the goal of their
  work is different from the main intention of anchor regression, the first rule
  is  similar to the condition that two version of anchor
  regression agree, as explained further below. 

Furthermore, from a rather different viewpoint, it is known that
many techniques for penalized regression
can be formulated as a solution
to~\eqref{eq:optimrob}, too.
To see this, consider some measurement error $\xi$ in $X$, i.e., that $(X+\xi,Y)$ under $\Ptrain$ has the same distribution as $(X,Y)$ under $\mathbb{P}_\text{test}$. If we assume further that the measurement errors $\xi_{k}$ are centered,
jointly independent and independent of $X$ and $Y$ under $\Ptrain$,
we can write
\begin{equation*}
\mathbb{E}_{\text{test}}[(Y-X^\intercal b)^{2}]
= \Etrain[(Y-(X+\xi)^{\intercal} b)^2]
= \Etrain[(Y-X^\intercal b)^{2}] + \sum_{k=1}^{d} \Etrain[\xi_{k}^{2}] b_{k}^{2}.
\end{equation*}
If $\mathcal{F}$
contains all such test distributions with measurement errors up to strength $\mathbb{E}[\xi_{k}^{2}] \le \gamma$, the optimization~\eqref{eq:optimrob} becomes
\begin{equation*}
\min_b \max_{F \in \mathcal{F}} \mathbb{E}_F[(Y-X^\intercal b)^{2}]
=  \min_b \Etrain[(Y-X^\intercal b)^{2}] + \gamma \sum_{k} b_{k}^{2}.
\end{equation*}
In words, under certain types of measurement errors, a (weighted) ridge
penalty is optimal for prediction under perturbations. This is well
  known in the measurement errors literature, see for example \citet{fuller2009measurement}. A similar result holds for the Lasso  \citep{xu2009robust}.

\section{Population anchor regression}\label{sec:pop_anchor}
We now discuss properties of the population version of the proposed estimator~\eqref{eq:anchor}.
The overall goal is to predict the target variable $Y \in \mathbb{R}$ with the observed covariate vector $X \in \mathbb{R}^{d}$. The covariates $X$ are potentially endogeneous, $A \in \mathbb{R}^{q}$ is a so-called anchor variable which is exogenous and $H \in \mathbb{R}^{ r}$ is a vector of unobserved, or ``hidden'', random variables.  In the case of categorical anchors, dummy encoding can be used to encode the categorical values with $A \in  \mathbb{R}^{q}$.

To understand \emph{anchor regression} and its properties, it is instructive to
recognize the difference to the following well-known estimation concepts:
\begin{align}\label{eq:28}
\begin{split}
  b_{\text{PA}} &:= \argmin_{b} \Etrain[((\mathrm{Id}- \mathrm{P}_{A})(Y-X^\intercal b))^{2}] = \argmin_{b} \Etrain[((Y - \mathrm{P}_{A} Y) - (X - \mathrm{P}_{A} X)^{\intercal} b)^{2}]\\
  b_{\text{OLS}} &:= \argmin_{b} \Etrain[(Y-X^\intercal b)^{2}] \\
b_{\text{IV}} &:= \argmin_{b}  \Etrain[ ( \mathrm{P}_{A}(Y - X^\intercal b))^{2}]\\
    b^{\gamma} &:= \argmin_{b}  \Etrain[((\mathrm{Id}- \mathrm{P}_{A})(Y-X^\intercal b))^{2}] + \gamma  \Etrain[( \mathrm{P}_{A}(Y-X^\intercal b))^{2}].
\end{split}
\end{align}
Here, PA stands for ``partialling out'', also sometimes called
  ``adjusting for'', and
refers to linearly regressing out A from X and Y and considering
residuals.
The abbreviation IV refers to the two-stage-least-squares
estimation principle in
instrumental variable settings.

Due to the decomposition $\Etrain[(Y-X^\intercal b)^{2}] = \Etrain[(P_A(Y-X^\intercal b))^{2}] + \Etrain[((\mathrm{Id} - P_A)(Y - X^\intercal b))^{2}]$,
  \emph{anchor regression} coincides with ordinary
least squares for $\gamma=1$. For $\gamma = 0$, \emph{anchor regression}
coincides with $b_{\text{PA}}$ and for $\gamma \rightarrow \infty$ it
converges to $b_{\text{IV}}$, that is:
\begin{align}
  b^{0} &= b_{\text{PA}} \nonumber\\
  b^{1} &= b_{\text{OLS}} \nonumber \\
b^{\rightarrow \infty} := \lim_{\gamma \rightarrow \infty} b^{\gamma} &= b_{\text{IV}}. \label{eq:btoinfty}
\end{align}
The latter equation holds if $b_{\text{IV}}$ is uniquely defined.
Hence, \emph{anchor regression} interpolates
between $b_{\text{PA}}$ and $b_{\text{OLS}}$ for $0 \le \gamma \le 1$ and
between $b_{\text{OLS}}$ and  $b_{\text{IV}}$ for $1 \le \gamma \le
\infty$.

  Generally speaking, with \emph{anchor regression}, we aim to learn a prediction
  mechanism that is reliable across $A$ such as specific time
  periods, circumstances,
  locations or experimental batches observed in the training data set, and
  has some robustness guarantees regarding distributional shifts of
    observed and potentially also hidden variables.  The structure of $A$ crucially
    determines the robustness which we aim to achieve. For example, if we desire
  to achieve robustness across locations, then $A$ should be chosen as a
  variable that encodes location in the training data set. If the desired
  robustness is with respect to experimental batches, then $A$ should be
  chosen as a variable that describes different batches in
  the training data set. 

While our estimator is defined under general conditions,
most of our theoretical results focus on a model class that we
introduce next.

\subsection{A linear structural causal model}\label{sec:setting-notation}
We assume that the data are generated from
a linear structural equation model (SEM),
also called
a structural causal model (SCM),
\citep{Bollen1989,Pearl2009}.
Let the distribution of $(X,Y,H,A)$ under $\Ptrain$ be a solution of the SEM
\begin{equation}\label{eq:32}
\begin{pmatrix}X \\ Y \\ H \end{pmatrix} = \mathbf{B} \cdot \begin{pmatrix}X \\ Y \\ H \end{pmatrix} + \varepsilon + \mathbf{M} A,
\end{equation}
where
$\mathbf{M} \in \mathbb{R}^{(d+1+r) \times q}$ and $\mathbf{B} \in \mathbb{R}^{(d+1+r) \times (d+1+r)}$ are unknown constant matrices and the anchors $A \in \mathbb{R}^q$, the hidden variables $H \in \mathbb{R}^r$, and the noise $\varepsilon \in \mathbb{R}^{d+1+r}$ are random vectors. We will call $\mathbf{M}$ the shift matrix.
The random vectors $A$ and $\varepsilon$ are assumed to be independent. Furthermore, we assume that under $\Ptrain$, $X$ and $Y$ are centered to mean zero, that $\varepsilon$ and $A$ have finite second moments and that the components of $\varepsilon$ are independent of each other.
Equation \eqref{eq:32} is potentially cyclic
and a priori
there may exist several or no distributions that satisfy this equation. In the
following, we assume that $\mathrm{Id}-\mathbf{B}$ is invertible. This
guarantees that the distribution of $(X,Y,H,A)$ under $\Ptrain$ is well-defined in terms of
$\mathbf{B}$, $\varepsilon$, $\mathbf{M}$ and $A$ as equation~\eqref{eq:32} has
only one solution (equilibrium) satisfying
\begin{equation*}
\begin{pmatrix}X \\ Y \\ H \end{pmatrix} =  (\mathrm{Id}- \mathbf{B})^{-1} ( \varepsilon + \mathbf{M} A ).
\end{equation*}
More details on the interpretation in the cyclic case can be found in the Appendix, Section~\ref{sec:interpr-model-class}. The model induces a directed graph $G$, with the edges given by the following construction: For every $\mathbf{M}_{k,l} \neq 0$, a directed edge is drawn from $A_{l}$ to the $k$-th variable in the $(d+1+r)$-dimensional vector $(X,Y,H)$. Analogously, for every $\mathbf{B}_{k,l} \neq 0$, a directed edge is drawn from the $l$-th variable in $(X,Y,H)$ to the $k$-th variable in $(X,Y,H)$. The (vector-valued) variable $A$ is called anchor since it corresponds to a source node in the directed graph, that is, there are no incoming edges
into $A$. We allow the graph $G$ to be cyclic. Note that the matrix $\mathrm{Id}-\mathbf{B}$ is always invertible if the graph $G$ is acyclic.
An exemplary graph $G$ that lies in our model class is given below. We also allow for self-cycles (for example an arrow from $Y$ to $Y$), which are not depicted in the example.
\noindent
\begin{center}
\begin{tikzpicture}[->,>=latex,shorten >=1pt,auto,node distance=1.2cm,
                    thick]
  \tikzstyle{every state}=[draw=black,text=black, inner sep=0.4pt, minimum size=17pt]

\node[state] (Y) {$Y$};
  \node[state] (H) [ left of=Y] {$H$};
  \node[state] (X) [below of=H] {$X$};
  \node[state] (A) [above left of = H] {$A$};
\path  (X)  edge  (Y);
\draw (Y) edge [bend left] (X);
\draw  (H)   edge  [bend left]  (Y);
\draw (Y) edge [bend left] (H) ;
\draw  (H)  edge [bend left]   (X);
\draw  (X)  edge [bend left]   (H);
\draw (A) edge [bend right] (X);
\draw (A) edge (H);
\draw (A) edge [bend left] (Y);
\end{tikzpicture}
\end{center}

Note that we do not assume $A$ to be an instrument \citep{didelez2010assumptions}; we explicitly allow that $A$ directly affects $H$ and/or $Y$. This has important consequences: predictive guarantees of \emph{anchor regression} do not exclusively apply to interventions on $X$ but potentially also cover interventions on $Y$ and $H$, depending on the data
generating mechanism. More exemplary graphs and a potential motivation can be found in the following example.

\begin{example}[Three examples of graphs $G$ which are in our
  model class]
  Consider a setting with one-dimensional variables $A$, $X$ and $H$. For example, $X$ could be the activity of a certain gene, $Y$ the activity of another gene and $H$ the activity of a third, unobserved gene that regulates the activity of both $X$ and $Y$. $A \in \{-1,1\}$ could be an indicator variable of data collected from several experimental batches. The distribution of $(X,Y,H)$ may change between the different batches $A \in \{-1,1\}$. The change in distribution can be "caused" through a change in the activity of gene $X$ (graph (i)), through a change in the activity of gene $Y$ (graph (ii)) or a change in the activity of gene $H$ (graph (iii)).
Our model class contains many more graphs $G$ than these three.
   Between the variables $(X,Y,H)$ there are up to $3 \cdot 2 = 6$ directed arrows that may be in the graph (or not) and there are up to 3 arrows from $A$ to $(X,Y,H)$ that may be in the graph (or not), leading to a total of $2^3 \cdot 2^6 = 512$ directed graphs that lie in our model class for one-dimensional $A$, $X$ and $H$.
\noindent
\begin{multicols}{3}
\begin{center}
\begin{tikzpicture}[->,>=latex,shorten >=1pt,auto,node distance=1.2cm,
                    thick]
  \tikzstyle{every state}=[draw=black,text=black, inner sep=0.4pt, minimum size=17pt]

\node[state] (Y) {$Y$};
  \node[state] (H) [ left of=Y] {$H$};
  \node[state] (X) [below of=H] {$X$};
  \node[state] (A) [above left of = H] {$A$};
\path  (X)  edge  (Y);
\draw  (H)   edge    (Y);
\draw  (H)  edge   (X);
\draw (A) edge  (X);
\end{tikzpicture}
(i)
\end{center}
\columnbreak
\begin{center}
\begin{tikzpicture}[->,>=latex,shorten >=1pt,auto,node distance=1.2cm,
                    thick]
  \tikzstyle{every state}=[draw=black,text=black, inner sep=0.4pt, minimum size=17pt]
\node[state] (Y) {$Y$};
  \node[state] (H) [ left of=Y] {$H$};
  \node[state] (X) [below of=H] {$X$};
  \node[state] (A) [above left of = H] {$A$};
\path  (Y)  edge   (X);
\draw  (H)   edge    (Y);
\draw  (H)  edge   (X);
\draw (A) edge  (Y);
\end{tikzpicture}
(ii)
\end{center}
\columnbreak
\begin{center}
\begin{tikzpicture}[->,>=latex,shorten >=1pt,auto,node distance=1.2cm,
                    thick]
  \tikzstyle{every state}=[draw=black,text=black, inner sep=0.4pt, minimum size=17pt]
\node[state] (Y) {$Y$};
  \node[state] (H) [ left of=Y] {$H$};
  \node[state] (X) [below of=H] {$X$};
  \node[state] (A) [above left of = H] {$A$};
\draw  (H)   edge   (Y);
\draw  (H)  edge  (X);
\draw (A) edge   (H);
\draw (X) edge (Y);
\end{tikzpicture}
(iii)
\end{center}
\end{multicols}
\end{example}

We aim to investigate the distribution of $(X,Y,H)$ under perturbations. In
the literature, so-called {point, hard or do-interventions are often employed for causal
modelling \citep{Pearl2009}.

Here, we aim to model the perturbed distributions as small, medium and
potentially large perturbations of the training
distribution. Interventions that act on the system in a linear fashion
are often natural as well as simple to study.
   Thus, we will consider so-called  shift interventions
  on
 $(X,Y,H)$, which simply shift a variable by a value,
   see~\eqref{eq:def-perturbed-distr} below. This change
 subsequently propagates through the system. Shift interventions
 can be
 seen as
 a special case of  a ``parametric'', ``imperfect'' or
 ``dependent'' intervention or a ``mechanism change''
 \citep{Eberhardt2007,korb2004,Tian2001}.
In particular, when $A$ represents a ``dummy encoding'' of different batches, for example, we regard this as a flexible class of interventions.

 The new interventional (perturbed) distribution is denoted by $\mathbb{P}_v$. The distribution of the variables $(X,Y,H)$ under $\mathbb{P}_v$ is defined as the solution of
\begin{equation}\label{eq:def-perturbed-distr}
\begin{pmatrix}X \\ Y \\ H \end{pmatrix} = \mathbf{B} \cdot \begin{pmatrix} X \\ Y \\ H \end{pmatrix} + \varepsilon + v, \end{equation}
where $v \in \mathbb{R}^{d+1+q}$ is a random or deterministic vector independent of $\varepsilon$, but not necessarily independent of $A$.
The distribution of $\varepsilon$ is assumed to be the same under $\Ptrain$ and under $\mathbb{P}_v$.
We call $v$ a shift. 
We potentially allow for interventions on $X$, $Y$ and $H$, i.e., we allow $v_{k} \not \equiv 0$ for all $k \in \{0,\ldots,d+q+1\}$.  The main intuition behind  shift interventions is that an external force shifts a certain variable by some amount.  This shift propagates through the SEM, changing the distribution of some of the other variables.

\subsection{Anchor regression: an example}\label{sec:example}

First, we give an example of a linear SEM and the effect of a shift
intervention. Then we will discuss the performance of ordinary least
squares (OLS), the instrumental variables approach (IV) and partialling out
$A$ (PA); and motivate \emph{anchor regression}.
We compare the estimators by training them on the training distribution $\Ptrain$ and evaluating their performance on a perturbed distribution $\mathbb{P}_v$.

 Consider a classical setting for the IV approach, where $A$ is an instrument, $X$ is endogenous and $H$ is a hidden confounder.
The structural equations of the unshifted distribution are defined on the left hand side of
Example~\ref{ex:2}. The equations under a shift $v=(1.8,0,0)^{\intercal}$
are depicted on the right-hand side.
The structural equations
are assumed to be the same, but the variable $X$ is shifted by
$+1.8$ and the change propagates through the SEM.
\begin{example}\label{ex:2}
The structural equations for $\Ptrain$ can be found on the left. On the right, structural equations for $\mathbb{P}_{v}$ with $v=(1.8,0,0)$.
\begin{align*}
  &A \sim \text{Rademacher} \qquad \qquad  &&  \text{ } \\
&\varepsilon_{H},\varepsilon_{X},\varepsilon_{Y}  \stackrel{\text{indep.}}{\sim} \mathcal{N}(0,1)  \qquad \qquad \qquad &&\varepsilon_{H},\varepsilon_{X},\varepsilon_{Y} \stackrel{\text{indep.}}{\sim} \mathcal{N}(0,1) \\
&H \leftarrow \varepsilon_{H} \qquad &&H \leftarrow \varepsilon_{H}\\
&X  \leftarrow A+H+\varepsilon_{X}   &&X  \leftarrow 1.8+H+\varepsilon_{X}\\
&Y  \leftarrow X+2H+\varepsilon_{Y}  &&Y  \leftarrow X+2H+\varepsilon_{Y}
\end{align*}

\end{example}
\noindent There are two extreme cases for dealing with the variable
$A$. The variation explained by $A$ can be removed by partialling out $A$, sometimes also called residualizing with respect to $A$ or adjusting for the effect from $A$. If we think about $A$ as a subpopulation indicator variable, doing so creates a more homogeneous
population and thus can correct for population stratification.
The other extreme case is to remove all variation except for the variation explained by $A$. Under instrumental variables assumptions, doing so removes possible confounding variables and allows estimation of causal effects.  For comparison, we thus consider partialling out the anchor variable (PA), ordinary least squares (OLS) and the instrumental variables approach (IV) in the form of two-stage least squares. All three are computed on $\Ptrain$, while their performance will be compared on the perturbed distribution $\mathbb{P}_v$.

If we regress $Y$ on $X$, we obtain regression coefficient
$b_{\text{OLS}} \approx 1.66$. The IV approach yields $b_{\text{IV}}=1$ and partialling out $A$ leads to $b_{\text{PA}}= 2$. 
For each coefficient $b^{\gamma}, \gamma \in [0,\infty)$ we compute the MSE on the shifted distribution $\mathbb{E}_{v}[(Y-X^\intercal b^{\gamma})^{2}]$. The results are depicted in Figure~\ref{fig:1}. None of the three methods IV, PA and OLS yield the lowest MSE. In fact, large
sections of the path of $b^{\gamma}, \gamma \in (1,\infty)$,
outperform IV, PA and OLS. In that sense, even if IV regression
identifies the true causal parameter, \emph{anchor regression} can exhibit
better prediction properties. This is not specific to the choice $v =
(1.8,0,0)^{\intercal}$ but holds for other perturbations $v$ as well. This
will be discussed further in Section~\ref{sec:optim-pred-perf}; it turns
out that we can give optimality guarantees under certain interventions $v$,
which depend on the underlying structural equation model. Furthermore,
\emph{anchor regression} will turn out to be useful even for cases where IV
regression cannot identify the causal parameter, i.e., when the exogenous
variable $A$ is a direct cause of $Y$ or the hidden confounder $H$. In the
next section we discuss why all three approaches OLS, PA and IV
have suboptimal performance in this example on the test data.

\begin{figure}
\begin{center}
\includegraphics[scale=0.75]{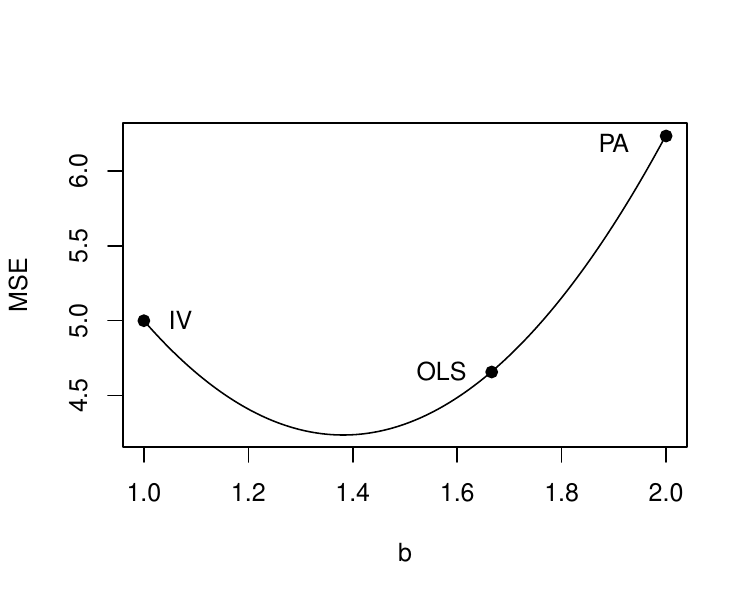}
\caption{IV, OLS, PA and \emph{anchor regression} coefficients are computed on unshifted data.  The plot shows the MSE  $\mathbb{E}_{v}[(Y-X^\intercal b)^{2}]$
on shifted variables for varying coefficients $b=b^{\gamma}$, $\gamma \in (0,\infty)$. The SEM for both shifted and unshifted data is given in Example~\ref{ex:2}. The optimal coefficient lies between IV and OLS.}\label{fig:1}
\end{center}
\end{figure}

\subsection{Trading off performance on perturbed and unperturbed data}\label{sec:trad-perf-pert}
Why did the three approaches OLS, IV and PA deliver suboptimal performance in the preceding example? Recall that the overall goal is to find $b$ such that predictive performance is not only good on the training distribution but also under perturbed distributions. In  this sense, we want to avoid ``overfitting'' to the particular distribution of the training data set. This can be investigated by considering the minimax loss
\begin{equation}\label{eq:6}
 \argmin_{b}  \sup_{v \in C} \mathbb{E}_{v}[(Y - X^\intercal b)^{2}] \text{ for a suitable set } C \subseteq \mathbb{R}^{d+q+1}.
\end{equation}
The crucial point here is to choose a ``reasonable'' set of perturbations
$C$. If $C$ is small, then the solution of equation~\eqref{eq:6} will
usually not deliver good predictive performance under perturbations. If $C$
is too large, then the solution may be unnecessarily conservative. Now let us return to the example of Section~\ref{sec:example}. It can be shown that $b_{\text{PA}}$ solves the minimax problem for $C_{\text{PA}} = \{0\}$, i.e.,
\begin{equation*}
  b_{\text{PA}} = \argmin_{b} \sup_{v \in C_{\text{PA}}} \mathbb{E}_{v}[(Y -X^\intercal b)^{2}].
\end{equation*}
Hence it is not surprising that $b_{\text{PA}}$ showed suboptimal performance under the intervention $v = (1.8,0,0)^{\intercal}$. Ordinary least squares solves the minimax problem for
$C_{\text{OLS}} = \{ v \in \mathbb{R}^{3} : v_{2} = v_{3} = 0 \text{ and } v_{1}^{2} \le \Etrain[A^{2}] \}$, i.e., \begin{equation*}
  b_{\text{OLS}} = \argmin_{b} \sup_{v \in C_{\text{OLS}}} \mathbb{E}_{v}[(Y -X^\intercal b)^{2}].
\end{equation*}
Loosely speaking, ordinary least squares optimizes the predictive performance under shifts in $X$ up to strength $v_1^2 \le \Etrain[A^2]$.
On the other hand, it can be shown that in the given example IV regression solves the minimax problem for $C_{\text{IV}} = \{ v \in \mathbb{R}^3 : v_2 = v_3 = 0\}$:
\begin{equation*}
  b_{\text{IV}} = \argmin_{b} \sup_{v \in C_{\text{IV}}} \mathbb{E}_{v}[(Y -X^\intercal b)^{2}].
\end{equation*}
In words, the causal parameter (IV) solves the minimax problem if the
supremum is taken over arbitrary strong shifts in $X$. Such shifts are not always realistic, hence from a prediction perspective the causal parameter can be unnecessarily conservative.
The vector $b_\text{PA}$ is optimized for prediction under zero perturbations $C_{\text{PA}} = \{ 0\}$ and does not exhibit stable predictive performance under shifts in $X$. As discussed earlier, ordinary least squares is somewhat in between.

 The tradeoff is depicted in Figure~\ref{fig:2}: predictive performance of
 the four methods (PA, IV, OLS and \emph{anchor regression} with $\gamma =
 5$) is shown under varying intervention strength. While the causal
 parameter (IV) is the most stable, for small and medium-sized shifts other
 methods are preferable. On the other hand, OLS and PA show good
 performance only under small perturbations, with rapidly growing MSE for
 larger perturbations. Let $C^{5} =  \{ v \in \mathbb{R}^{3}: v_{2} = v_{3} = 0 \text{ and } v_{1}^{2} \le 5 \}$. For the example it can be shown
 (cf. Theorem~\ref{thm:anchor-regression}) that \emph{anchor regression}
 for $\gamma=5$ solves the minimax problem
\begin{equation*}
  \argmin_{b} \sup_{v \in C^{5}} \mathbb{E}_{v}[(Y - X^\intercal b)^{2}].
\end{equation*}
This gives us a convenient interpretation of $b^{\gamma}$ for $\gamma=5$: it minimizes the risk under shift interventions on $X$ up to strength $|v_{1}| \le \sqrt{5}$.
The next section discusses the optimality of \emph{anchor regression} under perturbations up to a given strength beyond the specific SEM of Example~\ref{ex:2}.
\begin{figure}
\begin{center}
\includegraphics[scale=0.5]{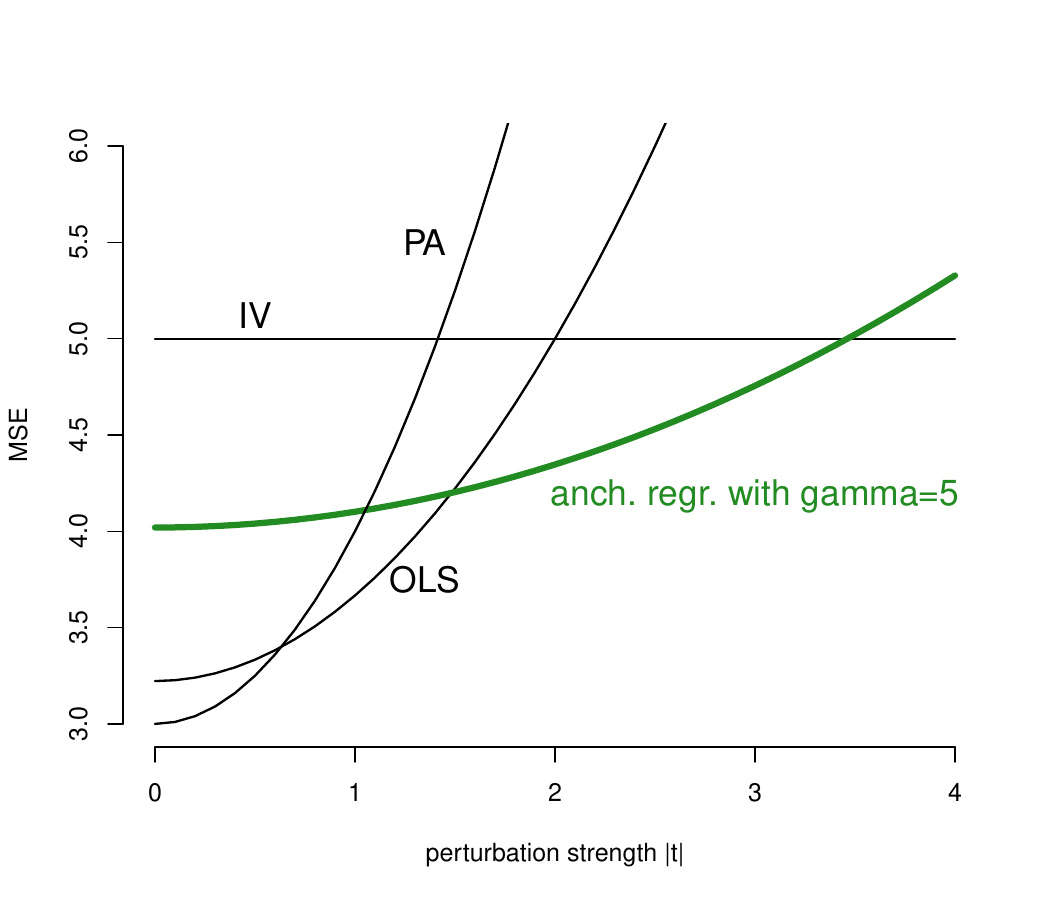}
\caption{Predictive performance of the direct causal effect (IV), AP, OLS and \emph{anchor regression} with $\gamma=5$ under varying interventions on $X$. The SEM is taken from Example~\ref{ex:2}. The MSE $\mathbb{E}_{v}[(Y - X^\intercal b)^{2}]$ is depicted under perturbation strength $v = (t,0,0)^{\intercal}$. The causal parameter (IV) exhibits constant predictive performance under arbitrary perturbation strength $|t|$, but predictive performance under small perturbations is subpar. PA and OLS have very good performance under small interventions but performance suffers under larger interventions. \emph{Anchor regression} with $\gamma=5$ trades performance on unperturbed data ($t=0$) for more stability, i.e., better performance on medium-sized interventions. In particular, it is minimax optimal under shifts $C^{5} =  \{(t,0,0)^{\intercal} : |t| \le \sqrt{5} \approx 2.24 \}$, cf. Theorem~\ref{thm:anchor-regression}. For large shifts $|t|$ the IV method eventually outperforms \emph{anchor regression}. Note that all shown solutions are anchor solutions, under respective penalties $\gamma=0$ (PA), $\gamma=1$ (OLS), $\gamma=5$ and $\gamma=\infty$ (IV).
}
\label{fig:2}
\end{center}
\end{figure}

\subsection{Optimal predictive performance under perturbations}\label{sec:optim-pred-perf}
In this section we will discuss a first main result, namely a
  fundamental connection between the population version of
  \emph{anchor regression} and the worst case
  risk over a class of shift interventions.
In Section~\ref{sec:example} we saw that neither PA, OLS nor IV are optimal for prediction under the given intervention strength.
The following theorem gives guarantees for the prediction error of
\emph{anchor regression} under shift interventions up to a given
perturbation strength.  
Recall that $P_A$ denote the $L_2$-projection on the linear span from the components of $A$.
Under the assumptions of Section~\ref{sec:setting-notation}, we have  $\mathrm{P}_{A}(X)  = \Etrain[ X |A]$ and $\mathrm{P}_{A}(Y) = \Etrain[Y|A]$.
Let $X$ and $Y$ have mean zero. \newpage
\begin{theorem}\label{thm:anchor-regression}
 Let the assumptions of Section~\ref{sec:setting-notation} hold. For any $b \in \mathbb{R}^{d}$
we have
\begin{equation}\label{eq:27}
 \Etrain[((\mathrm{Id}- \mathrm{P}_{A})(Y-X^\intercal b))^{2}] + \gamma  \Etrain[( \mathrm{P}_{A}(Y-X^\intercal b))^{2}]  = \sup_{v \in C^{\gamma}} \mathbb{E}_{v}[ (Y-X^\intercal b)^{2}  ],
\end{equation}
where
\begin{equation*}
  C^{\gamma} :=\{ v \in  \mathbb{R}^{d+q+1} \text{ such that } v v^{\intercal} \preceq \gamma \mathbf{M} \Etrain[A A^{\intercal}] \mathbf{M}^{\intercal}\}.
\end{equation*}
and $\mathbf{M}$ is the shift matrix, cf.\ equation~\eqref{eq:32}. A
formulation of the result where $v$ is allowed to be random can be found in
the Appendix, Section~\ref{sec:theor-refthm:-regr}. \end{theorem}
Here, for two positive semidefinite matrices $A$ and $B$ we write $A\preceq B$ if and only if $B-A$ is positive semidefinite.
In particular, we have
$C^{\gamma} \subseteq \text{span}(\mathbf{M})$.
Readers familiar with the concept of interventions
may thus think about $\mathbb{P}_{v}$ as the distribution under a point intervention on $A$, where the condition $v \in C^{\gamma}$ restricts the set of interventions to a certain strength.

There are two important takeaways from this theorem:
First, the squared $L_2$-risk under certain worst-case shift interventions is equal to adding a
penalty to the squared $L_2$-risk.

Second, as population \emph{anchor
  regression} optimizes the penalized criterion (on the left-hand side of
equation~\eqref{eq:27}), \emph{anchor regression} minimizes the worst-case
MSE under shift interventions  up to a given strength in certain
directions, cf. equation~\eqref{eq:28}.
We have discussed in
Section~\ref{sec:trad-perf-pert} why it can be desirable to consider
interventions only up to a given strength. In the following we want to
briefly discuss the direction of the shift interventions in $C^{\gamma}$.
  To this end, note that
\begin{equation*}
 \text{span}(\mathbf{M}) = \lim_{\gamma \rightarrow \infty} C^{\gamma}.
\end{equation*}
Here, for ease of interpretation we made the assumption that $\Etrain[A A^{\intercal}]$ is positive definite.
We explicitly allow $A$ to have a direct effect on $X$, $Y$ or $H$. In other words, in the shift matrix $\mathbf{M}$, we allow $\mathbf{M}_{k \bullet} \not \equiv 0$ for some (or all)  $k \in \{1,\ldots,d+r+1\}$. Hence $C^{\gamma}$ potentially contains interventions that affect not only $X$ but also  $Y$ or $H$. We discuss this in more detail in Section~\ref{sec:three-examples} in the Appendix.

Generally speaking, we have introduced a penalty that encourages good
predictive performance under distributional shifts. Penalties of the form
$\gamma \| b \|_{2}^{2}$ or $ \gamma \| b \|_{1}$ are widely employed
for finite sample regression to prevent overfitting the data with
  estimated parameters. Here, we deal with a different
type of ``overfitting'' that may even affect the population version.
  For $\gamma=0$ the population estimator will
``overfit'' to the particular distribution $\Ptrain$, in the sense that it is not guaranteed to work well under shifted distributions $\mathbb{P}_{v}$.
For $\gamma >0$ we obtain predictive guarantees for both, shifted and
unshifted data. As $\gamma \rightarrow \infty$, population \emph{anchor
  regression} works
  increasingly well
  under strong interventions, at the
price of deteriorating MSE on unshifted or moderately shifted data. In the
finite sample case, additional regularization in form of an $\ell_{1}$-penalty can be advisable. This is discussed in Section~\ref{sec:high-dimens-estim}.

\subsection{Limitations of using direct causal effects for prediction}

In Section~\ref{sec:example} we saw that using causal effects for prediction is in general not recommended if the perturbation strength is relatively small. In this section, we show that a similar caveat holds for the directions of the perturbations. Using direct (or total) causal effects in settings with perturbations on $Y$ and $H$ can be ill-advised, even if the perturbation strength is arbitrarily strong. Using direct causal effects for prediction does not protect against arbitrary perturbations.

As an example, consider the following structural equation model and a shift in the distribution of the hidden confounder $H$. On the left, the structural equation for the unperturbed distribution $\Ptrain$ is defined. On the right, the data generating mechanism for the perturbed distribution $\mathbb{P}_{v}$ is given under a shift $v=(0,0,t)^{\intercal}$, $ t \in \mathbb{R}$. \newpage \begin{multicols}{2}
\noindent
\begin{align*}
\begin{split}
A &\sim \text{Rademacher} \\
\varepsilon_{H},\varepsilon_{X},\varepsilon_{Y} &\stackrel{\text{indep.}}{\sim} \mathcal{N}(0,1) \\
H &\leftarrow A+\varepsilon_{H} \\
X & \leftarrow H+\varepsilon_{X} \\
Y & \leftarrow { 1} \cdot X+2H+\varepsilon_{Y}
\end{split}
\end{align*}
\columnbreak
\begin{align}\label{eq:26}
\begin{split}
\text{ } \\ \varepsilon_{H},\varepsilon_{X},\varepsilon_{Y} &\stackrel{\text{indep.}}{\sim} \mathcal{N}(0,1) \\
H &\leftarrow  { t }  + \varepsilon_{H}  \\
X & \leftarrow H+\varepsilon_{X} \\
Y & \leftarrow { 1} \cdot X+2H+\varepsilon_{Y}
\end{split}
\end{align}
\end{multicols}
\noindent Assume that through some oracle (or previous experiments) we know
that the direct causal effect from $X$ to $Y$  \citep[][page
127]{Pearl2009} is $1$, that is it equals the coefficient for $X$ in the
  structural equation for $Y$.
\emph{Anchor regression} is trained on data from the SEM on the left; the
predictive performance of \emph{anchor regression} and the direct causal
effect are compared on the shifted distribution $\mathbb{E}_{v}[(Y -
X^\intercal b)^{2}]$. The results are shown in Figure~\ref{fig:causalbad}. The direct causal effect is uniformly outperformed by PA, OLS and \emph{anchor regression} with $\gamma =
5$. Roughly speaking, this is due to the fact that the direct causal effect
is geared towards prediction under interventions on $X$, as discussed in
Section~\ref{sec:trad-perf-pert}. Interventions on $H$ induce a very
different distributional shift. Comparing PA and \emph{anchor regression}
leads to a similar conclusion as in Figure~\ref{fig:2}. Under small
perturbations, PA and OLS are slightly better than \emph{anchor regression}. However,
\emph{anchor regression} exhibits a stable performance across a large
  range of perturbation strengths and outperforms the other methods for
  medium or strong perturbations.

\begin{figure}
\begin{center}
 \includegraphics[scale=0.75]{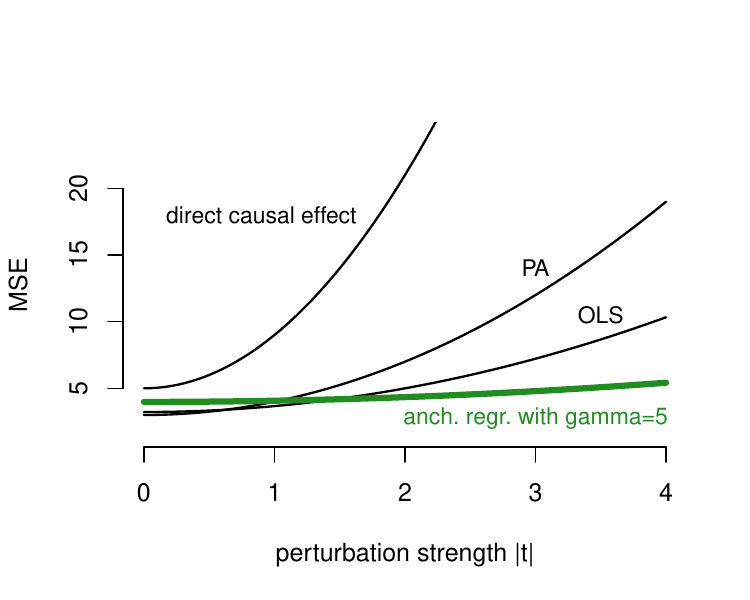}
\end{center}
\caption{Predictive performance of the direct causal effect, PA, OLS and \emph{anchor regression} under varying interventions on $H$. The MSE $\mathbb{E}_{v}[(Y-X^\intercal b)^{2}]$ is depicted under varying perturbations $v=(0,0,t)^{\intercal}$. The corresponding structural equation models are given in equation~\eqref{eq:26}. For small perturbations, PA and OLS perform  better than \emph{anchor regression}. The direct causal effect exhibits large MSE for all values of $t$. While the direct causal effect shows stable predictive performance under interventions on $X$ (as discussed in Section~\ref{sec:trad-perf-pert}), this is at the expense of predictive stability under interventions on $H$ or $Y$. The MSE of \emph{anchor regression} with $\gamma=5$ slowly grows in $|t|$.} \label{fig:causalbad}
\end{figure}

\subsection{Interpretation of anchor regression via quantiles}\label{sec:interpr-anch-regr-quant}
We now provide an interpetation of \emph{anchor regression} without using
structural equation models. For reasons of simplicity, we present the
result for continuous anchors. A similar result for discrete anchors can be
found in the Appendix, Section~\ref{sec:lemma-refl-anch}. For the result
of this section, the assumptions mentioned in
Section~\ref{sec:setting-notation} are not necessary, but instead we assume
multivariate Gaussianity of $(X,Y,A)$, see Lemma
\ref{lemma:interpr-anch-regr}. Define $Q(\alpha)$ as the $\alpha$-th
quantile of $\mathbb{E}[(Y-X^\intercal b)^{2} |A]$. Recall that with the
notation defined in Section~\ref{sec:our-contribution} if $(X,Y,A)$ is
multivariate Gaussian we have $(\mathrm{Id}-\mathrm{P}_{A})(Y-X^\intercal
b) = Y - X^\intercal b - \mathbb{E}[Y-X^\intercal b |A]$ and
$\mathrm{P}_{A}(Y-X^\intercal b) = \mathbb{E}[Y-X^\intercal b |A]$.  
\begin{lemma}\label{lemma:interpr-anch-regr}
  Assume that the variables $(X,Y,A)$ follow a centered multivariate normal distribution under $\mathbb{P}$. Then, for $0 \le \alpha \le 1$,
\begin{equation*}
  Q(\alpha) = \mathbb{E}[((\mathrm{Id}- \mathrm{P}_{A})(Y-X^\intercal b))^{2}] + \gamma  \mathbb{E}[( \mathrm{P}_{A}(Y-X^\intercal b))^{2}],
\end{equation*}
where $\gamma$
equals the $\alpha$-th quantile of a $\chi^{2}$-distributed random variable with one degree of freedom.
\end{lemma}
Note that the right-hand side of the equation in
Lemma~\ref{lemma:interpr-anch-regr} is the objective function of
\emph{anchor regression}. Thus, this shows that \emph{anchor regression} can be used to optimize quantiles of $\mathbb{E}[(Y-X^\intercal b)^{2}|A]$, for example  minimization of the $95 \%$-quantile of $\mathbb{E}[(Y-X^\intercal b)^{2}|A]$ is achieved by $b^{ \gamma}$ with $ \gamma = \chi_{1}^{2}(0.95)$. In spirit, this result is similar to Theorem~\ref{thm:anchor-regression}. The perturbed distributions $\mathbb{P}_{v}$ in Theorem~\ref{thm:anchor-regression} play a similar role as the conditional distributions $\mathbb{P}[\bullet|A=a]$ in Lemma~\ref{lemma:interpr-anch-regr}. For increasing $\gamma$, the predictions are increasingly reliable across distributions $\mathbb{P}[\bullet|A=a]$.

  \section{Replicability and Anchor Stability}

We consider here the question of replicability when estimation is done a
second time on a new perturbed dataset which has different data generating
distributions than the original unperturbed but typically heterogeneous
data. Replicability in this context is about potential differences in the
regression parameters or prediction losses under different distributions: it
is a ``first order'' problem instead of inferential
statements about statistical uncertainties due to finite samples.

For the following two sections, we sometimes need a condition that the loss of anchor regression remains finite for $\gamma\rightarrow \infty$. We say the  \emph{projectability condition} is fulfilled if
\begin{equation} \label{def:projectability} \text{rank}(\text{Cov}_{\text{train}}(A,X))= \text{rank}(\text{Cov}_{\text{train}}(A,X)| \text{Cov}_{\text{train}}(A,Y)),
\end{equation}
where $\text{Cov}_{\text{train}}(A,X)| \text{Cov}_{\text{train}}(A,Y)$ is a
$q\times (d+1)$ matrix, consisting of the $q\times d$ covariance matrix
$\text{Cov}_{\text{train}}(A,X)$, extended by the $q\times 1$ vector
$\text{Cov}_{\text{train}}(A,Y)$. The reason why we call this the
  ``projectability condition'' becomes clear in Lemma~\ref{lemma:project}
  below.

The projectability condition~\eqref{def:projectability} is fulfilled, for
example, if $\text{Cov}_{\text{train}}(A,X)$ is of full rank and
$q \leq d$  (sometimes called the under- or
just-identified case as the dimension of $A$ is less or equal to the
dimension of $X$). The condition can also be fulfilled for $q>d$ under
additional constraints on the 
nature of the link $A \to Y$. In general, the projectability
  condition allows that the anchor variables $A$ directly influence also $Y$
  or $H$, and the example above for $q \le d$ requires only a full rank
  condition on $\text{Cov}_{\text{train}}(A,X)$.
\begin{lemma}\label{lemma:project}
  Assume that $\mathbb{E}_{\text{train}}[ A A^{\intercal}]$ is invertible.

The projectability condition~\eqref{def:projectability} is fulfilled if and only if
\begin{equation}\label{eq:reachzero}
  \min_b \Etrain[( \mathrm{P}_{A}(Y-X^\intercal b))^{2}] = 0.
\end{equation}
\end{lemma}

The projectability assumption is testable in practice.
The following results cover predictive stability and replicability under
perturbations.

\subsection{Replicability of the parameter $b^{\rightarrow \infty}$}\label{sec:repl-param-bright}

Our first goal is to investigate the replicability of the parameter
$b^{\rightarrow \infty}$. As stated in Theorem~\ref{thm:anchor-regression},
this parameter vector is protecting against certain worst case shift
perturbations of arbitrary strength and as such, it has an interesting
interpretation; in analogy to causality which corresponds to worst case
risk optimization for a different class of perturbations of arbitrary
strength, see \eqref{eq:33b}.

We consider two different data-generating distributions, and for notational
coherence with before we denote them by ``train'' and ``test''. The
training data is generated according to \eqref{eq:32} and
\eqref{eq:def-perturbed-distr}
\begin{equation}\label{mod-add1}
\begin{pmatrix}X \\ Y \\ H \end{pmatrix} = \mathbf{B}
  \cdot \begin{pmatrix} X \\ Y \\ H \end{pmatrix} + \varepsilon + v,\\
  v = \mathbf{M} \delta,\ \delta = \kappa A + \xi,
\end{equation}
where $\xi$ is a random vector with mean zero and independent of
$\varepsilon$ and $A$ and $\kappa \neq 0$. Note that with $\kappa = 1$ and $\xi = 0$ we have
the model in \eqref{eq:32}.

The test data is from the following model:
\begin{equation}\label{mod-add2}
\begin{pmatrix}X' \\ Y' \\ H' \end{pmatrix} = \mathbf{B}
  \cdot \begin{pmatrix} X' \\ Y' \\ H' \end{pmatrix} + \varepsilon' + v',\\
  v' = \mathbf{M} \delta',\ \delta' = \kappa' A' + \xi',
\end{equation}
where $\xi'$ is a random vector with mean zero and independent of
$\varepsilon'$ and $A'$ and $\kappa' \neq 0$. We note that $v'$ and $A'$ can have arbitrarily
different distributions than $v$ and $A$ but we assume that the
dimensionalities are the same. The parameters $\mathbf{B}$ and $\mathbf{M}$ are the same in
both models \eqref{mod-add1} and \eqref{mod-add2} and we assume that
\begin{equation}\label{cov-add}
  \text{Cov}_{\text{test}}(\varepsilon') = L
  \text{Cov}_{\text{train}}(\varepsilon)\ \mbox{for some $L > 0$},\
  \mathbb{E}_{\text{test}}[\varepsilon'] = \Etrain[\varepsilon] = 0.
\end{equation}
Roughly speaking, the models in the training and test dataset differ by
arbitrary shifts in $\text{span}(\mathbf{M})$ and a scalar factor in the noise
distribution.

Consider the parameter $b^{\rightarrow \infty}$ as defined in
\eqref{eq:btoinfty},
\begin{eqnarray*}
  & &b^{\rightarrow \infty} = \argmin_{b \in I} \Etrain[(Y - X^\intercal b)^2],\\
  & &I = \{b; \Etrain[Y - X^\intercal b|A] \equiv 0\},
\end{eqnarray*}
which is
a functional of the distribution in model \eqref{mod-add1}. For its
analogue on a
new test dataset with observed variables $A', X',Y'$ we define
\begin{eqnarray*}
  & &b'^{\rightarrow \infty} = \argmin_{b \in I'}
      \mathbb{E}_{\text{test}}[(Y' - (X')^{\intercal}b)^2],\\
  & &I' = \{b;\ \mathbb{E}_{\text{test}}[Y' - (X')^{\intercal}b|A'] \equiv 0\}.
\end{eqnarray*}
\begin{theorem}[Replicability of $b^{\rightarrow \infty}$]\label{thm:repbinfty}
Consider the models in \eqref{mod-add1} and \eqref{mod-add2} for the
training and test data, respectively. Assume \eqref{cov-add} and
$\Etrain[A A^{\intercal}]$ and $\mathbb{E}_{\text{test}}[A' (A')^{\intercal}]$ are
invertible  and assume that the projectability
condition~\eqref{def:projectability} holds.

 Then,
 \begin{eqnarray*}
 b'^{\rightarrow \infty} = b^{\rightarrow \infty}.
 \end{eqnarray*}
\end{theorem}
Replicability of statistical estimands is arguably a desirable property, but it is a separate question whether $b^{\infty}$ is a meaningful quantity. As discussed at the beginning of this section, $b^{\rightarrow \infty}$ has an interpretation as a coefficient vector that optimizes a certain worst-case risk. 
Beyond this interpretation, we believe that the role of $A$ matters to determine whether the components of $b^{\rightarrow \infty}$ are scientifically relevant.
Loosely speaking, in instrumental variables settings, $A$ induces
associations between $X$ and $Y$ that are due to the causal pathway between
$X$ and $Y$. Hence, $b^{\rightarrow \infty}$ has a scientific
interpretation as the causal effect from $X$ to $Y$. However, if $A$ plays
the role of a confounder (a variable that induces spurious associations
between $X$ and $Y$), then it is common practice to adjust for $A$, leading
to $b^{0}$. Under slightly weaker assumptions than in the result above we also get replicability of $b^{0}$. 
In practice, there may be 
 some uncertainty about whether $A$ is an instrument or a confounder, or whether both sets of assumptions are violated. In the next section we will show that  anchor regression can be used in such settings to screen for replicable coefficients that have a causal interpretation.

\subsection{Anchor stability}\label{sec:anchor-stability}

If all solutions of \emph{anchor regression} agree (i.e., if $b^0 = b^\gamma$ for all $\gamma \in [0,\infty)$) we call the coefficient vector \emph{anchor stable}.

We will show that under anchor stability we have predictive stability and
replicability of variable selection under
certain perturbations. Additionally, we will show that  \emph{anchor
  stability} allows a causal interpretation of the coefficient vector under
otherwise comparatively weak assumptions.
As in the previous section, in the following we assume that the limit $b^{\rightarrow \infty} := \lim_{\gamma  \rightarrow \infty} b^{\gamma}$ exists.

One of the anchor stability results (Theorem~\ref{thm:stabcaus}) can be generalized to cases where the anchor is endogeneous.
This relaxation is relevant for our application in Section~\ref{sec:genotype-tiss-expr}. A rigorous treatment of endogeneous anchors warrants the introduction of a class of models that subsumes acyclic models in Section~\ref{sec:setting-notation}. Thus, for reasons of readability we defer the most general version of the theorem to the Appendix, Section~\ref{sec:gener-vers-theor}.

  Our first result shows that we have anchor stability if the two endpoints of anchor regression agree.
\begin{proposition}\label{prop:constant}
  If $b^0 = b^{\rightarrow \infty}$ then
  \begin{equation*}
    b^0 = b^\gamma \text{ for all } \gamma \in (0,\infty).
  \end{equation*}
\end{proposition}

The proposition is valid without necessarily assuming the projectability condition, which is, however, needed for the following result on anchor stability in the case that the solutions match for $\gamma\in \{0,\infty\}$.

\begin{theorem}[Anchor stability, predictive stability and replicability]\label{thm:anchor-stability-pred-stab}

  Let the assumptions of Section~\ref{sec:setting-notation} hold,
and in addition assume the projectability condition~\eqref{def:projectability}
 and that the Gram matrix $\Etrain [A A^\intercal]$ is invertible.
If $b^0 = b^{\rightarrow \infty}$, then, for all random or constant vectors $v$ that are
uncorrelated of $\varepsilon$ and take values in $\text{span}(\mathbf{M})$,
\begin{enumerate}
  \item $\Etrain[(Y-X^\intercal b^0)^2] = \mathbb{E}_{v}[(Y - X^\intercal b^0)^2]$, \text{ and}
\item $b^0 = \argmin_b \mathbb{E}_{v}[(Y - X^\intercal b)^2]$. 
\end{enumerate}
\end{theorem}
Part (a) of the theorem implies that the risk is constant as long as the perturbations
$v$ lie in the span of the shift matrix $\mathbf{M}$, i.e.\ in $\text{span}(\mathbf{M})$. This can be seen as a form of
predictive stability
across a range of distributions. Part (b) together with Proposition~\ref{prop:constant} imply that running a
regression on perturbed data sets in the
population case returns the same coefficients as the ones computed on the
training data as long as the perturbations $v$ lie in
$\text{span}(\mathbf{M})$.
In this sense, we have replicability
across certain distributions. 

Now let us turn to the interpretation of the
individual coefficients in this case.
The individual coefficients can be interpreted using the concepts
of d-separation, causal directed acyclic graphs and
do-interventions. For reasons of readability and as the concepts are
otherwise not needed in this paper, we will not define them here
but rather refer the reader to e.g.\ \citet{Pearl2009}, Chapter~1. An
interpretation of the result in the one-dimensional case is given in Section~\ref{sec:anchor-stability-one}.
The faithfulness assumption
\citep{Spirtes2000, Pearl2009}
connects d-separation statements to statements of conditional independences.
As \emph{anchor regression} only deals with covariances, we have to make an assumption that connects d-separation statements to partial correlations.
We assume that $G$ is acyclic and that for every disjoint sets of variables $V_1, V_2, V_3 \subset (X,Y,H,A)$, $V_1$ is d-separated of $V_2$ in $G$ given $V_3$ if and only if the partial correlation $\text{part.cor}(V_1,V_2|V_3) = 0$. This can be seen as a linear version of faithfulness.
\begin{theorem}[Anchor stability implies causality] \label{thm:stabcaus}
    Let the assumptions of Section~\ref{sec:setting-notation} hold with an acyclic graph $G$,
    and assume the projectability
    condition~\eqref{def:projectability}.

 Furthermore, assume that for every disjoint sets of variables $V_1, V_2, V_3
      \subset (X,Y,H,A)$,  $V_1$ is d-separated of $V_2$ in $G$
      given $V_3$ if and only if the partial correlation
      $\text{part.cor}(V_1,V_2|V_3) = 0$. Furthermore assume that for each
  $X_k$ there exists $k'$ such that $A_{k'} \rightarrow X_k$.
    If $b^{\rightarrow \infty}
  = b^0$, then
    \begin{equation}
      b^{\rightarrow \infty} = b^0 = \partial_x \mathbb{E}[Y|do(X=x)],
    \end{equation}
    where the do-operator $\mathbb{E}[\bullet |do(X=x)]$ is defined as in \citet{Pearl2009}, Chapter~1. In addition, there is no confounder between $X$ and $Y$, i.e., there is no $H_{k}$ that is both an ancestor of some $X_{k'}$ and $Y$ in $G$.
  \end{theorem}
A more general version of this result that allows for endogeneous anchors can be found in Section~\ref{sec:gener-vers-theor}. Roughly speaking, the theorem says that under anchor stability, the coefficients $b^{\rightarrow \infty} = b^0$ have a causal interpretation and there is no confounder between $X$ and $Y$.  If confounders were present between $X$ and $Y$, intervening (or conditioning) on them could potentially change the anchor regression coefficient $b^{0}$. In this sense, the absence of confounding between $X$ and $Y$ may be seen as a positive indication for distributional replicability.    

Anchor stability is testable on data and if it holds, under relatively weak assumptions, the coefficients allow for a causal interpretation. In empirical studies using instrumental variables, one often compares IV estimates with OLS estimates. The above result formalizes the implications when these estimates are equal.

\subsection{Anchor stability in the one-dimensional case}\label{sec:anchor-stability-one}
In the special case where $X$, $Y$, $H$ and $A$ are all one-dimensional random variables, the theorem
 can be interpreted in the following way: Suppose we know that $A$ is exogeneous and $A \rightarrow X$ but we do not know whether it is a valid instrument, i.e., potentially we have $A \rightarrow Y$ or $A \rightarrow H \rightarrow Y$. We may not know either
whether we could obtain the causal coefficients by simply regressing $Y$ on
$X$ or $Y$ on $(X,A)$, i.e.,
  we are unsure
whether there exists a hidden confounder $H$ with $X \leftarrow H
\rightarrow Y$. Under the assumptions of Theorem~\ref{thm:stabcaus} and if $b^{0} \neq 0$, the models agree if and only if $A \rightarrow X \rightarrow Y$ and if no other arrows (or confounders) are present. Using the theorem, if the two anchor solutions agree, then
both the IV and regression adjustment are correct for estimating the causal effect. This approach is restrictive, but
can potentially be useful in cases
where we have little knowledge about the underlying structure and not much
reason to prefer one of these models over the other. An application of this
approach is shown in the data section.
We anticipate that the concept of \emph{anchor stability} is most
useful for screening causal effects in large-scale settings. An analogous statement holds for the multivariate case.

\section{Properties of anchor regression estimators}\label{sec:finite-sample}

In this section we discuss the properties of finite-sample \emph{anchor regression}. Section~\ref{sec:low-dimensional-case} treats the low-dimensional case; the high-dimensional case is discussed in Section~\ref{sec:high-dimens-estim}.
In the following we assume to have $n$ i.i.d. observations of $(X,Y,A)$. Concatenating the observations of $X$ row-wise forms an $n \times d$-dimensional matrix that we denote by $\X$. Analogously, the matrix containing the observations of $A$ is denoted by  $\A \in \mathbb{R}^{n \times q}$ and the vector containing the observations of $Y$ is denoted by $\Y \in \mathbb{R}^{n}$. In the following, we tacitly assume that the population parameter $b^{\gamma}$ as defined in equation~\eqref{eq:28} is unique.

\subsection{Estimator in the
    low-dimensional setting}\label{sec:low-dimensional-case}

As discussed before, in the low-dimensional case where $d < n$ we recommend using a simple plug-in estimator for the anchor-regression coefficient $b^{\gamma}$:
\begin{equation}\label{eq:29}
   \hat b^{ \gamma} = \argmin_{b} \|(\mathrm{Id}-\Pi_{\A})( \Y- \X b) \|_{2}^{2} +  \gamma \|  \Pi_{\A} ( \Y- \X b)\|_{2}^{2},
\end{equation}
where $\Pi_{\A} \in \mathbb{R}^{n \times n}$ is the matrix that
projects on the column space of $\A$, i.e., if $  \A^{\intercal} \A$
is invertible, then $\Pi_{\A} := \A ( \A^{\intercal} \A)^{-1}
\A^{\intercal}$. In Section~\ref{sec:setting-notation} we made the
assumption that $X$ and $Y$ have mean zero. Hence, in practice, we
recommend to center $\X$ and $\Y$ in a pre-processing step.

 Computation of the \emph{anchor
     regression} estimator in \eqref{eq:29} is simple, as it can be cast
 as an ordinary least squares problem on a transformed data set. To this
 end, define
\begin{equation}\label{eq:14}
  \tilde \X :=(\mathrm{Id} - \Pi_{\A}) \X + \sqrt{ \gamma}  \Pi_{\A}  \X  \quad \text{ and } \quad \tilde \Y :=   (\mathrm{Id} - \Pi_{\A}) \Y + \sqrt{ \gamma}  \Pi_{\A}  \Y.
\end{equation}
The estimator in \eqref{eq:29} can then be represented as
  follows:
\begin{equation*}
\hat{b}^\gamma = \argmin_{b} \| \tilde \Y - \tilde \X b \|_{2}^{2}.
\end{equation*}
The transformed data set $(\tilde \X, \tilde \Y)$ can be interpreted as artificially generated interventional (``perturbed'') data. In this sense, \emph{anchor regression} can be seen as a two-step procedure. First, generate perturbed data $(\tilde \X, \tilde \Y)$ for a given perturbation strength $\gamma$. Then, run ordinary least squares on the artificial data set.

By the law of large numbers for $n \rightarrow \infty$ the empirical
covariance matrix of $(X,Y,A)$ converges to the population covariance
matrix of $(X,Y,A)$. By continuity, $\hat b^{\gamma} = (\tilde
\X^{\intercal} \tilde \X)^{-1} \tilde \X^{\intercal} \tilde \Y$ converges
to the population parameter $b^{\gamma}$. Hence, $\hat b^{\gamma}$ is a
consistent estimator of $b^{\gamma}$.

The transformation \eqref{eq:14} is for computational reasons only.

Even if $(X,Y,A)$ follows a multivariate Gaussian distribution, in
  general it might \emph{not} be true that $\hat b^{\gamma} \sim
  \mathcal{N}(b^{\gamma}, V)$ for some covariance matrix $V$ since possible
  confounding complicates the matter. Hence $p$-values or confidence
  intervals from ordinary least squares regression of the transformed data
  $(\tilde{\Y},\tilde{\X})$ cannot be
  used.

Since a  main goal in this paper is to establish good predictive performance on future data
sets, it is less important to provide distributional results for  $\hat
b^{\gamma} - b^{\gamma}$, than to quantify
the excess predictive risk on new data sets. A finite sample bound for the excess risk, even covering the
  high-dimensional setting, can be found in Section~\ref{sec:finite-sample-bound}.

\subsection{Estimator in the
    high-dimensional setting}\label{sec:high-dimens-estim}
If the number of predictors $d$ exceeds the number of observations $n$, then the sample estimate defined in~\eqref{eq:14} is not well-defined. In high-dimensional settings, one typically employs $\ell_1$- or
  $\ell_2$-norm penalties
for regularization and shrinkage. The $\ell_{1}$-penalized estimators are
usually consistent under appropriate sparsity and distributional
assumptions, see for example \citet{buhlmann2011statistics}.

While high-dimensionality is
  allowed in terms of $d \gg n$, we will assume here that the number of
  anchor variables $q$ is of smaller order than $n$. High-dimensionality in
  terms of $q \gg n$ would be another issue, as $\Pi_A$ is ill-posed, and
  should be addressed with an $\ell_{\infty}$ regularization scheme,
  replacing the $\ell_2$-norm term $\gamma \|\Pi_A(Y - X^\intercal b)\|_2^2$.
We propose high-dimensional estimation of \emph{anchor regression} as a solution of
\begin{equation}\label{eq:12}
  \hat b^{\gamma,\lambda} =\argmin_{b} \|(\mathrm{Id}-\Pi_{\A})( \Y- \X b) \|_{2}^{2} +  \gamma \|  \Pi_{\A} ( \Y- \X b)\|_{2}^{2} + 2\lambda \| b \|_{1}.
\end{equation}

Compared to unregularized \emph{anchor regression}, the penalty term
$2\lambda \| b \|_{1}$ favours coefficient vectors $b$ that are sparse. For
$\gamma=1$, the estimator coincides with the Lasso
\citep{tibshirani1996regression}, whereas for $\lambda=0$, the estimator
 coincides with unregularized \emph{anchor regression}.

As in the low-dimensional case with the linear transformation in
  \eqref{eq:14}, computation of regularized \emph{anchor regression}
is easy. We can rewrite regularized \emph{anchor regression} as
\begin{align*}
 &\argmin_{b} \|(\mathrm{Id}-\Pi_{\A})( \Y- \X b) \|_{2}^{2} +  \gamma \|  \Pi_{\A} ( \Y- \X b)\|_{2}^{2} +2 \lambda \| b \|_{1} \\
=& \argmin_{b} \| \tilde \Y - \tilde \X b \|_{2}^{2} + 2 \lambda \| b \|_{1},
\end{align*}
where $\tilde \Y$ and $\tilde \X$ are defined as in
equation~\eqref{eq:14}. Hence, solving a high-dimensional
\emph{anchor regression} for fixed $\gamma$ is reduced to solving a Lasso
problem. This is typically done by coordinatewise descent
\citep{friedman2007pathwise} to approximately compute the solution path. In the next section we will investigate finite-sample performance of $\ell_1$-norm regularized \emph{anchor regression}.

\subsection{Finite-sample bound for discrete anchors}\label{sec:finite-sample-bound}

We will derive a finite sample bound for discrete anchors. 
There are no fundamental issues that prevent the derivation of similar results for continuous anchors.
We write $\mathcal{A}$ for the set of levels of the random variable $A$. Unbalanced settings can impose difficulties in the finite-sample case as it becomes more challenging to estimate the penalty term. We analyse the behaviour of \emph{anchor regression} in the case where all anchor levels $A = a$, $a \in \mathcal{A}$, are explicitly given equal weight in the optimization procedure, i.e., the objective function for population \emph{anchor regression} is
\begin{equation*}
R(b) :=  \Etrain[(Y-X^\intercal b - \Etrain[Y-X^\intercal b |A])^{2}] + \frac{\gamma}{| \mathcal{A}|} \sum_{a \in \mathcal{A}}  (\Etrain[Y-X^\intercal b|A=a])^{2}.
\end{equation*}
Such a re-weighting is usually advisable in unbalanced settings. Otherwise, very few levels of $A$ can dominate the penalty term and limit its usefulness.  Note that, by Theorem~\ref{thm:anchor-regression}, $R(b)$ corresponds to the maximum $\ell_{2}$-risk under a uniform distribution on the levels of $A$:
\begin{align*}
R(b) &= \sup_{v \in C^{\gamma}} \mathbb{E}_{v}[(Y - X^{\intercal} b)^{2}].
\end{align*}
For data with unbalanced discrete anchor levels, the shape of $C^{\gamma}$
changes as anchor levels that occur with small probability are given less
weight. For discrete anchors, interpreting \emph{anchor regression} via
quantiles is only justified under re-weighting, see
Lemma~\ref{lemma:interpr-anch-regr-discr} in the Appendix.. 

To formulate the assumptions in a convenient form  we introduce additional notation for the special case of discrete anchors. We write $n_{a}$ for the number of observations for level $A=a$ and $n_{\text{min}}$ for the minimum number of observations, i.e., $n_{\text{min}} := \min_{a \in \mathcal{A}} n_{a}$. We write $\X^{(a)} \in \mathbb{R}^{n_{a} \times d}$ for the observations for which $A=a$. In other words, the rows of $\X^{(a)}$ consist of observations $\X_{i, \bullet}$ for which $\A_{i} = a$. Furthermore we write $\overline{\X}^{(a)}$ for the mean within the group, i.e., $\overline{\X}^{(a)} = \frac{1}{n_{a}} \sum_{i=1}^{n_{a}} \X_{i,\bullet}^{(a)} $. Analogously we define $\Y^{(a)} \in \mathbb{R}^{n_{a}}$ and $\overline{\Y}^{(a)}$.  
Using this notation, the high-dimensional \emph{anchor regression} estimator in \eqref{eq:12} but with equal weight regularization, analogous to the definition of $R(b)$ above, equals
\begin{align*}
\begin{split}
\hat b := \, & \argmin_{b} \frac{1}{| \mathcal{A} |} \sum_{a \in \mathcal{A}} \frac{1}{n_{a}}  \sum_{i=1}^{n_{a}} \left( \Y^{(a)}_{i} - \overline{\Y}^{(a)} - (\X^{(a)}_{i,\bullet} - \overline{\X}^{(a)} ) b \right)^{2} + \frac{\gamma}{| \mathcal{A}|} \sum_{a \in \mathcal{A}}  \left( \overline{\Y}^{(a)} - \overline{\X}^{(a)} b \right)^{2} + 2 \lambda  \| b \|_{1}.
\end{split}
\end{align*}
Here and in the following, we suppress the dependence of $\hat b$ on $\gamma$ and $\lambda$.
 For any $S \subseteq \{1,\ldots,d\}$ and stretch factor $L>0$ define the \emph{anchor compatibility constant}
\begin{align*}
&\hat{\phi}^{2}(L,S) := \\
&\min_{\| b_{S} \|_{1} =1 , \| b_{-S} \|_{1} \le L} |S|  \left( \frac{1}{| \mathcal{A}|} \sum_{a \in \mathcal{A}} \frac{1}{n_{a}} \sum_{i=1}^{n_{a}}  \left( (\X_{i,\bullet}^{(a)} - \overline{\X}^{(a)})  b \right)^{2}  + \frac{\gamma}{| \mathcal{A}|} \sum_{a \in \mathcal{A}} ( \overline{\X}^{(a)} b)^{2}  \right).
\end{align*}
To proceed, we need a lower bound on the compatibility constant $\hat{\phi}^{2}(L,S^{*})$ for $S^{*} := \{k : b_{k}^{\gamma} \neq 0\}$, the active set of $b^{\gamma}$. Note that for all $S$
\begin{equation*}
  \hat{\phi}^{2}(L,S) \ge  \min(\gamma,1)  \min_{\| b_{S} \|_{1} =1 , \| b_{-S} \|_{1} \le L}  \frac{|S|}{| \mathcal{A}|} \sum_{a \in \mathcal{A}} \frac{1}{n_{a}} \sum_{i=1}^{n_{a}}  \left( \X_{i,\bullet}^{(a)} b \right)^{2}.
\end{equation*}

For $|\mathcal{A}| =1$ the quantity on the right corresponds to the
ordinary compatibility constant in high-dimensional linear regression
\citep{van2016estimation}. The anchor compatibility constant can be
bounded analogously as the ordinary compatibility constant, see e.g.\
\citet{van2016estimation}.

When presenting asymptotic results as both $d = d_n > n \to
  \infty$, we  allow that the set $\mathcal{A}$, the shift matrix $\mathbf{M}$,
  the target quantity $b^{\gamma}$ and the structural equation model change for varying $n$.
\begin{theorem}\label{thm:finite-sample-bound-1}

Consider the model in \eqref{eq:32} and assume that $\varepsilon$ is
multivariate Gaussian. Moreover, assume that
$(\X_{i,\bullet}^{(a)},\Y_{i}^{(a)})$,  $i=1,\ldots,n_{a}$, are
i.i.d. random variables that follow the distribution of $(X,Y)|A=a$ under $\Ptrain$. Fix
$\gamma > 0$
and assume that $\hat \phi^{2}(8,S^{*})
\ge c$ for some constant $c>0$ with probability
$1-\delta$, and that $S^{*} \neq \emptyset$.
  Choose $t \ge 0$ such that
\begin{equation*}
  |S^{*}|^{2} (t+\log(d)+\log(|\mathcal{A}|))/n_{\mathrm{min}} \le c',
\end{equation*}
for some constant $c'>0$. Then, for $\lambda \ge C \sqrt{ (t + \log(d)+ \log( | \mathcal{A}|))/n_{\mathrm{min}}}$, with probability exceeding $1-10 \exp(-t)-\delta$,
\begin{equation*}
   R(\hat b) \le \min_{b}  R(b)  + C' \lambda^{2} |S^{*}|,
\end{equation*}
where the constants $ C,C' < \infty$ depend on $\max_{k} \mathrm{Var}(X_{k})$, $ \mathrm{Var}( Y - X^\intercal b^{\gamma}) )$,
$\max_{a \in \mathcal{A}} \|\Etrain[X | A=a] \|_{\infty}$, $\max_{a \in \mathcal{A}} |\Etrain[Y - X^\intercal b^{\gamma} | A=a] | ) $, $\gamma$, $c$ and $c'$. The variances are meant with respect to the measure $\Ptrain$.
\end{theorem}
  There are no fundamental issues that prevent the derivation of similar
  results for continuous anchors. The constant $8$ in the anchor compatibility constant $\hat \phi^{2}(8,S^{*})$ does not represent a theoretically meaningful critical value, it was chosen in an ad-hoc fashion to simplify the result.

Under the assumptions mentioned above, if we choose $\lambda \asymp \kappa  C \sqrt{ (t + \log(d)+ \log(|\mathcal{A}|))/n_{\mathrm{min}}} $ for $\kappa > \sqrt{2}$, $ t = \log(d)$ and assume that $\delta \rightarrow 0$, we obtain the following asymptotic result. For $d,n \rightarrow \infty$, with probability going to one,
\begin{equation*}
  R(\hat b) - \min_{b}   R(b) = \mathcal{O} \left(\frac{|S^{*}|(\log(d)+ \log(|\mathcal{A}|))}{n_{\mathrm{min}}} \right).
\end{equation*}
As $\hat b$ coincides with the Lasso for $\gamma=1$ and $|\mathcal{A}|=1$, it is worthwhile to compare this bound to risk bounds of the Lasso. The excess predictive risk of the Lasso in a comparable setting with appropriate choice of $\lambda$ is of the order  $\mathcal{O} \left( |S^{*}| \log(d)/n \right)$, see, e.g., \citet[][Chapter 6]{buhlmann2011statistics}. Hence the risk bounds will be of comparable order as long as $n/n_{\text{min}}$ is bounded.

\section{Numerical examples}\label{sec:real-world-example}

We provide two numerical examples. The first example shows how anchor regression can be used to improve replicability across perturbed data. In the second example, we discuss a prediction problem under distributional shifts. The code is available on \url{github.com/rothenhaeusler}.
\subsection{Genotype-tissue expression}\label{sec:genotype-tiss-expr}
The data was obtained from the Genotype-Tissue Expression (GTEx)
  portal \citep{gtex}.
  One of the GTEx datasets contains gene expression data from 53 tissues of 714
human donors,
in total comprising $n=11688$ observations of
$d=12948$ genes.

These samples were collected postmortem.  Gene expressions are subject to
various types of heterogeneity.
They vary not only between humans but also between different
tissues and individual cells. 13 out of the 53 tissues contain more than
300 observations. We conducted our analysis on these 13 tissues.

We will compare features that are relevant for prediction on one tissue with the features that are relevant for prediction on another tissue.  Our goal is to find relevant features that are not particular to the specific tissue at hand, but can also be found
  (replicated) on the other tissues.  Due to the heterogeneity between
  the tissues, this is a challenging task. The response variable $Y$ is the expression  of a target gene and the covariates $X$ are the expressions of all other
  genes. Mathematically, we associate with $y \in \{1,\ldots ,d\}$ the gene
  index of the target
  variable and $x = \{1,\ldots ,d\} \setminus y$ the gene indices of the
  expression covariates.

 For each tissue, the gene expressions and additional covariates
 are available.  These covariates contain geno-typing principal components,
 PEER factors, sex and genotyping platform. The geno-typing principal
 components and PEER factors 
(which are constructed from covariates
and gene expressions)
 account for some (but not
 all) of the
 confounding sources of expression variation, such as batch effects, environmental influences and sample history \citep{stegle2012using}.
Originally, it has been suggested to include the PEER factors when regressing
gene expression on genotype. Here, we use them  in an analysis of co-expression,
in spirit similarly to \citet{furlotte2011} or \citet{Steglenips}.
We will use these additional covariates as the anchor
 variables\footnote{From a theoretical standpoint, using the tissues as anchor is a reasonable choice as well. However, the empirical conditional expectations
     of each gene expression given the tissues is zero. The gene
   expressions have been normalized within each tissue and hence using the
   tissues as the anchor variable is not  meaningful for this dataset.}. 
We consider combinations of biological entities, and the PEER factors are
partially computed from the gene expressions.
Therefore, strictly speaking, the assumptions in Section~\ref{sec:anchor-stability} are not satisfied.
Assuming, however, that these PEER factors and geno-typing principal
components
are correlated with confounding sources of variation, 
using anchor stability with these proxy variables as anchor may 
still increase 
replicability of feature selections across data sets. Note that using anchor stability is justified even in cases where anchors are endogeneous, see 
the discussion in Section~\ref{sec:anchor-stability}  and the
corresponding theorem in the Appendix,
Section~\ref{sec:gener-vers-theor}.

\subsubsection{Improved replicability with stable anchor regression}
The goal is to investigate whether features that are relevant for prediction on one tissue are also relevant for prediction on other tissues. 
More specifically, we compute and rank
variables using the Lasso  and penalized \emph{anchor regression}
 on one specific tissue $t$. Then, we check whether the discoveries can also be replicated on the other tissues $t' \neq t$.

 How should we rank the covariates in an \emph{anchor regression} framework? By the discussion below Theorem~\ref{thm:stabcaus}, anchor stability is potentially a positive indicator for distributional replicability.
This suggests that ranking by anchor stability should improve replicability across heterogeneous domains of the data set. In cases where the anchor is only weakly correlated with the covariates, estimation of $b^{\gamma}$ will be unstable for $\gamma \rightarrow \infty$.
Thus, in the following, we do not test whether the coefficients are
 invariant across $\gamma \in [0,\infty)$ but check whether the
 individual \emph{anchor regression} coefficients are bounded away
 from $0$ for $\gamma \in [0,1]$. This can be seen as a weak form of
 anchor stability. 

Consider a fixed tissue $t$. For the \emph{anchor regression} method, we compute
\begin{equation}\label{eq:anchor-ranking}
  a_{y,k,t} := \min_{\gamma \in [0,1]}| \hat b_k^{\gamma,\lambda} |,
\end{equation}
where $\hat b^{\gamma,\lambda}$
is the $p-1$-dimensional anchor coefficient of a \emph{anchor regression} of
target variable $y \in \{1,\ldots,p \}$ on the other gene expressions
$x =  \{1,\ldots,p \} \setminus \{ y \}$. As regularization parameter
$\lambda$ we use the same as for the Lasso  regression (see
below). We also consider for \eqref{eq:anchor-ranking} the ranges
  $\gamma \in \{[0,0.25],[0,16]\}$ and show the results in Section~\ref{sec:addfig}.

 For comparison, we compute the Lasso  coefficients
\begin{equation}\label{eq:lasso-ranking}
  l_{y,k,t} := | (\hat b_\text{lasso})_k |,
\end{equation}
where $\hat b_\text{lasso}$ is the $p-1$-dimensional Lasso coefficient of a
Lasso regression of target variable $y \in \{ 1, \ldots,d \}$ on all other
variables $x = \{1,\ldots ,d\} \setminus y$, after removing the
effect of the anchor variables. By definition, $\hat b_{\text{lasso}}=  \hat b^{0,\lambda} $, i.e.\ the Lasso coefficient vector coincides with \emph{anchor regression} for $\gamma=0$ which implies $a_{y,k,t} \le l_{y,k,t}$. Hence, any nonzero effect found using anchor regression is also a nonzero effect using the Lasso. However the ranking for the two methods is different.
For both methods, a regularization parameter $\lambda$ has to be
chosen. We use the one from cross-validation as implemented in the function
\texttt{cv.glmnet} in the \texttt{R}-package \texttt{glmnet}. To make the methods comparable, this regularization parameter was also used for the \emph{anchor regression} method.

 We evaluate how many of the largest effects found by stable anchor
   regression or Lasso can be replicated on
 another tissue. The results are depicted in
 Figure~\ref{fig:gtexcomparison}. The black solid line depicts how many of
 the $K=1,\ldots,20$ largest effects $l_{y,k,t}$ are also among the
   $K$ largest effects $l_{y,k,t'}$ on another tissue $t' \neq t$ for a fixed
 target $y$ (and then averaged over $y$, see below). Analogously,
 the red dashed line shows how many of the $K$
 largest effects $a_{y,k,t}$ are also among the $K$ largest effects
 $l_{y,k,t'}$ on a tissue $t' \neq t$. Finally, the green dotted line shows
 how many of the $K$ largest effects $a_{y,k,t}$ are also among the
   $K$ largest effects $a_{y,k,t'}$ on a tissue $t' \neq t$. The results are
 summed over all choices of $t' \neq t$ and averaged over 200 random
 choices of $y \in \{1,\ldots,12948\}$.

 Both anchor stable and Lasso methods are better than random guessing.
 Ranking by anchor stable regression results in improved
 replicability across tissues. 
Note that this is a challenging data set and
the predictive power among genes is small: the average $R^2$ for a Lasso
run estimated and evaluated on disjoint parts of one tissue is .37. The
average $R^{2}$ for a Lasso run estimated on one tissue and evaluated on
another tissue is slightly negative. In Section \ref{sec:addfig}, we also discuss the degree of
  replicability for the parameter $b^{\rightarrow \infty}$. 

\begin{figure}
  \begin{center}
  \includegraphics[scale=0.5]{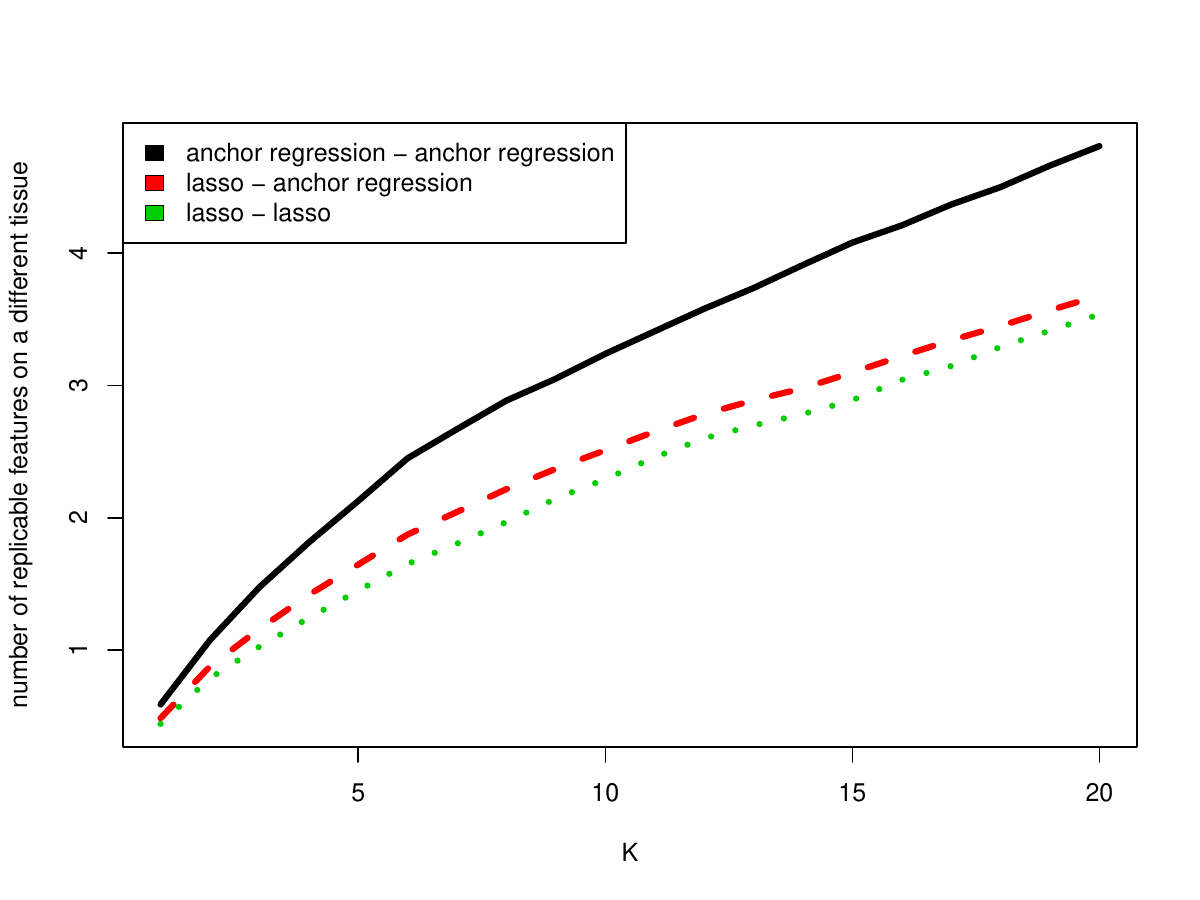}
  \end{center}
  \caption{Replicability of variable selection in GTEx
      data. Plotting how many of the $K \in \{1,\ldots ,20\}$ top-ranked
    features found by \emph{anchor regression} and Lasso  on one tissue $t$ are  also one of the $K$ top-ranked features on another
    tissue $t'$. The results are summed over all
    other 
    tissues $t' \neq t$, averaged over all tissues $t$ and
    averaged over 200 random choices of $y$, and they are plotted
    as $y$-coordinates.
    For \emph{anchor regression} the ranking is according to
      \eqref{eq:anchor-ranking}, and for 
      Lasso according to \eqref{eq:lasso-ranking}. The
      legend describes the method used on one tissue $t$ and the method used on another
      tissue $t'$. \emph{Anchor regression}
      exhibits the highest degree of
      replicability.}\label{fig:gtexcomparison}
\end{figure}

\subsection{Bike sharing data set}\label{sec:bike-sharing-data}

The data set is taken from the UCI machine learning repository \citep{bikedataset,Dua:2017}. It contains $n=17379$ hourly counts of bike rentals from 2011 to 2012 of the Capital bike share in Washington D.C. The goal is to predict bike rentals (variable cnt) using weather data reliably across days. As the variable cnt is a count, a square-root transformation was carried out. The effect of categorical variables, for which shift interventions are not meaningful
(this includes the variables working day, weekday, holiday), was removed in a pre-processing step. While we generally recommend removing the effect of variables that cannot be shifted, in this particular example the pre-processing step makes no discernible difference in the resulting plot, see Figure~\ref{fig:pre-processing} in the Appendix.
The data set contains the numerical covariates temperature, feeling temperature, humidity and windspeed. The variable hour is nested within the variable ``date''. We will first conduct the analysis ignoring the variable ``hour'' as this application is closest to Theorem~\ref{thm:anchor-regression}, Lemma~\ref{lemma:interpr-anch-regr} and Lemma~\ref{lemma:interpr-anch-regr-discr}. In practice, one would also want to include ``hour'' as a predictor in the model. We discuss this case further below.

There are large fluctuations in the usage of bikes that cannot be explained
by weather data alone \citep{bikedataset}. Instead of using the discrete
variable 'date' for prediction, we use it as an anchor $A$.  
More detailed, the anchor variable is discrete with one level per
  day.

This choice of anchor variable allows us to investigate the
performance of the algorithm in a setting with strong
heterogeneities. The goal is to predict the count of bike rentals in a
reliable fashion using the covariates temperature, feeling
temperature, humidity and windspeed.

As evaluation metric, we consider quantiles of the
conditional mean squared error given the anchor variable. Intuitively speaking, we want to train a prediction rule that works reliably across days. Practically, this means that for each fixed day, we average over the prediction loss and then compute quantiles across days.   The quantiles
of the conditional squared error
$\mathbb{E}[(Y-X^\intercal b)^{2} |A]$ are a proxy for the right-hand side of
equation~\eqref{eq:27} being the worst case risk across perturbations of a
certain level, cf. Lemma~\ref{lemma:interpr-anch-regr-discr} in the Appendix.
The data was split into $5$ consecutive blocks. The estimator was trained on $4$ of the $5$ blocks and tested on the left-out block. Results are averaged over the five possible train-test split.
Quantiles of the daily averaged squared error on the test data set $\hat{\mathbb{E}}_{\text{test}}[(Y-X^{\intercal} \hat{b}^\gamma)^{2} |A]$,
are depicted in Figure~\ref{fig:cvxquantiles}.

The optimal choice of $\gamma$ as evaluated on the test data set as a function
of the quantile and the corresponding predictive performance can be found
in Figure~\ref{fig:optgamma}. This motivates choosing $\gamma$ by minimizing quantiles of the loss on held-out data. We describe this procedure in more detail below. 
Figure~\ref{fig:cvxquantiles} shows that for small quantiles, small values
of $\gamma$ are slightly preferred, while for quantiles close to one, large values of $\gamma$ clearly outperform smaller values. This is in line with the theory
presented in Section~\ref{sec:optim-pred-perf}.

However, as the direction and strength of the perturbations usually also
changes to some extent between training and test data set we do not
recommend simply using $\lim_{\gamma \rightarrow \infty} \hat
b^{\gamma}$. In practice, we do not advise to choose $\gamma$ based on
Lemma~\ref{lemma:interpr-anch-regr} or Lemma~\ref{lemma:interpr-anch-regr-discr} as the interplay of the penalization parameter and quantiles of $\mathbb{E}[(Y-X^\intercal b^{\gamma})^{2} |A]$ is more
involved for non-Gaussian distributions.
 Instead, we recommend choosing an optimal $\gamma$ based on cross-validation.

 The cross-validation approach (as used in Figure~\ref{fig:optgamma}) proceeds as follows. First, choose a quantile $\alpha$ (for example $\alpha=90\%$). In each of the folds, the data is split in a training data set and a test data set, such that each level of the anchor variable only appears in one of the data sets. Then, for varying $\gamma$, compute $\hat b^{\gamma}$ on the training data set and estimate
the $\alpha$-quantile of $\mathbb{E}[(Y-X^\intercal b^{\gamma})^{2}|A]$ on the test data set.  After averaging the estimated quantiles over the folds, choose $\gamma$ such that the chosen quantile is
minimized. For this approach to work, we have to make an assumption that heterogeneities of the future data generating process are in some sense similar to the heterogeneities observed in the training data set. This assumption is made precise in Lemma~\ref{lemma:interpr-anch-regr-discr} in the Appendix for discrete anchors.  
\begin{figure}
\begin{center}
\includegraphics[scale=0.4]{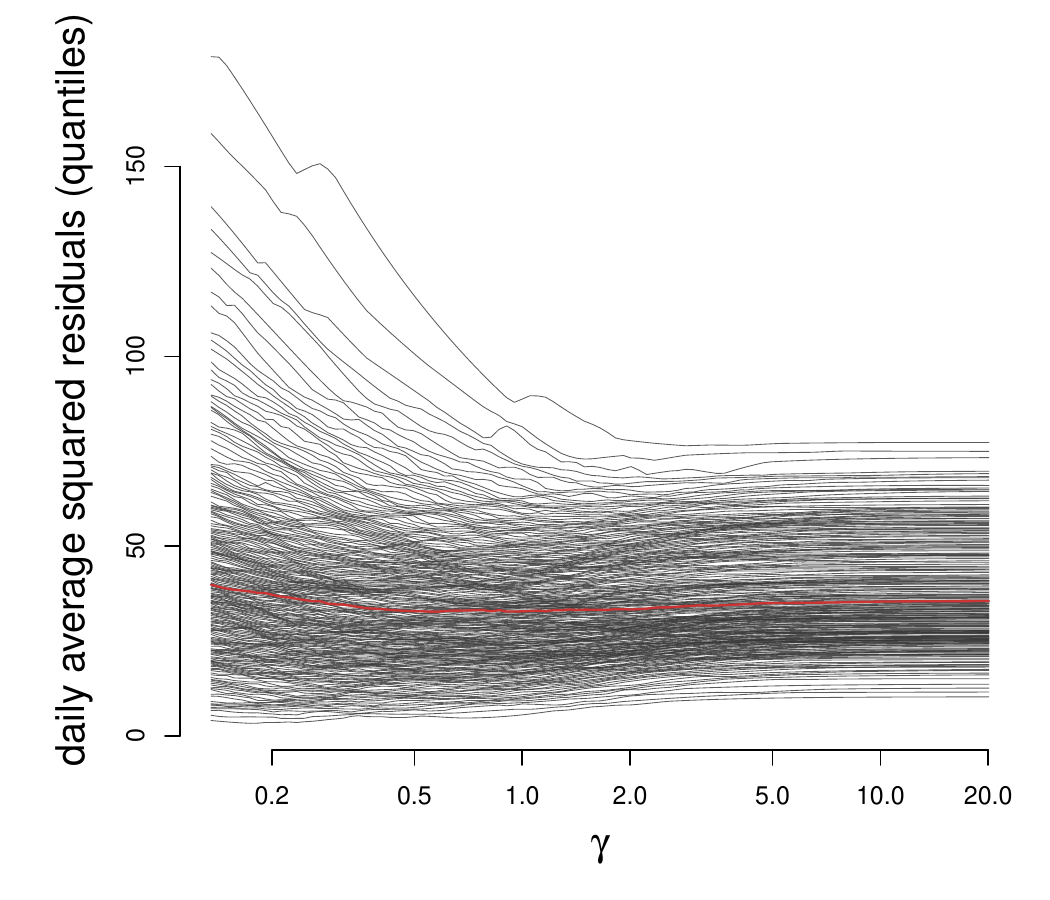}
\end{center}
\caption{Daily average squared residuals $\hat{\mathbb{E}}_{\text{test}}[(Y-X^{\intercal} \hat b^{\gamma})^{2}|A]$ as a function of $\gamma$. Each line corresponds to a quantile of $\hat{\mathbb{E}}_{\text{test}}[(Y-X^{\intercal} \hat b^{\gamma})^{2}|A]$. The quantiles are chosen in the set $\{0.05,0.01,\ldots,0.995\}$, with the median marked in red. For growing $\gamma$, the upper percentiles of $\hat{\mathbb{E}}_{\text{test}}[(Y-X^{\intercal} \hat b^{\gamma})^{2}|A]$
are decreasing while the lower percentiles are slightly increasing. This is in line with the theory presented in Section~\ref{sec:optim-pred-perf}. The distribution of bike rentals is expected to change from day to day. For growing $\gamma$, the upper percentiles of the loss are reduced, i.e., predictions are increasingly reliable across days.
A comparison to OLS with $\gamma = 1$ is given in the right panel of
  Figure \ref{fig:optgamma}.
}
\label{fig:cvxquantiles}
\end{figure}

\begin{figure}
\begin{center}
\includegraphics[scale=0.6]{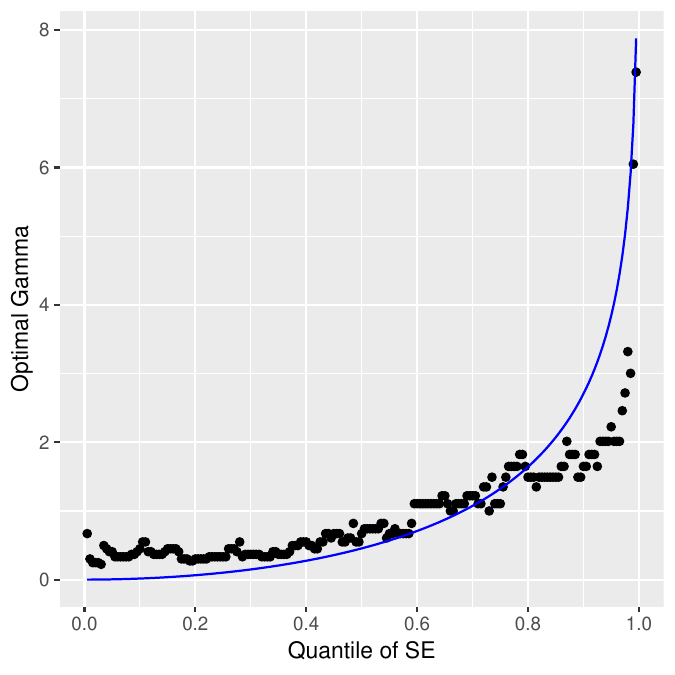}
\includegraphics[scale=0.6]{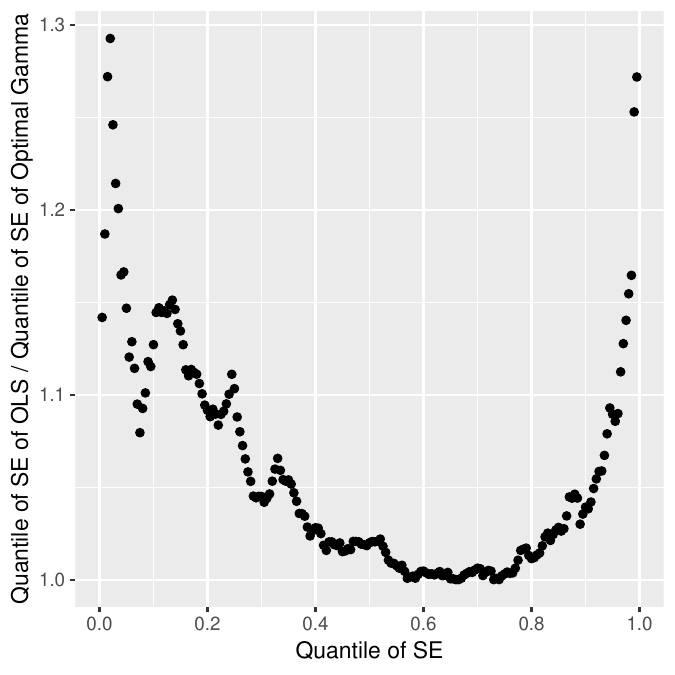}
\end{center}
\caption{
Optimal choice of $\gamma$ and predictive performance of
  $\emph{anchor regression}$ for varying quantiles of the squared error on
  the bike-sharing data set. On the left-hand side, the optimal choice of
  $\gamma$ is depicted as a function of quantiles of the daily averaged error, $\hat{\mathbb{E}}_{\text{test}}[(Y-X^{\intercal} \hat{b}^{\gamma})^{2} |A]$. The blue line shows the theoretically optimal
  choice of $\gamma$ using Lemma~\ref{lemma:interpr-anch-regr}. The black dots show the optimal choice of $\gamma$ as evaluated on the test data set.   For
  growing quantiles, the optimal choice $\gamma=\gamma_{\text{opt}}$
  increases. For example, $\gamma \approx 0.35$ is optimal for minimizing
  the $5\%$-Quantile of $\hat{\mathbb{E}}_{\text{test}}[(Y-X^{\intercal} \hat
  b^{\gamma})^{2}|A]$. Similarly, $\gamma \approx 2$ is optimal for
  minimizing the $90\%$-Quantile of $\hat{\mathbb{E}}_{\text{test}}[(Y-X^{\intercal}
  \hat{b}^{\gamma})^{2} |A]$. On the right-hand side, the performance with
  the optimal estimated $\gamma$ is shown in terms of quantiles of
  $\hat{\mathbb{E}}_{\text{test}}[(Y - X^{\intercal} \hat{b}^{\gamma})^{2}|A]$, relative to
  ordinary least squares (OLS). For example, for the
  $90\%$-quantile, the optimal choice of $\gamma$ leads to a
  $10\%$-improvement of $\emph{anchor regression}$ compared to ordinary
  least squares. The biggest improvements compared to OLS are obtained for both very small and very large quantiles. The quantiles of $\hat{\mathbb{E}}_{\text{test}}[(Y-X^{\intercal} \hat{b}^{\gamma})^{2}|A]$ were estimated using 5-fold cross-validation.}\label{fig:optgamma}
\end{figure}

As discussed above, the application above is close to the theory presented in Section~\ref{sec:pop_anchor}, but in practice one would also want to include the predictor ``hour''. 
As an alternative experiment to the one shown above, we run
a regression of the target variable on the predictor ``hour'' and run anchor regression on the residuals. For the final prediction, we then 
add the predictions from both models. The variable hour differs from the other variables in the sense that it is nested within the anchor date. Thus, building the overall model in such a hierarchical fashion is not supported by our current theory. 
  The results can be found in the Appendix: Figure~\ref{fig:hour} in Section~\ref{sec:add-hour} 
is equivalent to Figure~\ref{fig:cvxquantiles}, but \emph{anchor regression} is run on the residuals after regressing out the effect of ``hour''. For large quantiles of the conditional loss, $\gamma \gg 1$ outperforms $\gamma < 1$, but the relationship is not monotonous. Figure~\ref{fig:optgamma_hour} in Section~\ref{sec:add-hour} of the Appendix is similar to Figure~\ref{fig:cvxquantiles} but with the modified anchor regression procedure described above. The anchor regression procedure performs better than ordinary least-squares $(\gamma = 1)$ for all considered quantiles.

\section{Practical guidance}

In this section we summarize our results and give high-level
  guidance for using \emph{anchor regression}, based on our empirical experience
  and theoretical results. 
 
\paragraph{Possible Applications.} 
  \emph{Anchor regression} can be applied 
in settings, where we are given data from a target 
variable $Y$ and covariates $X$ and
are interested in generalizing across
heterogeneous data sets.
Examples of such distribution changes include batch effects, population
shifts, and heterogeneity across time or locations. In
the case of prediction, 
the approach aims to achieve robust predictions
across data sets. 
\emph{Anchor regression} is optimal 
if the
data sets differ by (restricted)
shift interventions.
For the goal of parameter estimation, \emph{anchor regression}
can be used to find features that are invariant 
across a (restricted) set of distributions, see Section~\ref{sec:anchor-stability}.
Thus,  the approach might  help to increase
 the replicability of discoveries across data sets. 

\paragraph{Choice of the anchor variable.}
In the case of prediction, the main assumptions are linearity of the system
and exogeneity of the anchor. We recommend to choose the anchor based on
the type of robustness or invariance one aims to obtain. For example, if one intends to
obtain robustness of the prediction rule across locations, we recommend using location as an anchor variable. If the goal is to achieve robustness across time, we recommend using discretized time windows as an anchor variable.
In our theory, this recommendation is justified by
Theorem~\ref{thm:anchor-regression}. Different choices of the anchor
correspond to different matrices $\mathbf{M}$, which in turn provide
protection against different distributional shifts.

In the case of estimation, the exogeneity assumption for the anchor
variable can be dropped. Details can be found in
Section~\ref{sec:anchor-stability} and in
Section~\ref{sec:gener-vers-theor} in the Appendix. In that case, the anchor should be chosen such that it affects the covariates of interest as much as possible. 

\paragraph{Choice of the regularization parameter.} When using \emph{anchor regression} for prediction, one has to
choose a regularization parameter $\gamma$. If
possible, this should be done 
based on subject matter knowledge. For example, if one expects
perturbations on future data sets to be at most 1.5 times as
large as on the training data sets, $\gamma = 1.5$ is a sensible choice. 
If the anchor variable has many categorical levels, it is also
possible to choose $\gamma$ using some form of leave level out
cross-validation. This approach is described in
Section~\ref{sec:bike-sharing-data}. For data sets where the above
considerations do not apply, we believe that $\gamma = 2$ is a good default
choice. 

For screening via anchor stability, in theory it is sufficient to test
whether the two endpoints $\gamma=0$ and $\gamma=\infty$ of \emph{anchor regression} agree, see
Proposition~\ref{prop:constant}. In cases where the anchor is only weakly associated with the covariates, estimation of $b^{\infty}$ will be unstable. Thus, in practice we recommend to screen based
on a weak form of anchor stability, as in
equation~\eqref{eq:anchor-ranking}. 
That choice can be considered a heuristic, as its theoretical implications are yet to be investigated.

\paragraph{Limitations.}

All extrapolation statements of \emph{anchor regression}
rely on 
the assumption of linearity.
Using \emph{anchor regression} for prediction generally does not guarantee
protection against ``black swan events''.  More specifically, \emph{anchor
  regression} 
  is not leading to robust
  prediction
  when the heterogeneity between the data sets is different from the
  restricted set of shift interventions that have been observed on the
  training data sets.

For example, in Theorem~\ref{thm:anchor-regression}, the set $C^{\gamma}$ contains shifts that lie in the span of $\mathbf{M}$, as opposed to shifts in arbitrary directions. In cases where distribution shifts are 
complex, in the sense that distributions change arbitrarily between data sets, 
neither \emph{anchor regression} nor any other method can provide reliable predictions. If the anchor does not shift any distributions, i.e.\ if the distribution of $(X,Y)$ is constant across values of $A$ then there is no benefit from using the \emph{anchor regression} approach. Note however, that in this case there is also little harm from using the \emph{anchor regression} approach as the penalty term in equation~\eqref{eq:27} will be close to zero.

\section{Discussion and outlook}

  We have introduced \emph{anchor regression}, a regularization approach
  for fitting linear models. We have shown that this approach optimizes
  worst case prediction risk over a class of perturbations and that it also
  leads to improved replicability of variable selection across different
  perturbed heterogeneous datasets. The methodology has relations to
  invariance properties from causality and the concrete proposed procedure of
  \emph{anchor regression} interpolates between three common statistical
  estimation schemes, namely partialling out (i.e., adjusting for)
  exogeneous variables, ordinary least squares and two-stage least squares
  from instrumental variables regression (with exogeneous instruments).

The penalty in \emph{anchor regression} corresponds to the change in
prediction loss under certain perturbations. More specifically, these
perturbations are modelled as random or deterministic shift interventions
and are estimated
from a heterogeneous training data set. We have explored the prediction
behavior, both in terms of size and direction of the considered
perturbations. When considering the regularization path of \emph{anchor
  regression} as a function of the penalty or regularization parameter, we
prove some stability and replicability for variable importance or variable
selection over a range of perturbations, i.e., a range of potentially new
heterogeneous data sets.
Thus, \emph{anchor regression} also contributes to
much desired improved replicability of variable importance.
We also derived
a finite sample bound for worst case prediction in the
high-dimensional case.

We consider the behavior of \emph{anchor regression} on real-data applications, in terms of replicability of variable selection and prediction on new
potentially perturbed data. We believe that it is worthwhile to explore
penalization schemes that exploit heterogeneities that occur in the
training distribution and lead to robustness and replicability on new
perturbed test data, i.e., generalizing to new unobserved heterogeneity.
Such a regularization allows to
explicitly balance the tradeoffs between predictive performance on
perturbed and unperturbed data sets, while avoiding
the loss in prediction accuracy that is incurred when using more
conservative approaches (e.g., causal parameters).

Looking ahead, there are
some avenues which we think are worthwhile to pursue. In the following, we
outline two directions that seem particularly promising.

\paragraph{Beyond shift interventions.}

Instead of considering shift interventions, it may be interesting to look
at penalty schemes that arise from other types of perturbations, such as
noise, edge functions and do-interventions. Depending on the
application, such interventions may be more appropriate than
shift-interventions. In this light, structural equation modelling can serve
as a scheme to generate and explore new types of perturbation penalties.
Furthermore, it allows to obtain optimality statements to better understand
the tradeoffs between perturbation stability and predictive performance.

\paragraph{Nonlinear models.}
For the \emph{anchor regression} method to be practical in a wide range of scenarios, it is important to extend it to non-linear models. Using a bias-variance decomposition, with $P_{A} = \Etrain[\bullet|A]$ the prediction loss of a non-linear function $g(X)$ can be decomposed as
\begin{align*}
  \Etrain[(Y-g(X))^2|A] &= \Etrain[((\mathrm{Id} - P_A) (Y-g(X)))^2 |A] + (P_A(Y-g(X)))^2
\end{align*}
If the conditional variance is constant across strata defined by $A=a$, then the conditional loss simplifies to
\begin{align*}
  \Etrain[(Y-g(X))^2|A] &= \Etrain[((\mathrm{Id} - P_A) (Y-g(X)))^2] + (P_A(Y-g(X)|A))^2.
\end{align*}
This decomposition motivates non-linear \emph{anchor regression}, which we define as the solution to
\begin{equation*}
  g^\gamma := \arg \min_{g \in \mathcal{G}} \Etrain[((\mathrm{Id} - P_A) (Y-g(X)))^2] + \gamma \Etrain[(P_A(Y-g(X)))^2],
\end{equation*}
for an appropriate set of functions $\mathcal{G}$. Qualitatively this
estimator behaves similarly to \emph{anchor regression}. As before, it
interpolates between nonlinear versions of PA, OLS and IV. For $\gamma
\rightarrow \infty$, non-linear \emph{anchor regression} will strive for
invariance in the sense that it tries to keep $\mathbb{E}[(Y-g(X))^2|A]$
constant across all levels of $A$. The set of interventions that
nonlinear \emph{anchor regression} protects against for a fixed $\gamma$ is not as
straightforward to describe as in
Theorem~\ref{thm:anchor-regression}. However, we conjecture that this
estimator behaves similarly to linear \emph{anchor regression} on data sets, in the sense that it potentially improves replicability across heterogeneous
regimes and improves robustness of prediction rules across the strata
defined by $A$. Other non-linear extensions of \emph{anchor regression} and some preliminary empirical evidence can be found in \citet{buhlmann2018invariance}.
We believe that it is a promising avenue to further investigate the behaviour of these and related estimators both in theory and practice.

\section*{Acknowledgements}
We thank Martin Emil Jakobsen for pointing out the link between anchor regression and $k$-class estimators. 
We thank several reviewers for various helpful comments.
DR received funding from the ONR grant N00014-17-1-2176. PB received funding from the European Research Council under the grant agreement No. 786461 (CausalStats -- ERC-2017-ADG).

\bibliographystyle{plainnat}
\bibliography{bibliography}

\begin{thebibliography}{53}
\providecommand{\natexlab}[1]{#1}
\providecommand{\url}[1]{\texttt{#1}}
\expandafter\ifx\csname urlstyle\endcsname\relax
  \providecommand{\doi}[1]{doi: #1}\else
  \providecommand{\doi}{doi: \begingroup \urlstyle{rm}\Url}\fi

\bibitem[Aldrich(1989)]{aldrich1989}
J.~Aldrich.
\newblock Autonomy.
\newblock \emph{Oxford Economic Papers}, 41:\penalty0 15--34, 1989.

\bibitem[Bollen(1989)]{Bollen1989}
K.A. Bollen.
\newblock \emph{Structural Equations with latent variables}.
\newblock John Wiley \& Sons, 1989.

\bibitem[Boucheron et~al.(2013)Boucheron, Lugosi, and
  Massart]{boucheron2013concentration}
S.~Boucheron, G.~Lugosi, and P.~Massart.
\newblock \emph{Concentration inequalities: A nonasymptotic theory of
  independence}.
\newblock Oxford university press, 2013.

\bibitem[Bowden and Turkington(1990)]{bowden1990instrumental}
R.J. Bowden and D.A. Turkington.
\newblock \emph{Instrumental variables}, volume~8.
\newblock Cambridge University Press, 1990.

\bibitem[B{\"u}hlmann(2018)]{buhlmann2018invariance}
P.~B{\"u}hlmann.
\newblock Invariance, causality and robustness.
\newblock \emph{arXiv preprint arXiv:1812.08233}, 2018.

\bibitem[B{\"u}hlmann and {van de Geer}(2011)]{buhlmann2011statistics}
P.~B{\"u}hlmann and S.~{van de Geer}.
\newblock \emph{Statistics for high-dimensional data: Methods, theory and
  applications}.
\newblock Springer, 2011.

\bibitem[Carithers et~al.(2015)Carithers, Ardlie, Barcus, Branton, Britton,
  Buia, Compton, DeLuca, Peter-Demchok, Gelfand, Guan, Korzeniewski, Lockhart,
  Rabiner, Rao, Robinson, Roche, Sawyer, Segrè, Shive, Smith, Sobin, Undale,
  Valentino, Vaught, Young, and Moore]{gtex}
L.~Carithers, K.~Ardlie, M.~Barcus, P.~Branton, A.~Britton, S.~Buia,
  C.~Compton, D.~DeLuca, J.~Peter-Demchok, E.~Gelfand, P.~Guan,
  G.~Korzeniewski, N.~Lockhart, C.~Rabiner, A.~Rao, K.~Robinson, N.~Roche,
  S.~Sawyer, A.~Segrè, C.~Shive, A.~Smith, L.~Sobin, A.~Undale, K.~Valentino,
  J.~Vaught, T.~Young, and H.~Moore.
\newblock A novel approach to high-quality postmortem tissue procurement: The
  gtex project.
\newblock \emph{Biopreservation and Biobanking}, 13\penalty0 (5):\penalty0
  311–319, 2015.

\bibitem[Dawid(2000)]{dawid2000causal}
P.~Dawid.
\newblock Causal inference without counterfactuals.
\newblock \emph{Journal of the American Statistical Association}, 95:\penalty0
  407--424, 2000.

\bibitem[Dheeru and Karra~Taniskidou(2017)]{Dua:2017}
D.~Dheeru and E.~Karra~Taniskidou.
\newblock {UCI} machine learning repository, 2017.
\newblock URL \url{http://archive.ics.uci.edu/ml}.

\bibitem[Didelez et~al.(2010)Didelez, Meng, and
  Sheehan]{didelez2010assumptions}
V.~Didelez, S.~Meng, and N.A. Sheehan.
\newblock Assumptions of {IV} methods for observational epidemiology.
\newblock \emph{Statistical Science}, 25:\penalty0 22--40, 2010.

\bibitem[Eberhardt and Scheines(2007)]{Eberhardt2007}
F.~Eberhardt and R.~Scheines.
\newblock Interventions and causal inference.
\newblock \emph{Philosophy of Science}, 74:\penalty0 981--995, 2007.

\bibitem[Entner et~al.(2013)Entner, Hoyer, and Spirtes]{entner2013data}
D.~Entner, P.~Hoyer, and P.~Spirtes.
\newblock Data-driven covariate selection for nonparametric estimation of
  causal effects.
\newblock In \emph{Artificial Intelligence and Statistics}, pages 256--264,
  2013.

\bibitem[Fan and Zhang(1999)]{fan1999statistical}
J.~Fan and W.~Zhang.
\newblock Statistical estimation in varying coefficient models.
\newblock \emph{Annals of Statistics}, 27:\penalty0 1491--1518, 1999.

\bibitem[Fanaee-T and Gama(2013)]{bikedataset}
H.~Fanaee-T and J.~Gama.
\newblock Event labeling combining ensemble detectors and background knowledge.
\newblock \emph{Progress in Artificial Intelligence}, pages 1--15, 2013.
\newblock ISSN 2192-6352.
\newblock \doi{10.1007/s13748-013-0040-3}.
\newblock URL \url{[Web Link]}.

\bibitem[Friedman et~al.(2007)Friedman, Hastie, H{\"o}fling, and
  Tibshirani]{friedman2007pathwise}
J.~Friedman, T.~Hastie, H.~H{\"o}fling, and R.~Tibshirani.
\newblock Pathwise coordinate optimization.
\newblock \emph{Annals of Applied Statistics}, 1\penalty0 (2):\penalty0
  302--332, 2007.

\bibitem[Fuller(2009)]{fuller2009measurement}
W.~Fuller.
\newblock \emph{Measurement error models}, volume 305.
\newblock John Wiley \& Sons, 2009.

\bibitem[Furlotte et~al.(2011)Furlotte, Kang, Ye, and Eskin]{furlotte2011}
N.~A. Furlotte, H.~M. Kang, C.~Ye, and E.~Eskin.
\newblock Mixed-model coexpression: calculating gene coexpression while
  accounting for expression heterogeneity.
\newblock \emph{Bioinformatics}, 27\penalty0 (13):\penalty0 i288--i294, 2011.

\bibitem[Gao et~al.(2017)Gao, Chen, and Kleywegt]{gao2017wasserstein}
R.~Gao, X.~Chen, and A.~Kleywegt.
\newblock Wasserstein distributional robustness and regularization in
  statistical learning.
\newblock \emph{arXiv preprint arXiv:1712.06050}, 2017.

\bibitem[Greenland et~al.(1999)Greenland, Pearl, and
  Robins]{greenland1999causal}
S.~Greenland, J.~Pearl, and J.M. Robins.
\newblock Causal diagrams for epidemiologic research.
\newblock \emph{Epidemiology}, 10:\penalty0 37--48, 1999.

\bibitem[Haavelmo(1944)]{haavelmo1944}
T.~Haavelmo.
\newblock The probability approach in econometrics.
\newblock \emph{Econometrica}, 12:\penalty0 S1--S115 (supplement), 1944.

\bibitem[Hastie and Tibshirani(1993)]{hastie1993varying}
T.~Hastie and R.~Tibshirani.
\newblock Varying-coefficient models.
\newblock \emph{Journal of the Royal Statistical Society, Series B},
  55:\penalty0 757--796, 1993.

\bibitem[Heinze-Deml and Meinshausen(2018)]{heinze2017guarding}
C.~Heinze-Deml and N.~Meinshausen.
\newblock Conditional variance penalties and domain shift robustness.
\newblock \emph{arXiv preprint arXiv:1710.11469}, 2018.

\bibitem[Huber(1964)]{huber1964robust}
P.J. Huber.
\newblock Robust estimation of a location parameter.
\newblock \emph{Annals of Mathematical Statistics}, 35\penalty0 (1):\penalty0
  73--101, 1964.

\bibitem[Huber(1973)]{huber1973robust}
P.J. Huber.
\newblock Robust regression: Asymptotics, conjectures and monte carlo.
\newblock \emph{Annals of Statistics}, pages 799--821, 1973.

\bibitem[Klepper and Leamer(1984)]{klepper1984consistent}
S.~Klepper and E.~Leamer.
\newblock Consistent sets of estimates for regressions with errors in all
  variables.
\newblock \emph{Econometrica}, pages 163--183, 1984.

\bibitem[Korb et~al.(2004)Korb, Hope, Nicholson, and Axnick]{korb2004}
K.~Korb, L.~Hope, A.~Nicholson, and K.~Axnick.
\newblock Varieties of causal intervention.
\newblock In \emph{Proceedings of the Pacific Rim Conference on AI}, pages
  322--331, 2004.

\bibitem[Lauritzen and Spiegelhalter(1988)]{lauritzen1988local}
S.L. Lauritzen and D.J. Spiegelhalter.
\newblock Local computations with probabilities on graphical structures and
  their application to expert systems.
\newblock \emph{Journal of the Royal Statistical Society, Series B},
  50:\penalty0 157--224, 1988.

\bibitem[Leamer(1978)]{leamer1978least}
E.~Leamer.
\newblock Least-squares versus instrumental variables estimation in a simple
  errors in variables model.
\newblock \emph{Econometrica}, pages 961--968, 1978.

\bibitem[Magliacane et~al.(2018)Magliacane, van Ommen, Claassen, Bongers,
  Versteeg, and Mooij]{magliacane2017causal}
S.~Magliacane, T.~van Ommen, T.~Claassen, S.~Bongers, P.~Versteeg, and J.~M.
  Mooij.
\newblock Domain adaptation by using causal inference to predict invariant
  conditional distributions.
\newblock In S.~Bengio, H.~Wallach, H.~Larochelle, K.~Grauman, N.~Cesa-Bianchi,
  and R.~Garnett, editors, \emph{Advances in Neural Information Processing
  Systems 31}, pages 10846--10856. Curran Associates, Inc., 2018.

\bibitem[Meinshausen(2018)]{meinshausen18robust}
N.~Meinshausen.
\newblock Causality from a distributional robustness point of view.
\newblock In \emph{2018 IEEE Data Science Workshop (DSW)}, pages 6--10, June
  2018.

\bibitem[Meinshausen and B{\"u}hlmann(2015)]{meinshausen2015maximin}
N.~Meinshausen and P.~B{\"u}hlmann.
\newblock Maximin effects in inhomogeneous large-scale data.
\newblock \emph{Annals of Statistics}, 43\penalty0 (4):\penalty0 1801--1830,
  2015.

\bibitem[Nagar(1959)]{nagar1959bias}
A.~Nagar.
\newblock The bias and moment matrix of the general k-class estimators of the
  parameters in simultaneous equations.
\newblock \emph{Econometrica}, pages 575--595, 1959.

\bibitem[Pan and Yang(2010)]{pan2010survey}
S.~Pan and Q.~Yang.
\newblock A survey on transfer learning.
\newblock \emph{IEEE Transactions on Knowledge and Data Engineering},
  22\penalty0 (10):\penalty0 1345--1359, 2010.

\bibitem[Pearl(2009)]{Pearl2009}
J.~Pearl.
\newblock \emph{Causality: Models, reasoning, and inference}.
\newblock Cambridge University Press, 2nd edition, 2009.

\bibitem[Pearl and Bareinboim(2014)]{pearl2014external}
J.~Pearl and E.~Bareinboim.
\newblock External validity: From do-calculus to transportability across
  populations.
\newblock \emph{Statistical Science}, pages 579--595, 2014.

\bibitem[Peters et~al.(2016)Peters, B{\"u}hlmann, and
  Meinshausen]{peters2016causal}
J.~Peters, P.~B{\"u}hlmann, and N.~Meinshausen.
\newblock Causal inference by using invariant prediction: Identification and
  confidence intervals.
\newblock \emph{Journal of the Royal Statistical Society, Series B},
  78\penalty0 (5):\penalty0 947--1012, 2016.

\bibitem[Peters et~al.(2017)Peters, Janzing, and Sch\"olkopf]{Peters2017}
J.~Peters, D.~Janzing, and B.~Sch\"olkopf.
\newblock \emph{Elements of causal inference: Foundations and learning
  algorithms}.
\newblock MIT Press, 2017.

\bibitem[Pfister et~al.(2019)Pfister, Bauer, and Peters]{Pfister2018}
N.~Pfister, S.~Bauer, and J.~Peters.
\newblock Learning stable and predictive structures in kinetic systems.
\newblock \emph{Proceedings of the National Academy of Sciences}, 116\penalty0
  (51):\penalty0 25405--25411, 2019.

\bibitem[Pinheiro and Bates(2000)]{pinheiro2000linear}
J.C. Pinheiro and D.M. Bates.
\newblock Linear mixed-effects models: Basic concepts and examples.
\newblock \emph{Mixed-effects models in S and S-Plus}, pages 3--56, 2000.

\bibitem[Robins et~al.(2000)Robins, Hernan, and Brumback]{robins2000marginal}
J.M. Robins, M.A. Hernan, and B.~Brumback.
\newblock Marginal structural models and causal inference in epidemiology.
\newblock \emph{Epidemiology}, 11:\penalty0 550--560, 2000.

\bibitem[Rojas-Carulla et~al.(2018)Rojas-Carulla, Sch{\"o}lkopf, Turner, and
  Peters]{rojas2015causal}
M.~Rojas-Carulla, B.~Sch{\"o}lkopf, R.~Turner, and J.~Peters.
\newblock Causal transfer in machine learning.
\newblock \emph{Journal of Machine Learning Research}, 19\penalty0
  (36):\penalty0 1--34, 2018.

\bibitem[Rubin(2005)]{rubin2005causal}
D.B. Rubin.
\newblock Causal inference using potential outcomes.
\newblock \emph{Journal of the American Statistical Association}, 100:\penalty0
  322--331, 2005.

\bibitem[Sinha et~al.(2018)Sinha, Namkoong, and Duchi]{sinha2017certifiable}
A.~Sinha, H.~Namkoong, and J.~Duchi.
\newblock Certifying some distributional robustness with principled adversarial
  training.
\newblock In \emph{Sixth International Conference on Learning Representations
  (ICLR)}, 2018.

\bibitem[Spirtes et~al.(2000)Spirtes, Glymour, and Scheines]{Spirtes2000}
P.~Spirtes, C.~Glymour, and R.~Scheines.
\newblock \emph{Causation, prediction, and search}.
\newblock MIT Press, 2nd edition, 2000.

\bibitem[Stegle et~al.(2011)Stegle, Lippert, Mooij, Lawrence, and
  Borgwardt]{Steglenips}
O.~Stegle, C.~Lippert, J.~M. Mooij, N.~D. Lawrence, and K.~Borgwardt.
\newblock Efficient inference in matrix-variate gaussian models with iid
  observation noise.
\newblock In J.~Shawe-Taylor, R.~S. Zemel, P.~L. Bartlett, F.~Pereira, and
  K.~Q. Weinberger, editors, \emph{Advances in Neural Information Processing
  Systems 24}, pages 630--638. Curran Associates, Inc., 2011.

\bibitem[Stegle et~al.(2012)Stegle, Parts, Piipari, Winn, and
  Durbin]{stegle2012using}
O.~Stegle, L.~Parts, M.~Piipari, J.~Winn, and R.~Durbin.
\newblock Using probabilistic estimation of expression residuals ({PEER}) to
  obtain increased power and interpretability of gene expression analyses.
\newblock \emph{Nature protocols}, 7\penalty0 (3):\penalty0 500, 2012.

\bibitem[Theil(1958)]{theil1958economic}
H.~Theil.
\newblock \emph{Economic forecasts and policy}.
\newblock North-Holland, 1958.

\bibitem[Tian and Pearl(2001)]{Tian2001}
J.~Tian and J.~Pearl.
\newblock Causal discovery from changes.
\newblock In \emph{Proceedings of the 17th Conference on Uncertainty in
  Artificial Intelligence ({UAI})}, pages 512--522, 2001.

\bibitem[Tibshirani(1996)]{tibshirani1996regression}
R.~Tibshirani.
\newblock Regression shrinkage and selection via the {L}asso.
\newblock \emph{Journal of the Royal Statistical Society, Series B},
  58:\penalty0 267--288, 1996.

\bibitem[{van de Geer}(2016)]{van2016estimation}
S.~{van de Geer}.
\newblock \emph{Estimation and testing under sparsity}.
\newblock Springer, 2016.

\bibitem[Wright(1928)]{Wright1928}
P.G. Wright.
\newblock \emph{The tariff on animal and vegetable oils}.
\newblock The Macmillan company New York, 1928.

\bibitem[Xu et~al.(2009)Xu, Caramanis, and Mannor]{xu2009robust}
H.~Xu, C.~Caramanis, and S.~Mannor.
\newblock Robust regression and lasso.
\newblock In \emph{Advances in Neural Information Processing Systems}, pages
  1801--1808, 2009.

\bibitem[Yu and Kumbier(2020)]{yu2019three}
B.~Yu and K.~Kumbier.
\newblock Veridical data science.
\newblock \emph{Proceedings of the National Academy of Sciences}, 117\penalty0
  (8):\penalty0 3920--3929, 2020.

\end{thebibliography}

\newpage
\section{Appendix}

\subsection{Interpretation of the model class in the cyclic case}\label{sec:interpr-model-class}

If the graph $G$ is cyclic, then the model class in Section~\ref{sec:setting-notation} describes the distribution in an equilibrium state. To see this, let us write
\begin{equation*}
    V_0 = \varepsilon
\end{equation*}
and
\begin{equation*}
V_t = \mathbf{B} V_{t-1} + \varepsilon \text{ for all $t \ge 1$,}
\end{equation*}
where $V = (X,Y,H)^{\intercal}$. If the spectral norm of $ \mathbf{B} $ is strictly smaller than one, then for each $\varepsilon$ we have
\begin{equation*}
    \lim_{t \rightarrow \infty} V_t = \sum_{k \ge 0 } \mathbf{B}^k \varepsilon = (\mathrm{Id} - \mathbf{B})^{-1} \varepsilon.
\end{equation*}
Note that if $\mathbf{B}$ is acyclic then this limit always exists. Analogously one can define the shifted distribution as
\begin{align*}
V_0 &= \varepsilon + v \\
V_t &= (\mathrm{Id}+\mathbf{B}) V_{t-1} \\
\lim_{t \rightarrow \infty} V_t &= \sum_{k \ge 0}  \mathbf{B}^k (\varepsilon+v) = (\mathrm{Id} -\mathbf{B})^{-1} (\varepsilon + v)
\end{align*}
ence, by the definition of $V$, we have $V = \lim_{t \rightarrow \infty} V_t$ and our model describes the distribution of a cyclic causal model in its equilibrium.

\subsection{Sets $C^{\gamma}$ for three examples}\label{sec:three-examples}

In this section we discuss three examples to shed more light on Theorem~\ref{thm:anchor-regression} and the behaviour of \emph{anchor regression}. In particular, the sets $C^{\gamma}$ are discussed for the three simple examples. We will see that $C^{\gamma}$ can contain interventions not only on $X$ but potentially also on $Y$ and $H$. 
 The SEM and graph in each case are given in Example~\ref{fig:cgammaq}.
\begin{example}[Three SEMs and corresponding sets $C^{\gamma}$]\label{fig:cgammaq}
 In each of these SEMs for simplicity we assume that $\varepsilon \sim \mathcal{N}(0,\mathrm{Id}_{3})$ and $A \sim \mathcal{N}(0,1)$.
For (i), the corresponding SEM is $H \leftarrow  \varepsilon_{3}$, $X \leftarrow H + A + \varepsilon_{1}$, $Y \leftarrow 2H +X+ \varepsilon_{2}$. In this example, $C^{\gamma}$ contains interventions on $X$ up to strength $\gamma$, i.e.,
$C^{\gamma} = \{ (t,0,0)^{\intercal} : t^{2} \le \gamma \}$.
For (ii), the corresponding SEM is $H \leftarrow \varepsilon_{3}$, $X \leftarrow H+Y + \varepsilon_{1}$, $Y \leftarrow A+2H + \varepsilon_{2}$. In this example, $C^{\gamma}$ contains interventions on $Y$ up to strength $\gamma$, i.e., $C^{\gamma} = \{ (0,t,0)^{\intercal} : t^{2} \le \gamma \}$.
For (iii), the corresponding SEM is $H \leftarrow A + \varepsilon_{3}$, $X \leftarrow H + \varepsilon_{1}$, $Y \leftarrow 2H + X+ \varepsilon_{2}$. In this example, $C^{\gamma}$ contains interventions on $H$ up to strength $\sqrt{\gamma}$, i.e., $C^{\gamma} = \{ (0,0,t)^{\intercal} : t^{2} \le \gamma \}$.
$C^{\gamma}$ takes more complex forms when $A$ points to several variables. Examples of this phenomenon are discussed in Section~\ref{sec:full-rank-assumption}.\begin{multicols}{3}
\begin{center}
\begin{tikzpicture}[->,>=latex,shorten >=1pt,auto,node distance=1.2cm,
                    thick]
  \tikzstyle{every state}=[draw=black,text=black, inner sep=0.4pt, minimum size=17pt]

\node[state] (Y) {$Y$};
  \node[state] (H) [ left of=Y] {$H$};
  \node[state] (X) [below of=H] {$X$};
  \node[state] (A) [above left of = H] {$A$};

\path  (X)  edge node[below] {1} (Y);
\draw  (H)   edge node[above] {2}  (Y);
\draw  (H)  edge node[right] {1} (X);
\draw (A) edge  node[below left ] {1}(X);

\end{tikzpicture}
(i) \\
\vspace{0.5cm}
$C^{\gamma} = \{(t,0,0)^{\intercal} : t^{2} \le \gamma\}$
\end{center}

\columnbreak

\begin{center}
\begin{tikzpicture}[->,>=latex,shorten >=1pt,auto,node distance=1.2cm,
                    thick]
  \tikzstyle{every state}=[draw=black,text=black, inner sep=0.4pt, minimum size=17pt]

\node[state] (Y) {$Y$};
  \node[state] (H) [ left of=Y] {$H$};
  \node[state] (X) [below of=H] {$X$};
  \node[state] (A) [above left of = H] {$A$};

\path  (Y)  edge node[below] {1} (X);
\draw  (H)   edge node[below] {2}  (Y);
\draw  (H)  edge node[left] {1} (X);
\draw (A) edge  node[above right] {1}(Y);

\end{tikzpicture}
(ii) \\
\vspace{0.5cm}
$C^{\gamma} = \{(0,t,0)^{\intercal} : t^{2} \le \gamma\}$
\end{center}

\columnbreak

\begin{center}
\begin{tikzpicture}[->,>=latex,shorten >=1pt,auto,node distance=1.2cm,
                    thick]
  \tikzstyle{every state}=[draw=black,text=black, inner sep=0.4pt, minimum size=17pt]

\node[state] (Y) {$Y$};
  \node[state] (H) [ left of=Y] {$H$};
  \node[state] (X) [below of=H] {$X$};
  \node[state] (A) [above left of = H] {$A$};

\draw  (H)   edge node[above] {2}  (Y);
\draw  (H)  edge node[left] {1} (X);
\draw (A) edge  node[above right] {1} (H);
\draw (X) edge node[below right] {1}(Y);

\end{tikzpicture}
(iii) \\
\vspace{0.5cm}
$C^{\gamma} = \{(0,0,t)^{\intercal} :  t^{2} \le \gamma \}$
\end{center}
\end{multicols}
\end{example}
Example (i) corresponds to a classic IV setting. Here, we have $\mathbf{M}=(1,0,0)^{\intercal}$. Hence, $C^{\gamma}$ is the
set of interventions on $X$ up to ``strength'' $\gamma$,
i.e., $C^{\gamma}=\{(t,0,0)^{\intercal}: t^{2} \le \gamma \}$. By
Theorem~\ref{thm:anchor-regression}, $b^{\gamma}$ minimizes the
$\ell_{2}$-loss under shift interventions on $X$ up to ``strength''
$\gamma$. Similarly for example (ii): \emph{anchor regression} minimizes
the $\ell_{2}$-loss under interventions on $Y$. In example (iii),
\emph{anchor regression} minimizes the $\ell_{2}$-loss under
interventions on $H$. 
In the following we want to investigate  whether \emph{anchor regression} can achieve predictive stability, i.e., stable predictive performance in these SEMs under strong interventions.
 This question can be
answered by investigating the limit $b^{\rightarrow \infty}= \lim_{\gamma
  \rightarrow \infty} b^{\gamma}$.
  In example (i), we  obtain
$b^{\rightarrow \infty} = 1$. In example (ii), we have $b^{\rightarrow \infty} = 1$ and in example (iii) we have  $b^{\rightarrow \infty} = 3$.
A short calculation shows that the
distribution of $Y-X^\intercal b$ under $\mathbb{P}_{v}$ is invariant under shift interventions on
$X$. Formally,
\begin{equation*}
  Y-X^\intercal b^{\rightarrow \infty} \text{ under $\mathbb{P}_{v}$ has the same distribution for all }
  v=(t,0,0)^{\intercal}.
\end{equation*}
In particular, the MSE $\mathbb{E}_{v}[(Y-X^\intercal b)^{2}]$ is constant under shift interventions on $X$. Similarly in
example (ii), the distribution of $Y - X^\intercal b^{\rightarrow \infty}$ is invariant
under shift interventions on $Y$. And in example (iii), the distribution
of $Y - X^\intercal b^{\rightarrow \infty}$ is invariant under shift interventions
on $H$. This holds for any set of edge coefficients with one of the graph structures as in Example~\ref{fig:cgammaq}. However, for some graphs (for example for the graph that arises from reversing the edge between $X$ and $Y$ in (ii)), the invariance statement above does not hold.

In all examples, $A$ is correlated with $X$. Let $c_{x}$ denote the effect of $A$ on $X$, i.e.\ the regression coefficient when regressing $X$ on $A$. Let $c_{y}$ denote the effect of $A$ on $Y$. Thus, for $b = \frac{c_{y}}{c_{x}}$, the effect of $A$ on the synthetic variable $r = Y - Xb$ is zero and $r$ has invariant distribution under conditioning on $A$. Conditioning on $A$ can be interpreted as certain shift interventions on $(X,Y,H)$. This in turn implies the invariance properties discussed above. Thus, as long as the effect of a one-dimensional anchor variable $A$ on $X$ is non-zero, invariance is attainable.

Summarizing, in these examples, \emph{anchor regression} exhibits constant
predictive performance even under arbitrarily strong shift interventions. In Section~\ref{sec:invariance} we investigate the phenomenon of ``invariance under interventions''.

\subsection{Data-driven invariance}\label{sec:invariance}
In Section~\ref{sec:three-examples} we discussed three examples for which the distribution of $Y-X^\intercal b^{\rightarrow \infty}$ under $\mathbb{P}_{v}$ is invariant under certain shift interventions $v$. Here and in the following, we tacitly assume that the limit $b^{\rightarrow \infty} := \lim_{\gamma \rightarrow \infty} b^{\gamma}$ exists. 

We want to investigate the conditions under which we have invariance. Define $I := \{ b \in \mathbb{R}^{d} : \Etrain[A \cdot(Y-X^\intercal b)]=0\}$.
Then we have the following theorem.
\begin{theorem}\label{thm:invariance}
Assume that the Gram matrix $\mathbb{E}[A A^{\intercal}]$ is positive definite. Then,
\begin{equation*}
b \in I\qquad \iff \qquad  Y - X^\intercal b \text{ under $\mathbb{P}_{v}$ has the same distribution for all } v \in \mathrm{span}(\mathbf{M}).
\end{equation*}
\end{theorem}
Note that if the set $I$ is non-empty, then $b^{\rightarrow \infty}  \in
I$. Hence \emph{anchor regression} will have this invariance property for
$\gamma \rightarrow \infty$ if and only if there exists a $b$ that has this
property.

This invariance can be interpreted as follows. Assume we have an anchor $A \in \{-1,1\}$ that represents data collected from two environments. Data from environment $A=1$ and environment $A=-1$ differ by a shift intervention of $2\mathbf{M}$ on $(X,Y,H)$. The theorem above tells us that for all $b \in I$ the residual distribution (that means, the distribution of $Y-X^\intercal b$) is invariant under $\mathbb{P}_{v}$ with $v = \alpha \mathbf{M}, \alpha \in \mathbb{R}$. In this sense, \emph{anchor regression} with $\gamma \rightarrow \infty$ is invariant with respect to the heterogeneities that are observed in the training distribution (to be more precise, we obtain invariance of the residuals with respect to linear combinations of inhomogeneities in the training distribution, cf. Theorem~\ref{thm:invariance}).

 In the following discussion we will make an assumption that facilitates interpretation of the span of the shift matrix $\mathbf{M}$, i.e.\  of span($\mathbf{M}$). Define $T := \{k : \mathbf{M}_{ k, \bullet }\not \equiv 0\}$ as the rows of $\mathbf{M}$ that are not identically zero. In the following we will refer to $T$ as  \emph{children of $A$}.
Let  the Gram matrix of $(\mathbf{M}A)_{T}$ be positive definite. In the following, we will call this the \emph{full-rank assumption}. Then $\text{span}(\mathbf{M})=\{v : v_{-T} \equiv 0\}$, the set that contains arbitrary interventions on the children of $A$. In particular, for all $b \in \mathbb{R}^{d}$,
\begin{equation*}
  b \in I \qquad \iff \qquad  Y-X^\intercal b \text{ under $\mathbb{P}_{v}$ has the same distribution for all } v \in \mathbb{R}^{d+1+r} \text{ with } v_{-T} \equiv 0.
\end{equation*}
Hence the set $I = \{ b :  \Etrain[A \cdot(Y-X^\intercal b)] = 0 \}$ is exactly the
set of vectors $b$ for which $Y-X^\intercal b$ is invariant under
interventions on the children of $A$. This has consequences for the
interpretation of $b^{\rightarrow \infty} $. If $I$ is nonempty, i.e., if
invariance is attainable, then loosely speaking
\begin{align*}
  b^{\rightarrow \infty} = \argmin_{b} \Etrain[ ( Y-X^\intercal b )^{2}] &\text{ s.t. the distribution of } Y-X^\intercal b \text{ under $\mathbb{P}_{v}$ has invariant distribution  }\\
& \text{under shift interventions on the children of $A$.}
\end{align*}
Note that $\lim_{\gamma \rightarrow \infty }b^{\gamma}$ may exist, even in cases where $\mathbb{P}_{v}$ is not invariant under shift interventions $v \in C^{\gamma}$. Under the assumptions of Theorem~\ref{thm:anchor-regression}, we have
\begin{equation*}
 b^{\gamma} = \arg \min_{b} \sup_{v \in C^{\gamma}} \mathbb{E}_{v}[ (Y-X^\intercal b)^{2}  ].
\end{equation*}
Thus, if $b^{\gamma}$ converges for $\gamma \rightarrow \infty$, $b^{\rightarrow \infty}$ corresponds to the prediction rule that results in the least-growing worst-case prediction loss for $v \in C^{\gamma}$, $\gamma \rightarrow \infty$.

In the next discussion we will give two examples to shed some light on the full-rank assumption.

\subsection{Shape of $C^{\gamma}$}\label{sec:full-rank-assumption}
In the preceding section we saw examples where $A$ has only one child leading to very simple forms of $C^{\gamma}$. In this section we will discuss two slightly more involved examples, with two covariates $(X_{1},X_{2})$ and one hidden confounder $H$.
The examples are depicted in Figure~\ref{full-rank-graph}. Invariance in the sense of Theorem~\ref{thm:invariance} is only achievable for the graph on the right: It can be shown that $I = \emptyset$ for the graph on the left. 
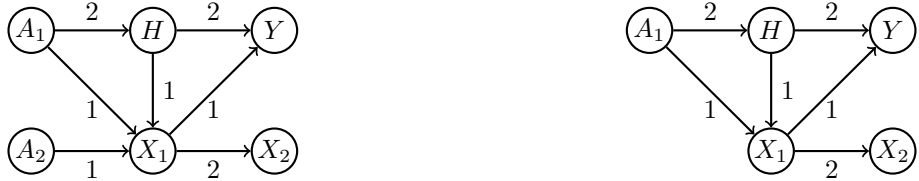
\begin{figure}
\begin{multicols}{2}
\begin{center}
\begin{tikzpicture}[->,node distance=1.6cm,thick]
  \tikzstyle{every state}=[draw=black,text=black, inner sep=0.4pt, minimum size=17pt]

  \node[state] (Y) {$Y$};
  \node[state] (H) [ left of=Y] {$H$};
  \node[state] (X1) [below of=H] {$X_{1}$};
  \node[state] (A) [left of = H] {$A_{1}$};
  \node[state] (X2) [below of = Y] {$X_{2}$};
  \node[state] (A2) [left of = X1] {$A_{2}$};

\draw  (X1)  edge node[below] {1} (Y);
\draw  (H)   edge node[above] {2}  (Y);
\draw  (H)  edge node[right] {1} (X1);
\draw (A) edge  node[below ] {1}(X1);
\draw (A2) edge node[below] {1}(X1);
\draw (A) edge node[above] {2} (H);
\draw (X1) edge node[below] {2} (X2);
\end{tikzpicture} \end{center}
\columnbreak
\begin{center}
\begin{tikzpicture}[->,node distance=1.6cm,thick]
  \tikzstyle{every state}=[draw=black,text=black, inner sep=0.4pt, minimum size=17pt]

  \node[state] (Y) {$Y$};
  \node[state] (H) [ left of=Y] {$H$};
  \node[state] (X1) [below of=H] {$X_{1}$};
  \node[state] (A) [left of = H] {$A_{1}$};
  \node[state] (X2) [below of = Y] {$X_{2}$};

\draw  (X1)  edge node[below] {1} (Y);
\draw  (H)   edge node[above] {2}  (Y);
\draw  (H)  edge node[right] {1} (X1);
\draw (A) edge  node[below ] {1}(X1);
\draw (A) edge node[above] {2} (H);
\draw (X1) edge node[below] {2} (X2);
\end{tikzpicture} \end{center}
\end{multicols}
\caption{In the example on the left, the full-rank assumption holds. In the example on the right, the full-rank assumption does not hold. The deterministic shifts in $C^{\gamma}$ for $\gamma=1$ are visualized in Figure~\ref{full-rank-cpgamma}. }\label{full-rank-graph}
\end{figure}

 In both examples, assume that $\Etrain[A A^{\intercal}] = \mathrm{Id}$. Then, in the example on the left, we have
\begin{equation*}
  \mathbf{M} = \begin{pmatrix}
    1 & 1 \\
    0 & 0 \\
    0 & 0 \\
    2 & 0
  \end{pmatrix}, \text{ and hence } C^{1} = \left\{ v \in \mathbb{R}^{4} \text{ such that } v_{2}=v_{3}=0 \text{ and } \left(v_{1}-\frac{v_{4}}{2} \right)^{2} + \frac{v_{4}^{2}}{4} \le 1  \right\}.
\end{equation*}
$v_{1}$ corresponds to interventions on $X_{1}$, whereas $v_{4}$ corresponds to interventions on $H$. The set $C^{1}$ is visualized in Figure~\ref{full-rank-cpgamma} on the left-hand side. As $v_{2}=v_{3}=0$ for all $v \in C^{1}$, only the dimensions $v_{1}$ and $v_{4}$ (interventions on $X_{1}$ and $H$) are shown.
The full-rank assumption holds, as
\begin{equation*}
  \mathbf{M}_{T,\bullet}= \begin{pmatrix}
    1 & 1 \\
    2 & 0
  \end{pmatrix} \text{ has full row-rank.}
\end{equation*}
 On the right-hand side the situation is different, as we only have one anchor.  Here,
\begin{equation*}
  \mathbf{M} = \begin{pmatrix}
    1  \\
    0  \\
    0  \\
    2
  \end{pmatrix}, \text{ hence } C^{1} = \left\{ v \in \mathbb{R}^{3} \text{ such that } v_{2}=v_{3}=0, 2v_{1}=v_{4}  \text{ and }   v_{1}^{2} \le 1  \right\}.
\end{equation*}
 Analogously as above, the deterministic shifts in $C^{1}$ are visualized in Figure~\ref{full-rank-cpgamma} on the right-hand side. The ellipsoid is degenerate and the full-rank assumption is not fullfilled as
\begin{equation*}
  \mathbf{M}_{T,\bullet} = \begin{pmatrix}
      1 \\
      2
  \end{pmatrix} \text{ does not have full row-rank.}
\end{equation*}
If $(\mathbf{M}A)_{T}$ is degenerate, then we observe shifts only in certain linear subspaces of $\mathbb{R}^{|T|}$ and \emph{anchor regression} optimizes the MSE only under these restricted interventions.
 It seems desirable to include as many anchors as possible to  optimize predictive performance under a wide range of interventions. However, this comes at a cost.  Adding anchors that correspond to shifts that will not occur in the test data set can result in overly conservative predictive performance.

\begin{figure}
\begin{multicols}{2}
\begin{center}

\includegraphics[scale=0.7]{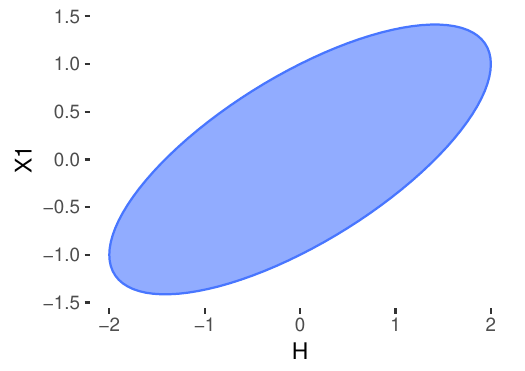}
\end{center}

\columnbreak
\begin{center}

\includegraphics[scale=0.7]{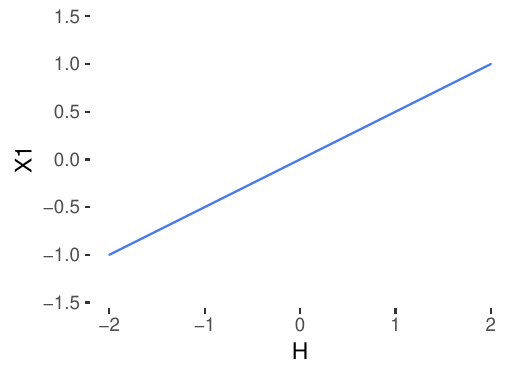}
\end{center}

\end{multicols}
\caption{The blue areas correspond to the interventions in $C^{1}$ for the examples in Figure~\ref{full-rank-graph}. On the left-hand side the full-rank assumption holds. Loosely speaking, for $\gamma \rightarrow \infty$ the ellipsoid grows larger and larger, eventually containing arbitrary shift interventions on $X_{1}$ and $H$. On the right-hand side, the full-rank assumption does not hold, hence $C^{\gamma}$ for $\gamma \rightarrow \infty$ only contains interventions on $X_{1}$ and $H$ that satisfy certain linear constraints.}\label{full-rank-cpgamma}
\end{figure}

\subsection{Theorem~\ref{thm:anchor-regression} for random shifts}\label{sec:theor-refthm:-regr}

\begin{theorem}\label{thm:anchor-regression-2}
For any $b \in \mathbb{R}^{d}$
we have
\begin{equation}
 \Etrain[((\mathrm{Id}- \mathrm{P}_{A})(Y-X^\intercal b))^{2}] + \gamma  \Etrain[( \mathrm{P}_{A}(Y-X^\intercal b))^{2}]  = \sup_{ \mathbb{P}_{v} \in C^{\gamma}} \mathbb{E}_{v}[ (Y-X^\intercal b)^{2}  ],
\end{equation}
where
\begin{align*}
  C^{\gamma} :=\{ & \text{probability measures } \mathbb{P}_{v} : \\
                  &\text{ the assumptions of Section~\ref{sec:setting-notation} are satisfied, and } \mathbb{E}_{v}[v v^{\intercal}] \preceq \gamma  \mathbf{M}\Etrain[A A^{\intercal}] \mathbf{M}^{\intercal}\}.
\end{align*}
\end{theorem}

\subsection{Proof of Theorem~\ref{thm:anchor-regression} and Theorem~\ref{thm:anchor-regression-2}}

\begin{proof}
  We will show Theorem~\ref{thm:anchor-regression}. The proof of Theorem~\ref{thm:anchor-regression-2} proceeds analogously. Using the model assumptions of Section~\ref{sec:setting-notation}, under $\mathbb{P}_{v}$,
\begin{equation*}
  Y-X^\intercal b = ( (\mathrm{Id}-\mathbf{B})_{d+1, \bullet}^{-1} - b^{\intercal}(\mathrm{Id}-\mathbf{B})_{1:d, \bullet}^{-1}) (\varepsilon+v).
\end{equation*}
In the following, for brevity we write $w = ( (\mathrm{Id}-\mathbf{B})_{d+1, \bullet}^{-1} - b^{\intercal}(\mathrm{Id}-\mathbf{B})_{1:d, \bullet}^{-1})^{\intercal}$.
As $\mathbb{E}_{v}[\varepsilon] =0$ and using that $\varepsilon$ and $v$ are uncorrelated under $\mathbb{P}_{v}$,
\begin{align*}
   \mathbb{E}_{v}[ (Y-X^\intercal b)^{2} ] &= \mathbb{E}_{0}[(Y-X^\intercal b)^{2}] +  \mathbb{E}_{v}[ (  w^{\intercal}v)^{2}].
\end{align*}
Taking the supremum over $C^{\gamma}$, using the definition of $C^{\gamma}$,
\begin{align}\label{eq:16}
\begin{split}
  \sup_{v \in C^{\gamma}} \mathbb{E}_{v}[ (Y-X^\intercal b)^{2} ] &= \mathbb{E}_{0}[(Y-X^\intercal b)^{2}]  + \sup_{v \in C^{\gamma}} \mathbb{E}_{v}[ (w^{\intercal} v)^{2}]\\
      &= \mathbb{E}_{0}[(Y-X^\intercal b)^{2}]  + \sup_{v \in C^{\gamma}} w^{\intercal}\mathbb{E}_{v}[ v v^{\intercal}] w\\
&=  \mathbb{E}_{0}[(Y-X^\intercal b)^{2}]  + \gamma  w^{\intercal} \mathbf{M} \Etrain[A A^{\intercal}] \mathbf{M}^{\intercal} w \\
&=  \mathbb{E}_{0}[(Y-X^\intercal b)^{2}]  + \gamma   \Etrain [(w^{\intercal} \mathbf{M} A)^{2}] \\
\end{split}
\end{align}
By the model assumptions of Section~\ref{sec:setting-notation}, $\varepsilon$ and $A$ are independent and $\Etrain[\varepsilon]=0$, which together with the definition of $w$ implies that under $\Ptrain$,
\begin{align}\label{eq:34}
  \begin{split}
  \Etrain[Y-X^\intercal b|A] &= \Etrain[w^{\intercal}(\varepsilon+\mathbf{M}A)|A]= w^{\intercal} \mathbf{M} A, \text{ and } \\
  Y-X^\intercal b - \Etrain[Y-X^\intercal b |A] &= w^{\intercal} (\varepsilon + \mathbf{M}A) - w^{\intercal} \mathbf{M} A = w^{\intercal} \varepsilon.
  \end{split}
\end{align}
Note that by definition under $\mathbb{P}_{0}$, $Y-X^\intercal b$ has the same distribution as $w^{\intercal} \varepsilon$ under $\Ptrain$. Hence, under $\Ptrain$, $Y-X^\intercal b - \Etrain[Y-X^\intercal b |A]$ has the same distribution as $Y-X^\intercal b$ under $\mathbb{P}_{0}$. Thus, using the equations~\eqref{eq:34} in equation~\eqref{eq:16} yields
\begin{equation*}
\sup_{v \in C^{\gamma}} \mathbb{E}_{v}[ (Y-X^\intercal b)^{2} ] = \Etrain[(Y-X^\intercal b - \Etrain[Y-X^\intercal b |A])^{2}] + \gamma \Etrain[(\Etrain[Y-X^\intercal b|A])^{2}],
\end{equation*}
which concludes the proof.
\end{proof}

\subsection{Proof of Lemma~\ref{lemma:interpr-anch-regr}}
\begin{proof}
We can rewrite $\mathbb{E}[(Y-X^\intercal b)^{2} |A]$:
\begin{align*}
\mathbb{E}[(Y-X^\intercal b)^{2} |A] = \mathbb{E}[(Y-X^\intercal b - \mathbb{E}[Y-X^\intercal b |A])^{2} |A ] + (\mathbb{E}[Y-X^\intercal b|A])^{2}
\end{align*}
As $(X,Y,A)$ follows a centered multivariate Gaussian distribution,
\begin{align*}
  \mathbb{E}[Y-X^\intercal b|A] &\sim \mathcal{N}(0,\mathbb{E}[(\mathbb{E}[Y-X^\intercal b|A])^{2}]) \\
\text{ and } \mathbb{E}[(Y-X^\intercal b - \mathbb{E}[Y-X^\intercal b |A])^{2} |A ] &= \mathbb{E}[(Y-X^\intercal b - \mathbb{E}[Y-X^\intercal b |A])^{2}  ].
\end{align*}
Hence the $\alpha$-th quantile of  $\mathbb{E}[(Y-X^\intercal b)^{2} |A]$ is equal to
\begin{equation*}
  \mathbb{E}[(Y-X^\intercal b - \mathbb{E}[Y-X^\intercal b |A])^{2} ] + \chi^{2}_{1}(\alpha) \mathbb{E}[(\mathbb{E}[Y-X^\intercal b|A])^{2}],
\end{equation*}
where $\chi_{1}^{2}(\alpha)$ denotes the $\alpha$-th quantile of a $\chi^{2}$-distributed random variable with one degree of freedom.
\end{proof}

\subsection{Lemma~\ref{lemma:interpr-anch-regr} for discrete anchors}\label{sec:lemma-refl-anch}

\begin{lemma}[Version of Lemma~\ref{lemma:interpr-anch-regr} for discrete anchors]\label{lemma:interpr-anch-regr-discr}
Assume that we have several training data sets $a \in \mathcal{A}$ and one test dataset. For each $a \in \mathcal{A}$, the data are drawn i.i.d. from $\mathbb{P}^a = \mathbb{P}[\bullet|A=a]$. On the test data set, the data are drawn i.i.d. from the distribution of $\mathbb{P}^\text{test}$. We write $\mathbb{E}^a[\bullet]$ to denote the expectation on data set $a \in \mathcal{A}$ and $\mathbb{E}^\text{test}$ to denote the expectation on the test data set. We assume that the data sets differ by a shift $\delta^a := \mathbb{E}^a[(X,Y)]$, i.e., that $(X,Y) - \mathbb{E}^a[(X,Y)]$ under $\mathbb{P}^a$ has the same distribution as $(X,Y) - \mathbb{E}^{a'}[(X,Y)]$ under $\mathbb{P}^{a'}$ for all $a \in \mathcal{A} \cup \{ \text{new}\}$. We assume that the shift $\delta^a$ is constant on each data set, but random between the data sets, with distribution $\delta^a \sim \mathcal{N}(0,\Sigma)$ for some positive semi-definite $\Sigma$.  Write $\mathbb{E}^{a,\delta}$ for the expectation both with respect to the randomness of $\mathbb{P}^a$ and the randomness of the shift $\delta^a$. Due to the randomness of $\delta^\text{new}$, the risk $\mathbb{E}^\text{new}[(Y-X^\intercal b)^2]$ is random and we write $Q(\alpha)$ for the quantiles of the risk on the new data set. Then,
\begin{equation}
  Q(\alpha) = \frac{1}{|\mathcal{A}|}  \sum_{a \in \mathcal{A}} \mathbb{E}^{a,\delta}[((\mathrm{Id}- \mathbb{E}^a)[Y-X^\intercal b])^{2}] +   \frac{\gamma}{|\mathcal{A}|} \sum_{a \in \mathcal{A}} \mathbb{E}^{a,\delta}[(\mathbb{E}^a[Y-X^\intercal b])^2],
\end{equation}
where $\gamma$
equals the $\alpha$-th quantile of a $\chi^{2}$-distributed random variable with one degree of freedom.
\end{lemma}
Note that with $\Etrain \equiv \frac{1}{|\mathcal{A}|} \sum_{a \in \mathcal{A}} \mathbb{E}^{a,\delta}$ and $\mathbb{E}^a \equiv \Etrain[\bullet | A=a]$ the right-hand side coincides with the anchor risk if each level of $A$ has equal probability (or under a re-weighting such that each level of $A$ has equal probability).

\begin{proof}
  First, note that using a bias-variance decomposition we can rewrite the risk on the new data set as
  \begin{equation*}
    \mathbb{E}^{\text{new}}[(Y-X^\intercal b)^2] = \mathbb{E}^\text{new}[(Y - \delta_{p+1}^\text{new}-(X - \delta_{1:p}^\text{new})^{\intercal} b)^2] + (\delta_{p+1}^\text{new} - (\delta_{1:p}^\text{new})^{\intercal} b)^2
  \end{equation*}
  By assumption, the first part of this term is constant and the second part is a centered Gaussian random variable. Hence, the $\alpha$-Quantile of this term is
  \begin{equation}\label{eq:discrete_decomp}
    Q(\alpha) = \mathbb{E}^\text{new}[(Y - \delta_{p+1}^\text{new}-(X - \delta_{1:p}^\text{new})^{\intercal} b)^2] + \gamma \mathbb{E}^{\text{new},\delta}[(\delta_{p+1}^\text{new} - (\delta_{1:p}^\text{new})^{\intercal} b)^2]
  \end{equation}
  Now using that the distribution of $(X,Y) - \delta^\text{new}$ under $\mathbb{P}^\text{new}$ is the same as the distribution of $(X,Y) - \delta^\text{new}$ under $\mathbb{P}^a$, we obtain for all $a \in \mathcal{A}$
  \begin{equation}\label{eq:firstparteq}
    \mathbb{E}^\text{new}[(Y - \delta_{p+1}^\text{new}-(X - \delta_{1:p}^\text{new})^{\intercal} b)^2] = \mathbb{E}^{\text{a},\delta}[(Y - \delta_{p+1}^\text{a}-(X - \delta_{1:p}^\text{a})^{\intercal} b)^2].
  \end{equation}
  Using that the $\delta^a$ are i.i.d.,
\begin{equation}\label{eq:secondparteq}
  \mathbb{E}^{\text{new},\delta}[(\delta_{p+1}^\text{new} - (\delta_{1:p}^\text{new})^{\intercal} b)^2] =  \mathbb{E}^{\text{a},\delta}[(\delta_{p+1}^\text{a} - (\delta_{1:p}^\text{a})^{\intercal} b)^2].
\end{equation}
Using equation~\eqref{eq:firstparteq} and equation~\eqref{eq:secondparteq} in equation~\eqref{eq:discrete_decomp} completes the proof.
\end{proof}

  \subsection{Proof of Lemma~\ref{lemma:project}}

  \begin{proof}

    Due to linearity of the model in \eqref{eq:32}, $\Etrain[Y|A]$ and
    $\Etrain[X|A]$ are linear functions of $A$ whose coefficients are
    given by the least squares principle. That is:
    \begin{eqnarray*}
     & &\mathbb{E}_{\text{train}}[Y|A] = 
         A^{\intercal} \mathbb{E}_{\text{train}}[A A^{\intercal}]^{-1}
         \mathbb{E}_{\text{train}}[A Y],\\
      & &\mathbb{E}_{\text{train}}[X^{\intercal}|A] = A^{\intercal}
          \mathbb{E}_{\text{train}}[A A^{\intercal}]^{-1}
          \mathbb{E}_{\text{train}}[A X^{\intercal}]. 
    \end{eqnarray*}
    Thus, $\mathbb{E}_{\text{train}}[Y - X^{\intercal} b|A] = 0$ if and only if
    \begin{eqnarray}\label{proof-add1}
\text{Cov}_{\text{train}}(A,X) b = \text{Cov}_{\text{train}}(A,Y),
    \end{eqnarray}
    where we used that $\mathbb{E}_{\text{train}} [A A^T]$ is
    invertible and covariances replace expectations since $X$ and $Y$ are
    assumed to have mean zero. 
    Equation \eqref{proof-add1} is a linear system of
      equations in the variables $b$: by the Rouch\'e-Capelli theorem, it
      has a solution if and only if
      \begin{eqnarray*}
        \text{rank}(\text{Cov}_{\text{train}}(A,X)) =
        \text{rank}\left(\text{Cov}_{\text{train}}(A,X) |
        \text{Cov}_{\text{train}}(A,Y)\right),
      \end{eqnarray*}
      which completes the proof.
      \end{proof}

  \subsection{Proof of Theorem~\ref{thm:repbinfty}}

\begin{proof}

Due to the projectability condition in \eqref{def:projectability} and Lemma
\ref{lemma:project} we know that $I \neq \emptyset$. The projectability
condition also holds for $X',Y',A'$ since one can verify that the rank
condition only depends on $\mathbf{B}$ and $\mathbf{M}$. Thus, we also have that $I'
\neq \emptyset$.

\medskip\noindent
(i) Characterization of $b^{\rightarrow \infty}$.\\
We first consider the residual term
\begin{eqnarray*}
  \eta_b = Y - X^\intercal b
\end{eqnarray*}
for any $b \in I$. Analogously as in the proof of Theorem
\ref{thm:anchor-stability-pred-stab} consider
\begin{eqnarray*}
  w_b = \left( (\mathrm{Id} - \mathbf{B})_{d+1,\bullet}^{-1} - b^\intercal (\mathrm{Id} -
  \mathbf{B})_{1:d,\bullet}^{-1} \right)^\intercal.
\end{eqnarray*}
We then have that $\eta_b = w_b^\intercal (\varepsilon + \mathbf{M} (\kappa A + \xi))$. Since $b \in I$,
we have that $\Etrain[\eta_b|A] = 0$ and therefore $\Etrain[\eta_b A^\intercal] =
\Etrain[w_b^\intercal (\mathbf{M} \kappa A A^\intercal)] = 0$ (where we used in the first equality
relation that  $\Etrain[\xi] =
\Etrain[\varepsilon] = 0$ and $\xi, A, \varepsilon$ are jointly
independent). Since $\Etrain[A A^\intercal]$ is invertible we obtain
\begin{eqnarray}\label{f1}
  w_b^\intercal \mathbf{M} = 0\ \forall\ b \in I.
\end{eqnarray}
Therefore $\eta_b = w_b^\intercal \varepsilon$ and thus we have:
\begin{eqnarray}\label{f3}
  b^{\rightarrow \infty} = \mbox{argmin}_{b \in I} \Etrain[\eta_b^2] =
  w_b^\intercal \Sigma_{\varepsilon} w_b,
  \end{eqnarray}
where $\Sigma_{\varepsilon} = \text{Cov}(\varepsilon)$.

\medskip\noindent
(i) Characterization of $b'^{\rightarrow \infty}$.\\
One can derive exactly along the same lines as above the analogue of
\eqref{f1}:
\begin{eqnarray}\label{f4}
  w_b^\intercal \mathbf{M} = 0\ \forall b \in I'.
\end{eqnarray}

\medskip\noindent
Because of \eqref{f4} and using \eqref{cov-add} we have that
\begin{eqnarray}\label{f5}
  b'^{\rightarrow \infty} &=& \mbox{argmin}_{b \in I'} \mathbb{E}[(Y' - (X')^{\intercal} b)^2]
                  \nonumber \\
  &=& \mbox{argmin}_{b \in I'} \mathbb{E}[(\eta'_b)^2] =
      \mbox{argmin}_{b \in I'} L w_b^\intercal \Sigma_{\varepsilon} w_b.
      \end{eqnarray}
This is to be compared with \eqref{f3}.

\medskip\noindent
It remains to show that
\begin{eqnarray}\label{f6}
  I = I'.
\end{eqnarray}
``$\subseteq$'': take $b \in I$. Then, by \eqref{f1}, $w_b^\intercal \mathbf{M} = 0$. Therefore
\begin{eqnarray*}
  \mathbb{E}_{\text{test}}[\eta'_b |A'] = \mathbb{E}[w_b^\intercal \varepsilon'|A'] = 0
\end{eqnarray*}
where the last equality follows by independence of $\varepsilon'$ and $A'$ and $\mathbb{E}_{\text{test}}[\varepsilon'] =
0$. Thus, $b \in I'$. \\
``$\supseteq$'': take $b \in I'$. Then, by \eqref{f4}, $w_b^\intercal \mathbf{M} = 0$. Therefore
\begin{eqnarray*}
  \Etrain[\eta_b |A] = \Etrain[w_b^\intercal \varepsilon|A] = 0
\end{eqnarray*}
where the last equality follows by independence of $\varepsilon$ and $A$
and $\Etrain[\varepsilon] = 0$. Thus, $b \in I$. \\
These two relations prove \eqref{f6}.

\medskip\noindent
By \eqref{f3}, \eqref{f5} and \eqref{f6}, we complete the proof of the
theorem.
\end{proof}

\subsection{Proof of Proposition~\ref{prop:constant}}

\begin{proof}
  Define $f(b) := \Etrain[ (\mathrm{P}_{A}(Y-X^\intercal b))^{2}]$ and $g(b) :=  \Etrain[ ((\mathrm{Id} - \mathrm{P}_{A})(Y-X^\intercal b))^{2}]$. By assumption, $\partial_{b} f(b^{0}) = \partial_{b} f(b^{\infty}) = \partial_{b} g(b^{0}) = \partial_{b} g(b^{\infty}) = 0$. The objective functional of anchor regression for a fixed value of $\gamma \ge 0$ is $g(b) + \gamma f(b)$. Hence also the derivative of the objective functional at $b^{0}$ is zero. As the objective functional $g(b) + \gamma f(b)$ is convex in $b$, $b=b^{0}$ is a minimizer of the objective functional. This completes the proof.
\end{proof}

\subsection{Proof of Theorem~\ref{thm:anchor-stability-pred-stab}}

\begin{proof}
  Define $\eta= Y - X^\intercal b^0$. As $b^0=b^{\rightarrow \infty}$, using equation~\eqref{eq:reachzero},
$\Etrain[\eta|A] =  0$.
This implies that \begin{equation*}
    \Etrain[\eta \cdot A ] = \Etrain[ \Etrain[\eta|A] \cdot A] = 0. \end{equation*}
Define $w = ( (\mathrm{Id}-\mathbf{B})_{d+1, \bullet}^{-1} -
(b^0)^{\intercal}(\mathrm{Id}-\mathbf{B})_{1:d,
 \bullet}^{-1})^{\intercal}$. By using the model assumptions, under $\Ptrain$,
\begin{equation*}\eta = w^\intercal (\varepsilon + \mathbf{M} A)
\end{equation*}
Using $\Etrain[w^\intercal \mathbf{M} A A^t] = \Etrain[\eta \cdot A] = 0$ and that $\Etrain[A A^t]$
is invertible, we have $w^\intercal
\mathbf{M} = 0$. Now, let $v$ be a random variable uncorrelated of $\varepsilon$ that takes values in $ \text{span}(\mathbf{M})$. As $w^{\intercal} \mathbf{M} =0$, $w^\intercal v = 0$. Thus, under $\mathbb{P}_{v}$,
\begin{equation*}Y - X^\intercal b^0 = w^\intercal (\varepsilon + v) = w^\intercal \varepsilon
 \end{equation*}
 Hence $Y-X^\intercal b^{0}$ has the same distribution under $\mathbb{P}_{v}$ as under $\mathbb{P}_{0}$. Thus, for all $b$ we have
\begin{equation*}
  \mathbb{E}_{v}[(Y - X^\intercal b)^2] \ge \mathbb{E}_{0}[(Y - X^\intercal b)^2] \ge \mathbb{E}_{0}[(Y - X^\intercal b^0)^2] =  \mathbb{E}_{v}[(Y - X^\intercal b^0)^2].
\end{equation*}
In the first step we used that $v$ is uncorrelated of $\varepsilon$ and equation~\eqref{eq:def-perturbed-distr}. In the second step we used the definition of $b^0$. In the third step we used that $Y-X^\intercal b^{0}$ has the same distribution under $\mathbb{P}_{v}$ as under $\mathbb{P}_{0}$.
Thus,
\begin{equation*}
  b^0 \in \argmin \mathbb{E}_{v}[(Y - X^\intercal b)^2].
\end{equation*}
This completes the proof. \end{proof}

\subsection{Generalized version of Theorem~\ref{thm:stabcaus}}\label{sec:gener-vers-theor}

In the following we will relax the assumptions to allow for endogeneous anchors.
\subsubsection{Relaxed anchor assumptions}\label{sec:relax-anch-assumpt}
 Let the distribution of $(X,Y,H,A)$ under $\Ptrain$ be a solution of the SEM
\begin{equation}\label{eq:35}
\begin{pmatrix}X \\ Y \\ H  \\ A \end{pmatrix} = \mathbf{B} \cdot \begin{pmatrix}X \\ Y \\ H \\ A \end{pmatrix} + \varepsilon.
\end{equation}
where 
$\mathbf{B} \in \mathbb{R}^{(d+1+r + q) \times (d+1+r + q)}$ is an unknown
constant matrix and the covariates $X$, the anchors $A \in \mathbb{R}^q$,
the hidden variables $H \in \mathbb{R}^r$, and the noise $\varepsilon \in
\mathbb{R}^{d+1+r}$ are random vectors. We assume that under $\Ptrain$, $X$ and $Y$
are centered to mean zero, that $\varepsilon$ and $A$ have finite second
moments and that the components of $\varepsilon$ are independent of each
other. To make the distribution of $X,Y,H,A$ well-defined, in the following
we assume that $\mathrm{Id}-\mathbf{B}$ is invertible. The model induces a directed graph $G$, with the edges given by the
following construction: For every $\mathbf{M}_{k,l} \neq 0$, a directed
edge is drawn from $A_{l}$ to the $k$-th variable in the $(d+1+r +
q)$-dimensional vector $(X,Y,H,A)$. Analogously, for every
$\mathbf{B}_{k,l} \neq 0$, a directed edge is drawn from the $l$-th
variable in $(X,Y,H,A)$ to the $k$-th variable in $(X,Y,H,A)$.

\begin{theorem}[Anchor stability implies causality] \label{thm:stabcaus-general}
    Let the assumptions of Section~\ref{sec:relax-anch-assumpt} hold with
    an acyclic graph $G$, and assume the projectability
    condition~\eqref{def:projectability}. 
 Furthermore, assume that for every disjoint sets of variables $V_1, V_2, V_3
      \subset (X,Y,H,A)$,  $V_1$ is d-separated of $V_2$ in $G$
      given $V_3$ if and only if the partial correlation
      $\text{part.cor}(V_1,V_2|V_3) = 0$. Furthermore assume that for each
  $X_k$ there exists $k'$ such that $A_{k'} \rightarrow X_k$.
    If $b^{\rightarrow \infty}
  = b^0$, then
    \begin{equation}
      b^{\rightarrow \infty} = b^0 = \partial_x \mathbb{E}[Y|do(X=x)],
    \end{equation}
    where the do-operator $\mathbb{E}[\bullet |do(X=x)]$ is defined as in \citet{Pearl2009}, Chapter~1. In addition, there is no confounder between $X$ and $Y$, i.e., there is no $H_{k}$ that is both an ancestor of some $X_{k'}$ and $Y$ in $G$.
  \end{theorem}

\subsection{Proof of Theorem~\ref{thm:stabcaus} and Theorem~\ref{thm:stabcaus-general}}
\begin{proof}
The proof for both theorems proceeds analogously. In the following, the covariances and partial correlations are meant with respect to the measure $\Ptrain$. Define $\eta= Y - X^\intercal b^0$. As $b^0=b^{\rightarrow \infty}$, using equation~\eqref{eq:reachzero},
$\Etrain[\eta|A] =  0$.
This implies that
\begin{equation*}
    \Etrain[\eta \cdot A ] = \Etrain[ \Etrain[\eta|A] \cdot A] = 0. \end{equation*}
  As $\Etrain[\eta \cdot A] = 0$ and $\eta$ is centered, we have $\text{Cov}(Y-X^\intercal b^0,A)=0$. Using Proposition~\ref{prop:constant}, $b^{1}=b^{0}$. Let us write $b'$ for the linear regression coefficient of $A$ on $X$. We have $0=\text{Cov}(Y - X^\intercal b^1, A) =  \text{Cov}(Y - X^\intercal b^1, A - X^\intercal b') $. Thus, by the definition of partial correlation, $\text{part.cor}(Y,A|X) = 0$. By assumption this implies that $Y$ and $A$ are d-separated given $X$ in $G$.

  We want to show that every backdoor path from $X$ to $Y$ is blocked given the empty set. If we can show this, by the Backdoor-Criterion \citep{Pearl2009}, due to linearity, $b^1$ is equal to the causal effect $\partial_{x} \mathbb{E}[Y|do(X=x)]$. As we showed that $b^1=b^0$, this would imply the claim of the theorem. Hence it suffices to show that every backdoor path from $X$ to $Y$ is blocked given the empty set.\\
\textbf{Step 1:} First, we note that no descendant of $Y$ can be in $X$. We will prove this by contradiction. Choose $k$ such that $X_k$ is a descendant of $Y$ and maximal in the sense that no other $X_{k'}$, $k' \neq k$ is a descendant of $Y$ and an ancestor of $X_k$.  By construction, there exists a directed path $ X_k \leftarrow \ldots \leftarrow Y$ such that the nodes on this path do not lie in $X$.  The nodes on this path do also not lie in $A$ as $A$ is d-separated of $Y$ given $X$. By assumption there exists a $k'$ such that $A_{k'} \rightarrow X_k$. Hence there exists a path $A_{k'} \rightarrow X_k \leftarrow \ldots \leftarrow Y$ that is open given $X$. Hence $A$ is not d-separated of $Y$ given $X$, contradiction. \\ 
\textbf{Step 2:} Assume there exists a backdoor path from $X$ to $Y$
that is open given the empty set. As the path from $X$ to $Y$ is open
given the empty set, it cannot contain a collider. Let this path
starts at $X_k$.\\ 
We have shown that the path does not contain a collider, and by Step~1, $X_k$ is not a descendant of $Y$. As the backdoor path is open given the empty set, it must be of the form $X_k \leftarrow ... \leftarrow Z \rightarrow \ldots \rightarrow Y$ and the nodes on the path do not lie in $X$. 
No node of the path lies in $A$ as we showed that $A$ is d-separated of $Y$ given $X$. To sum it up, we can assume that no node on the path lies in $A$ or $X$. However, we assumed that there exists $k'$ such that $A_{k'} \rightarrow X_k$. This gives us a path $A_{k'} \rightarrow X_k \leftarrow \ldots \leftarrow Z \rightarrow \ldots \rightarrow Y$ from $A_{k'}$ to $Y$ that is open given $X$. Contradiction! Hence, every backdoor path from $X$ to $Y$ is blocked given the empty set. By the Backdoor-Criterion \citep{Pearl2009}, due to linearity, $b^1$ is equal to the causal effect $\partial_{x} \mathbb{E}[Y|do(X=x)]$. As we showed that $b^1=b^0$, the claim of the theorem follows.
\end{proof}

\subsection{Proof of Theorem~\ref{thm:finite-sample-bound-1} and auxiliary results}

\textbf{Notation.} Define the ``residuals''  $\Z^{(a)} := \Y^{(a)} - \X^{(a)} b^{\gamma}$ for all $a \in \mathcal{A}$. We write $\overline{\X}^{(a)}$ for the empirical mean of $\X^{a}$, i.e., $\overline{\X}^{(a)} := \frac{1}{n_{a}} \sum_{i=1}^{n_{a}}   \X_{i,\bullet}^{(a)}$. Analogously define $\overline{\Y}^{(a)} :=  \frac{1}{n_{a}} \sum_{i=1}^{n_{a}}   \Y_{i}^{(a)}$ and $\overline{\Z}^{(a)} :=  \frac{1}{n_{a}} \sum_{i=1}^{n_{a}}   \Z_{i}^{(a)}$. Additionally define the conditional means $\mu_{\X}^{(a)} := \Etrain[X|A=a]$, $\mu_{\Y}^{(a)} := \Etrain[Y|A=a]$ and  $\mu_{\Z}^{(a)} := \Etrain[Y - X^\intercal b^{\gamma}|A=a]$ for $a \in \mathcal{A}$.

\subsubsection{Proof of Theorem~\ref{thm:finite-sample-bound-1}}\label{sec:proof-theor-refthm:f}
\begin{proof}
\textbf{Preliminaries.}
We want to derive bounds for
\begin{equation*}
R(\hat b) - \min_{b} R(b),
\end{equation*}
where
\begin{equation*}
R(b)  = \Etrain[( Y - \Etrain[Y | A] - (X - \Etrain[X | A])^{\intercal} b )^{2}] + \frac{\gamma}{| \mathcal{A}|} \sum_{a \in \mathcal{A}} (\Etrain[Y|A=a] - \Etrain[X|A=a]^{\intercal} b)^{2}.
\end{equation*}
and
\begin{align*}
\begin{split}
\hat b = \, & \argmin_{b} \frac{1}{| \mathcal{A} |} \sum_{a \in \mathcal{A}} \frac{1}{n_{a}}  \sum_{i=1}^{n_{a}} \left( \Y^{(a)}_{i} - \overline{\Y}^{(a)} - (\X^{(a)}_{i,\bullet} - \overline{\X}^{(a)} ) b \right)^{2} +\\
& \,\qquad \qquad  \frac{\gamma}{| \mathcal{A}|} \sum_{a \in \mathcal{A}}  \left( \overline{\Y}^{(a)} - \overline{\X}^{(a)} b \right)^{2} + 2 \lambda  \| b \|_{1}.
\end{split}
\end{align*}
Using the assumptions of Section~\ref{sec:setting-notation}, $(Y,X)$ has the same distribution under $\mathbb{P}_{0}$ as  $(Y- \Etrain[Y|A],X - \Etrain[X|A])$ under $\Ptrain$. Hence with $\mu_{\X}^{(a)} = \Etrain[X|A=a]$ and $\mu_{\Y}^{(a)} = \Etrain[Y|A=a]$ we can rewrite the risk as
\begin{equation*}
R(b) = \mathbb{E}_{0}[(Y-X^\intercal b)^{2}] +  \frac{\gamma}{| \mathcal{A} |} \sum_{a \in \mathcal{A}} ( \mu_{\Y}^{(a)} - (\mu_{\X}^{(a)})^{\intercal} b)^{2}.
\end{equation*}
\textbf{Step 1: rewriting the excess risk.}
By definition we have $b^{\gamma} = \argmin_{b} R(b)$. For all $b \in \mathbb{R}^{d}$, all $\lambda \ge 0$, and all $\gamma \ge 0$ we thus have the decomposition
\begin{align*}
\mathbb{E}_{0}[(Y-X^\intercal b)^{2}] +  \frac{\gamma}{| \mathcal{A} |} \sum_{a \in \mathcal{A}} ( \mu_{\Y}^{(a)} - (\mu_{\X}^{(a)})^{\intercal} b)^{2} = \,  & \mathbb{E}_{0}[(X (b-b^{\gamma}))^{2}] +  \frac{\gamma}{| \mathcal{A} |} \sum_{a \in \mathcal{A}} ( (\mu_{\X}^{(a)})^{\intercal} (b-b^{\gamma}))^{2} \\
 &+ \mathbb{E}_{0}[(Y-X^\intercal b^{\gamma})^{2}] +  \frac{\gamma}{| \mathcal{A} |} \sum_{a \in \mathcal{A}} ( \mu_{\Y}^{(a)} - (\mu_{\X}^{(a)})^{\intercal} b^{\gamma})^{2}.
\end{align*}
Hence if we write $W(b) =  \mathbb{E}_{0}[(X^{\intercal} (b-b^{\gamma}))^{2}] +  \frac{\gamma}{| \mathcal{A} |} \sum_{a \in \mathcal{A}} ( (\mu_{\X}^{(a)})^{\intercal} (b-b^{\gamma}))^{2}$, we can rewrite the excess risk as
\begin{equation*}
  R(\hat b) - \min_{b} R(b)  = W(\hat b).
\end{equation*}
We want to relate this excess risk to the empirical excess risk. Define
\begin{equation*}
  \hat W(b) := \frac{1}{| \mathcal{A}|} \sum_{a \in \mathcal{A}} \frac{1}{n_{a}} \sum_{i=1}^{n_{a}}  \left( (\X_{i,\bullet}^{(a)} - \overline{\X}^{(a)})  (b- b^{\gamma}) \right)^{2}  + \frac{\gamma}{| \mathcal{A}|} \sum_{a \in \mathcal{A}} ( \overline{\X}^{(a)} (b-b^{\gamma}))^{2}.
\end{equation*}
As $(X,Y,H)^{\intercal} = (\mathrm{Id} - \mathbf{B})^{-1} \varepsilon$ under $\mathbb{P}_{0}$ and $\varepsilon$ follows a centered multivariate Gaussian distribution, under $\mathbb{P}_{0}$, $(X,Y)$ follows a centered multivariate Gaussian distribution as well. Recall that in the proof of Theorem~\ref{thm:anchor-regression} we have shown that the distribution of $(Y,X)$ under $\mathbb{P}_{0}$ is the same as $(Y-\mathbb{E}[Y|A],X - \mathbb{E}[X|A]) = (Y-\mu_{\Y}^{A},X - \mu_{\X}^{A})$ under $\Ptrain$. As $A$ and $\varepsilon$ are independent, the distribution of $(Y,X)|(A=a)$ under $\Ptrain$ is the same as the distribution of $ (Y + \mu_{\Y}^{(a)},X + \mu_{X}^{(a)}) $ under $\mathbb{P}_{0}$. Using Lemma~\ref{lemma-new:excess-risk}, we obtain that with probability exceeding $1-4\exp(-t)$,
\begin{equation}\label{eq:new-18}
  W(\hat b) \le \frac{C''}{|S^{*}|} \| \hat b - b^{\gamma} \|_{1}^{2} + \hat W(\hat b),
\end{equation}
where $C''$ is a constant that depends on $c'$, $ \max_{k} \mathrm{Var}(X^{0}_{k})$, $\max_{a \in \mathcal{A}} \| \mu_{\X}^{(a)} \|_{\infty}$  and $\gamma$. \\
\textbf{Step 2: bounds for empirical excess risk $\hat W(\hat b)$ and $\| \hat b - b^{\gamma} \|_1$.}
It turns out that it is straightforward to derive finite-sample bounds for $\hat W(\hat b)$ and $\| \hat b - b^{\gamma} \|_{1}$, leveraging existing finite-sample bounds for the Lasso. To this end, let us define
\begin{equation*}
z^{*} :=  \left\| \frac{1}{| \mathcal{A} |}  \sum_{a \in \mathcal{A}}  \frac{1}{n_{a}} \sum_{i=1}^{n_{a}} (\X_{i,\bullet}^{(a)} - \overline{\X}^{(a)})^{\intercal} \left( \Z_{i}^{(a)} - \overline{\Z}^{(a)}  \right)  + \frac{\gamma}{| \mathcal{A} |} \sum_{a \in \mathcal{A}} (\overline{\X}^{(a)})^{\intercal} \overline{\Z}^{(a)}  \right\|_{\infty}
\end{equation*}
From Lemma~\ref{lemma-new:risk} it follows that with probability $1-6 \exp(-t)$ we have $ 2 \| z^{*} \|_{\infty} \le \lambda $.
Now we can use Lemma~\ref{lemma:finite-sample-bound}  to bound $\hat W( \hat b)$ and $ \| \hat b - b^{\gamma} \|_{1}$.  Lemma~\ref{lemma:finite-sample-bound}  follows directly from Theorem 2.2 in \citet{van2016estimation}, but the notation is different. Details can be found in Section~\ref{sec:proof-coroll-refc}. Use Lemma~\ref{lemma:finite-sample-bound} with $b=b^{\gamma}$, $S=S^{*}$, $ \lambda_{\varepsilon} = \| z^{*} \|_{\infty}$ and $\delta = 0.5$. This gives $\underline{\lambda} = \lambda - \| z^{*} \|_{\infty} \ge \frac{\lambda}{2}$, $ \overline{\lambda} = 1.5 \lambda + 0.5 \| z^{*} \|_{\infty} \le 2 \lambda$ and $L \le 8$ to obtain
\begin{align*}
\hat W(\hat b)
 &\le  \frac{4\lambda^{2} |S^{*}|}{\hat{\phi}^{2}(8,S)}, \\
  \text{ and }\| \hat b - b^{\gamma} \|_{1} &\le  \frac{ 8  \lambda |S^{*}|}{\hat{\phi}^{2}(8,S^{*})}.
\end{align*}
Combining these two bounds with equation~\eqref{eq:new-18} yields the desired result.
\end{proof}
\subsubsection{Lemma~\ref{lemma-new:excess-risk}}

\begin{lemma}\label{lemma-new:excess-risk}
Let $X$ follow a centered multivariate Gaussian distribution under $\mathbb{P}_{0}$. Let $\X_{i,\bullet}^{(a)}$, $i=1,\ldots,n_{a}$ be i.i.d. random variables that have the same distribution as $\mu_{\X}^{(a)} + X$ for some deterministic vectors $\mu_{\X}^{(a)} \in \mathbb{R}^{d}$ for $a \in \mathcal{A}$.
Let $\sigma_{\mathrm{max}}^{2} := \max_{k} \mathrm{Var}(X_{k})$, $n_{\mathrm{min}} := \min_{a \in \mathcal{A}} n_{a}$, $\mu_{\mathrm{max}} := \max_{a \in \mathcal{A}} \| \mu_{\X}^{(a)} \|_{\infty}$ and define the empirical means $\overline{\X}^{(a)} := \frac{1}{n_{a}} \sum_{i=1}^{n_{a}} \X_{i,\bullet}^{(a)}$ for $a \in \mathcal{A}$. Let $t \ge 0$ such that
\begin{equation}\label{eq:19}
 |S^{*}|^{2} ( t +  \log(d) + \log( | \mathcal{A}|))/n_{\mathrm{min}}   \le c',
\end{equation}
for some constant $c'>0$. Then, with probability exceeding $1-4 \exp(-t)$, for any vectors $b, b^{\gamma} \in \mathbb{R}^{d}$,
\begin{align*}
   & \mathbb{E}_{0}[(X^{\intercal}(b-b^{\gamma}))^{2}] + \frac{\gamma}{| \mathcal{A} |} \sum_{a \in \mathcal{A}} ((\mu_{\X}^{(a)})^{\intercal}(b-b^{\gamma}))^{2}  \\
\le \, &  \frac{1}{| \mathcal{A}|} \sum_{a \in \mathcal{A}} \frac{1}{n_{a}} \sum_{i=1}^{n_{a}}  \left( (\X_{i,\bullet}^{(a)} - \overline{\X}^{(a)})  (b- b^{\gamma}) \right)^{2}  + \frac{\gamma}{| \mathcal{A}|} \sum_{a \in \mathcal{A}} ( \overline{\X}^{(a)} (b-b^{\gamma}))^{2}  +  \frac{C''}{|S^{*}|}   \| b - b^{\gamma} \|_{1}^{2},
\end{align*}
where $C''$ is a constant that depends on $c'$, $\sigma_{\mathrm{max}}$, $\mu_{\mathrm{max}}$ and $\gamma$.
\end{lemma}
\begin{proof}
We will derive bounds for
\begin{align}\label{eq:10}
  \begin{split}
    \left| \frac{1}{| \mathcal{A}|} \sum_{a \in \mathcal{A}} \frac{1}{n_{a}} \sum_{i=1}^{n_{a}}  \left( (\X_{i,\bullet}^{(a)} - \overline{\X}^{(a)})  (b- b^{\gamma}) \right)^{2} - \mathbb{E}_{0}[(X^{\intercal}(b - b^{\gamma}))^{2}] \right|
  \end{split}
\end{align}
and
\begin{align}\label{eq:17}
  \begin{split}
 \left|  \frac{\gamma}{| \mathcal{A}|} \sum_{a \in \mathcal{A}} ( \overline{\X}^{(a)} (b-b^{\gamma}))^{2} - \frac{\gamma}{| \mathcal{A} |} \sum_{a \in \mathcal{A}} ( (\mu_{\X}^{(a)})^{\intercal}(b-b^{\gamma}))^{2} \right|
  \end{split}
\end{align}
separately. By elementary linear algebra,
\begin{align*}
\begin{split}
 & \left| \frac{1}{| \mathcal{A}|} \sum_{a \in \mathcal{A}} \frac{1}{n_{a}} \sum_{i=1}^{n_{a}}  \left( (\X_{i,\bullet}^{(a)} - \overline{\X}^{(a)})  (b- b^{\gamma}) \right)^{2} - \mathbb{E}_{0}[(X^{\intercal}(b - b^{\gamma}))^{2}] \right| \\
\le &\| b- b^{\gamma} \|_{1}^{2} \left\|  \frac{1}{| \mathcal{A}|} \sum_{a \in \mathcal{A}} \frac{1}{n_{a}} (\X_{i,\bullet}^{(a)} - \overline{\X}^{(a)})^{\intercal} (\X_{i,\bullet}^{(a)} - \overline{\X}^{(a)})  - \mathbb{E}_{0}[X X^{\intercal}]  \right\|_{\infty}
\end{split}
\end{align*}
Now, using $\sum_{i=1}^{n_{a}} (\X_{i,\bullet}^{(a)}- \overline{\X}^{(a)}) =0$ repeatedly,
\begin{align}\label{eq:7}
\begin{split}
   \frac{1}{| \mathcal{A}|} \sum_{a \in \mathcal{A}} \frac{1}{n_{a}} \sum_{i=1}^{n_{a}}  (\X_{i,\bullet}^{(a)} - \overline{\X}^{(a)})^{\intercal} (\X_{i,\bullet}^{(a)} - \overline{\X}^{(a)}) &=  \frac{1}{| \mathcal{A}|} \sum_{a \in \mathcal{A}} \frac{1}{n_{a}} \sum_{i=1}^{n_{a}} (\X_{i,\bullet}^{(a)} - (\mu_{\X}^{(a)})^{\intercal})^{\intercal} (\X_{i,\bullet}^{(a)} - (\overline{\X}^{(a)})) \\
&=   \frac{1}{| \mathcal{A}|} \sum_{a \in \mathcal{A}} \frac{1}{n_{a}} \sum_{i=1}^{n_{a}} (\X_{i,\bullet}^{(a)} - (\mu_{\X}^{(a)})^{\intercal})^{\intercal} (\X_{i,\bullet}^{(a)} - (\mu_{\X}^{(a)})^{\intercal}) \\
&-  \frac{1}{| \mathcal{A}|} \sum_{a \in \mathcal{A}} (\overline{\X}^{(a)} - (\mu_{\X}^{(a)})^{\intercal})^{\intercal} (\overline{\X}^{(a)} - (\mu_{\X}^{(a)})^{\intercal})
\end{split}
\end{align}
We treat these two terms separately. First, using a sub-Gamma tail bound \citep[][Chapter 2]{boucheron2013concentration}, with probability exceeding $1-2 \exp(-t)$,
\begin{align*}
 &\max_{a \in \mathcal{A}} \left\| \frac{1}{n_{a}} \sum_{i=1}^{n_{a}} (\X_{i,\bullet}^{(a)} - (\mu_{\X}^{(a)})^{\intercal})^{\intercal} (\X_{i,\bullet}^{(a)} - (\mu_{\X}^{(a)})^{\intercal}) - \mathbb{E}_{0}[X X^{\intercal}] \right\|_{\infty} \\
\le \, &\sigma_{\text{max}}^{2} \left( \sqrt{ \frac{4t + 4 \log(d^{2} \cdot | \mathcal{A}|)}{n_{\text{min}}}} + \frac{4 t + 4 \log(d^{2} \cdot | \mathcal{A}| )}{n_{\text{min}}} \right).
\end{align*}
Using a  sub-Gaussian tail bound \citep[][Chapter 2]{boucheron2013concentration}, with probability exceeding $1-2 \exp(-t)$,
\begin{equation}\label{eq:11}
 \max_{a \in \mathcal{A}} \| \overline{\X}^{(a)} - (\mu_{\X}^{(a)})^{\intercal} \|_{\infty} \le \sqrt{2 \frac{\sigma_{\text{max}}^{2}}{n_{\text{min}}}(t + \log(d \cdot | \mathcal{A}|))}.
\end{equation}
On this event,
\begin{align*}
  \left\|  \frac{1}{| \mathcal{A}|} \sum_{a \in \mathcal{A}} (\overline{\X}^{(a)} - (\mu_{\X}^{(a)})^{\intercal})^{\intercal} (\overline{\X}^{(a)} - (\mu_{\X}^{(a)})^\intercal) \right\|_{\infty} &\le \max_{a \in \mathcal{A}}  \left\| (\overline{\X}^{(a)} - (\mu_{\X}^{(a)})^\intercal)^{\intercal} (\overline{\X}^{(a)} - (\mu_{\X}^{(a)})^\intercal) \right\|_{\infty} \\  &\le  2 \frac{\sigma_{\text{max}}^{2}}{n_{\text{min}}} (t+ \log(d \cdot | \mathcal{A}|))
\end{align*}
Using these two bounds in equation~\eqref{eq:7}, we obtain the following bound for equation~\eqref{eq:10}:
\begin{align}\label{eq:21}
\begin{split}
& \quad \left| \frac{1}{| \mathcal{A}|} \sum_{a \in \mathcal{A}} \frac{1}{n_{a}} \sum_{i=1}^{n_{a}}  \left( (\X_{i,\bullet}^{(a)} - \overline{\X}^{(a)})  (b- b^{\gamma}) \right)^{2} - \mathbb{E}_{0}[(X^{\intercal}(b - b^{\gamma}))^{2}] \right| \\
& \le \| b -b^{\gamma} \|_{1}^{2} \left\|  \frac{1}{| \mathcal{A}|} \sum_{a \in \mathcal{A}} \frac{1}{n_{a}} \sum_{i=1}^{n_{a}}  (\X_{i,\bullet}^{(a)} - \overline{\X}^{(a)})^{\intercal} (\X_{i,\bullet}^{(a)} - \overline{\X}^{(a)}) - \mathbb{E}_{0}[X X^{\intercal}] \right\|_{\infty} \\
&\le   \| b -b^{\gamma} \|_{1}^{2} \left(  \sigma_{\text{max}}^{2} \left( \sqrt{ \frac{4t + 4 \log(d^{2} \cdot | \mathcal{A}|)}{n_{\text{min}}}} + \frac{4 t + 4 \log(d^{2} \cdot | \mathcal{A}| )}{n_{\text{min}}} \right)  \, +   2 \frac{\sigma_{\text{max}}^{2}}{n_{\text{min}}} (t+ \log(d \cdot | \mathcal{A}|)) \right)
\end{split}
\end{align}
Let us now treat equation~\eqref{eq:17}. Analogously as above,
\begin{align}\label{eq:22}
  \begin{split}
  &  \left|  \frac{\gamma}{| \mathcal{A}|} \sum_{a \in \mathcal{A}} ( \overline{\X}^{(a)} (b-b^{\gamma}))^{2} - \frac{\gamma}{| \mathcal{A} |} \sum_{a \in \mathcal{A}} ((\mu_{\X}^{(a)})^\intercal(b-b^{\gamma}))^{2} \right| \\
\le \, & \gamma \| b -b^{\gamma} \|_{1}^{2} \max_{a \in \mathcal{A}} \left\| (\overline{\X}^{(a)})^{\intercal} \overline{\X}^{(a)}  -  \mu_{\X}^{(a)} (\mu_{\X}^{(a)})^\intercal \right\|_{\infty}.
  \end{split}
\end{align}
Again, we can use a decomposition
\begin{align*}
\begin{split}
   (\overline{\X}^{(a)})^{\intercal} \overline{\X}^{(a)}  -  \mu_{\X}^{(a)}  (\mu_{\X}^{(a)})^{\intercal} &=  (\overline{\X}^{(a)})^{\intercal} (\overline{\X}^{(a)}-(\mu_{\X}^{(a)})^{\intercal}) + (\overline{\X}^{(a)})^{\intercal} (\mu_{\X}^{(a)})^{\intercal}  -  \mu_{\X}^{(a)} (\mu_{\X}^{(a)})^{\intercal} \\
&=  (\overline{\X}^{(a)} - (\mu_{\X}^{(a)})^{\intercal})^{\intercal} (\overline{\X}^{(a)}-(\mu_{\X}^{(a)})^{\intercal}) +  \mu_{\X}^{(a)} (\overline{\X}^{(a)}-(\mu_{\X}^{(a)})^{\intercal}) \\ & \qquad+ (\overline{\X}^{(a)}  -  (\mu_{\X}^{(a)})^{\intercal})^{\intercal} (\mu_{\X}^{(a)})^{\intercal}
\end{split}
\end{align*}
Using this decomposition in equation~\eqref{eq:22},
\begin{align*}
  \begin{split}
     &  \left|  \frac{\gamma}{| \mathcal{A}|} \sum_{a \in \mathcal{A}} ( \overline{\X}^{(a)} (b-b^{\gamma}))^{2} - \frac{\gamma}{| \mathcal{A} |} \sum_{a \in \mathcal{A}} ((\mu_{\X}^{(a)})^{\intercal}(b-b^{\gamma}))^{2} \right| \\
\le \, & \gamma \| b -b^{\gamma} \|_{1}^{2}  \max_{a \in \mathcal{A}} \bigg( \| (\overline{\X}^{(a)} - (\mu_{\X}^{(a)})^{\intercal})^{\intercal} (\overline{\X}^{(a)}-(\mu_{\X}^{(a)})^{\intercal}) \|_{\infty} + \|  \mu_{\X}^{(a)} (\overline{\X}^{(a)}- (\mu_{\X}^{(a)})^{\intercal}) \|_{\infty} \\
& + \| (\overline{\X}^{(a)}  -  (\mu_{\X}^{(a)})^{\intercal})^{\intercal} (\mu_{\X}^{(a)})^{\intercal}\|_{\infty}   \bigg)
  \end{split}
\end{align*}
Recall that $\mu_{\text{max}} = \max_{a \in \mathcal{A}} \| \mu_{\X}^{(a)} \|_{\infty}$. Using the sub-Gaussian tail bound of equation~\eqref{eq:11} for all three terms in the preceding equation, we obtain the following bound:
\begin{align}\label{eq:20}
\begin{split}
 &  \left|  \frac{\gamma}{| \mathcal{A}|} \sum_{a \in \mathcal{A}} ( \overline{\X}^{(a)} (b-b^{\gamma}))^{2} - \frac{\gamma}{| \mathcal{A} |} \sum_{a \in \mathcal{A}} ((\mu_{\X}^{(a)})^{\intercal}(b-b^{\gamma}))^{2} \right| \\
\le \, &   \gamma \| b -b^{\gamma} \|_{1}^{2} 3  \max \left( \mu_{\text{max}} \sqrt{2 \frac{\sigma_{\text{max}}^{2}}{n_{\text{min}}}(t + \log(d \cdot | \mathcal{A}|))} , 2 \frac{\sigma_{\text{max}}^{2}}{n_{\text{min}}}(t + \log(d \cdot | \mathcal{A}|))   \right)
\end{split}
\end{align}
Using equation~\eqref{eq:19} in equation~\eqref{eq:20} and equation~\eqref{eq:21} yields the desired result.
\end{proof}

\subsubsection{Lemma~\ref{lemma-new:risk}}

\begin{lemma}\label{lemma-new:risk} Let $(X,Y)$ follow a centered multivariate Gaussian distribution under $\mathbb{P}_{0}$. Let $\mathcal{A}$ be a finite set and $(\X_{i,\bullet}^{(a)},\Y_{i}^{(a)})$, $i=1,\ldots,n_{a}$ i.i.d. observations that have the same distribution as $(\mu_{\X}^{(a)} + X,\mu_{\Y}^{(a)} + Y)$ for some deterministic quantities $\mu_{\X}^{(a)} \in \mathbb{R}^{d}$ and $\mu_{\Y}^{(a)} \in \mathbb{R}$ for $a \in \mathcal{A}$.  Let $b^{\gamma} \in \mathbb{R}^{d}$ such that
\begin{equation}\label{eq:23}
  b^{\gamma} = \argmin_{b} \mathbb{E}_{0} \left[ \sum_{a \in \mathcal{A}}  \frac{1}{n_{a}}  \sum_{i=1}^{n_{a}}( \Y_{i}^{(a)} - \mu_{\Y}^{(a)} - (\X_{i,\bullet}^{(a)} - (\mu_{\X}^{(a)})^{\intercal}) b)^{2} \right] + \gamma \sum_{a \in \mathcal{A}}  ( \mu_{\Y}^{(a)} - (\mu_{\X}^{(a)})^{\intercal} b )^{2}.
\end{equation}
Define $n_{\mathrm{min}} := \min_{a \in \mathcal{A}} n_{a}$. 
Let $t \ge 0$ such that
\begin{equation*}
 \frac{t+ \log(d \cdot | \mathcal{A} |)}{n_{\text{min}}} \le c'
\end{equation*}
for some constant $c'>0$. Then, with probability exceeding $1- 6 \exp(-t)$,
\begin{align}\label{eq-new:7}
\begin{split}
z^{*} = & \left\| \frac{1}{| \mathcal{A} |}  \sum_{a \in \mathcal{A}}  \frac{1}{n_{a}} \sum_{i=1}^{n_{a}} (\X_{i,\bullet}^{(a)} - \overline{\X}^{(a)})^{\intercal} \left( \Z_{i}^{(a)} - \overline{\Z}^{(a)}  \right)  + \frac{\gamma}{| \mathcal{A} |} \sum_{a \in \mathcal{A}} (\overline{\X}^{(a)})^{\intercal} \overline{\Z}^{(a)}  \right\|_{\infty} \\
\le \, &  \frac{C}{2} \sqrt{\frac{t+ \log( d \cdot | \mathcal{A} |)}{n_{\text{min}}} },
\end{split}
\end{align}
where $\Z^{(a)} = \Y^{(a)} - \X^{(a)} b^{\gamma}$ and  $ \overline{\Z}^{(a)} = \frac{1}{n_{a}} \sum_{i=1}^{n_{a}} \Z_{i}^{(a)} $. The constant $C$ depends on $\mu_{\mathrm{max}}$, $\sigma_{\mathrm{max}}$, $\gamma$ and $c'$. Here, $\sigma_{\mathrm{max}}$ denotes the maximal standard deviation, i.e., $\sigma_{\mathrm{max}}^{2} := \max(\max_{k} \mathrm{Var}(X_{k}),\mathrm{Var}(Y - X^\intercal b^{\gamma}))$ and $\mu_{\mathrm{max}}$ denotes the maximal mean, i.e.,   $\mu_{\mathrm{max}} := \max( \max_{a \in \mathcal{A}} \| \mu_{\X}^{(a)} \|_{\infty} , | \mu_{\Y}^{(a)} - (\mu_{\X}^{(a)})^{\intercal} b^{\gamma} | )$.
\end{lemma}

\begin{proof}
Recall that $\mu_{\Z}^{(a)} = \mu_{\Y}^{(a)} - (\mu_{\X}^{(a)})^{\intercal} b^{\gamma}$ for $a \in \mathcal{A}$. By taking the derivative of the objective functional in equation~\eqref{eq:23} with respect to $b$,
\begin{equation*}
  \mathbb{E}_{0} \left[ \sum_{a \in \mathcal{A}} \frac{1}{n_{a}} \sum_{i=1}^{n_{a}}  (\X_{i,\bullet}^{(a)} - (\mu_{\X}^{(a)})^{\intercal})^{\intercal} \left( \Z_{i}^{(a)} - \mu_{\Z}^{(a)}  \right) \right] + \gamma \sum_{a \in \mathcal{A}}  \mu_{\X}^{(a)}   \mu_{\Z}^{(a)}  = 0.
\end{equation*}
Using this, we can decompose equation~\eqref{eq-new:7}:
\begin{align}\label{eq-new:9}
\begin{split}
& \left\| \frac{1}{| \mathcal{A} |}  \sum_{a \in \mathcal{A}}  \frac{1}{n_{a}} \sum_{i=1}^{n_{a}} (\X_{i,\bullet}^{(a)} - \overline{\X}^{(a)})^{\intercal} \left( \Z_{i}^{(a)} - \overline{\Z}^{(a)}  \right)  + \frac{\gamma}{| \mathcal{A} |} \sum_{a \in \mathcal{A}} (\overline{\X}^{(a)})^{\intercal} \overline{\Z}^{(a)}  \right\|_{\infty} \\
&\le  \frac{1}{| \mathcal{A}|} \sum_{a \in \mathcal{A}} \left\|  \frac{1}{n_{a}} \sum_{i=1}^{n_{a}} (\X_{i,\bullet}^{(a)} - \overline{\X}^{(a)})^{\intercal} \left( \Z_{i}^{(a)} - \overline{\Z}^{(a)}  \right) - \text{Cov}(\X_{i,\bullet}^{(a)}, \Z_{i}^{(a)}) \right\|_{\infty} \\
&+ \frac{\gamma}{| \mathcal{A} |}  \sum_{a \in \mathcal{A}} \left\|  (\overline{\X}^{(a)})^{\intercal}  \overline{\Z}^{(a)} - \mu_{\X}^{(a)} \mu_{\Z}^{(a)} \right\|_{\infty}
\end{split}
\end{align}
As $\sum_{i=1}^{n_{a}} (\Z_{i}^{(a)} - \overline{\Z}^{(a)}) = 0$ and $\sum_{i=1}^{n_{a}} (\X_{i,\bullet}^{(a)} - \overline{\X}^{(a)}) = 0$,
\begin{align*}
 & \frac{1}{n_{a}} \sum_{i=1}^{n_{a}} (\X_{i,\bullet}^{(a)} - \overline{\X}^{(a)})^{\intercal} ( \Z_{i}^{(a)} - \overline{\Z}^{(a)}) \\
= \, & \frac{1}{n_{a}} \sum_{i=1}^{n_{a}} (\X_{i,\bullet}^{(a)} -  (\mu_{\X}^{(a)})^{\intercal})^{\intercal}  ( \Z_{i}^{(a)} - \overline{\Z}^{(a)}) \\
= \, & \frac{1}{n_{a}} \sum_{i=1}^{n_{a}} (\X_{i,\bullet}^{(a)} -  (\mu_{\X}^{(a)})^{\intercal})^{\intercal} ( \Z_{i}^{(a)} - \mu_{\Z}^{(a)} ) + \frac{1}{n_{a}} \sum_{i=1}^{n_{a}} ( \X_{i,\bullet}^{(a)} - (\mu_{\X}^{(a)})^{\intercal})^{\intercal} ( \mu_{\Z}^{(a)} - \overline{\Z}^{(a)}) \\
 =  \, & \frac{1}{n_{a}} \sum_{i=1}^{n_{a}} (\X_{i,\bullet}^{(a)} -  (\mu_{\X}^{(a)})^{\intercal})^{\intercal} ( \Z_{i}^{(a)} - \mu_{\Z}^{(a)} ) +  ( \overline{\X}^{(a)} - (\mu_{\X}^{(a)})^{\intercal})^{\intercal} ( \mu_{\Z}^{(a)} - \overline{\Z}^{(a)}).
\end{align*}
Similarly,
\begin{align*}
  &  (\overline{\X}^{(a)})^{\intercal}  \overline{\Z}^{(a)} - \mu_{\X}^{(a)} \mu_{\Z}^{(a)} \\
 = \quad &    (\overline{\X}^{(a)} - (\mu_{\X}^{(a)})^{\intercal}  )^{\intercal} ( \overline{\Z}^{(a)} - \mu_{\Z}^{(a)} ) +  ( \overline{\X}^{(a)} - (\mu_{\X}^{(a)})^{\intercal} )^{\intercal} \mu_{\Z}^{(a)}  + \mu_{\X}^{(a)} ( \overline{\Z}^{(a)} - \mu_{\Z}^{(a)} ).
\end{align*}
Combining these decompositions with equation~\eqref{eq-new:9},
\begin{align}\label{eq-new:10}
\begin{split}
  & \left\| \frac{1}{| \mathcal{A} |}  \sum_{a \in \mathcal{A}}  \frac{1}{n_{a}} \sum_{i=1}^{n_{a}} (\X_{i,\bullet}^{(a)} - \overline{\X}^{(a)})^{\intercal} \left( \Z_{i}^{(a)} - \overline{\Z}^{(a)}  \right)  + \frac{\gamma}{| \mathcal{A} |} \sum_{a \in \mathcal{A}} (\overline{\X}^{(a)})^{\intercal} \overline{\Z}^{(a)} \right\|_{\infty} \\
\le \, & \max_{a \in \mathcal{A}}  \left\| \frac{1}{n_{a}} \sum_{i=1}^{n_{a}} (\X_{i,\bullet}^{(a)} -  (\mu_{\X}^{(a)})^{\intercal})^{\intercal} ( \Z_{i}^{(a)} - \mu_{\Z}^{(a)} ) - \text{Cov}(\X_{i,\bullet}^{(a)}, \Z_{i}^{(a)}) \right\|_{\infty} \\
+ \, & (\gamma+1) \max_{a \in \mathcal{A}} \|  ( \overline{\X}^{(a)} - (\mu_{\X}^{(a)})^{\intercal})^{\intercal} ( \mu_{\Z}^{(a)} - \overline{\Z}^{(a)}) \|_{\infty} \\
+ \, &  \gamma \max_{a \in \mathcal{A}} \| ( \overline{\X}^{(a)} - (\mu_{\X}^{(a)})^{\intercal} )^{\intercal} \mu_{\Z}^{(a)} \|_{\infty} \\
+ \, &  \gamma \max_{a \in \mathcal{A}} \|  \mu_{\X}^{(a)} ( \overline{\Z}^{(a)} - \mu_{\Z}^{(a)} )\|_{\infty}
\end{split}
\end{align}
Using a sub-Gaussian tail bound  \citep[][Chapter 2]{boucheron2013concentration}, with probability exceeding $1 - 4 \exp(-t)$ we have
\begin{equation*}
 \max_{a \in \mathcal{A}} \max ( \| \overline{\X}^{(a)} - (\mu_{\X}^{(a)})^{\intercal} \|_{\infty} ,|\overline{\Z}^{(a)} - \mu_{\Z}^{(a)}|) \le \sqrt{ \frac{2 \sigma_{\text{max}}^{2} (\log(d \cdot |\mathcal{A}|) + t )}{n_{\text{min}}}},
\end{equation*}
With a sub-Gamma tail bound \citep[][Chapter 2]{boucheron2013concentration}, with probability exceeding $1 - 2 \exp(-t)$ we have that
\begin{align*}
 &\max_{a \in \mathcal{A}} \left\| \frac{1}{n_{a}} \sum_{i=1}^{n_{a}} (\X_{i,\bullet}^{(a)} - (\mu_{\X}^{(a)})^{\intercal})^{\intercal} (\Z_{i}^{(a)} - \mu_{\Z}^{(a)}) - \text{Cov}_{\text{train}}(\X_{i,\bullet}^{(a)},\Z_{i}^{(a)})\right\|_{\infty} \\
\le \qquad & \sigma_{\text{max}}^{2} \left( \frac{4 t + 4 \log(d \cdot |\mathcal{A} |)}{n_{\text{min}}} + \sqrt{\frac{4t+4 \log(d \cdot | \mathcal{A} |)}{n_{\text{min}}} } \right).
\end{align*}
Recall that by assumption
\begin{equation*}
 \frac{t+ \log(d \cdot | \mathcal{A} |)}{n_{\text{min}}} \le c'.
\end{equation*}
Using these bounds in equation~\eqref{eq-new:10}, we obtain that with probability exceeding $1 - 6 \exp(-t)$
\begin{align*}
  \| z^{*} \|_{\infty} \le \frac{C}{2} \sqrt{\frac{t+ \log(d \cdot | \mathcal{A} |)}{n_{\text{min}}} },
\end{align*}
where $C$ depends on $\sigma_{\text{max}}$, $\mu_{\text{max}}$, $c'$ and $\gamma$.
\end{proof}

\subsubsection{Lemma~\ref{lemma:finite-sample-bound}}\label{sec:proof-coroll-refc}

The following result provides a bound on $ \| \hat b - b^{\gamma} \|_{1}$ and  $\hat W( \hat b)$, with $\hat{W}(\bullet)$ and $\hat b$ defined as in Section~\ref{sec:proof-theor-refthm:f}. It follows directly from Theorem 2.2 in \citet{van2016estimation}, but the notation is different.
\begin{lemma}\label{lemma:finite-sample-bound}
  Let $\lambda_{\varepsilon}$ satisfy $\lambda_{\varepsilon} \ge \| z^{*} \|_{\infty}$. Let $0 \le \delta < 1$ be arbitrary and define for $\lambda > \lambda_{\varepsilon}$ and all $S \subseteq  \{1,\ldots, d\}$
\begin{align*}
  \underline{\lambda} &:= \lambda - \lambda_{\varepsilon}, \qquad \overline{\lambda} := \lambda + \lambda_{\varepsilon} + \delta \underline{\lambda}, \qquad L := \frac{\overline{\lambda}}{(1-\delta) \underline{\lambda}}, \\
\hat{\phi}^{2}(L,S)&:= \\
\min_{\| b_{S} \|_{1} =1 , \| b_{-S} \|_{1} \le L} |S| & \left( \frac{1}{| \mathcal{A}|} \sum_{a \in \mathcal{A}} \frac{1}{n_{a}} \sum_{i=1}^{n_{a}}  \left( (\X_{i}^{(a)} - \overline{\X}^{(a)}) b \right)^{2}  + \frac{\gamma}{| \mathcal{A}|} \sum_{a \in \mathcal{A}} ( \overline{\X}^{(a)} b)^{2}  \right).
\end{align*}
Then for all $b \in \mathbb{R}^{d}$ and all sets $S$,
\begin{align*}
  2 \delta \underline{\lambda} \| \hat b - b\|_{1} +&  \hat{W}(\hat b)  \le \hat{W}(b)+ \frac{\overline{\lambda}^{2} |S|}{\hat{\phi}^{2}(L,S)} + 4 \lambda \| b_{-S}\|_{1}.
\end{align*}
\end{lemma}

\begin{proof}
From a mathematical perspective, the proof of this result is immediate. However, it requires a change of notation. Define $\tilde n := \sum_{a \in \mathcal{A}} n_{a} + |\mathcal{A}|$. With some abuse of notation we can define $\tilde \Y \in \mathbb{R}^{\tilde n}$ as the row-wise concatenation of $\sqrt{\frac{\tilde n}{| \mathcal{A}| n_{a}}}(\Y^{(a)}- \overline{\Y}^{(a)}) \in \mathbb{R}^{n_{a}}$, $a \in \mathcal{A}$ and $\sqrt{\frac{\tilde n \gamma}{| \mathcal{A}|}} \cdot\overline{\Y}^{(a)} \in \mathbb{R}$, $a \in \mathcal{A}$. Analogously define $\tilde \X \in \mathbb{R}^{ \tilde n \times d}$ as the row-wise concatenation of  $\sqrt{\frac{\tilde n}{| \mathcal{A}| n_{a}}}(\X^{(a)}- \overline{\X}^{(a)})$, $a \in \mathcal{A}$ and  $\sqrt{\frac{\tilde n \gamma}{| \mathcal{A}|}}\cdot\overline{\Y}^{(a)}$, $a \in \mathcal{A}$. Recall that
\begin{equation}\label{eq:24}
  \hat W(b) = \frac{1}{| \mathcal{A}|} \sum_{a \in \mathcal{A}} \frac{1}{n_{a}} \sum_{i=1}^{n_{a}}  \left( (\X_{i,\bullet}^{(a)} - \overline{\X}^{(a)})  (b- b^{\gamma}) \right)^{2}  + \frac{\gamma}{| \mathcal{A}|} \sum_{a \in \mathcal{A}} ( \overline{\X}^{(a)} (b-b^{\gamma}))^{2}.
\end{equation}
 By this definition, we can rewrite $\hat{W}(b)$ as
\begin{equation*}
  \hat W(b) = \frac{1}{\tilde n} \sum_{i=1}^{\tilde n} ( \tilde \X_{i, \bullet} (b - b^{\gamma}))^{2}.
\end{equation*}
Furthermore, \emph{anchor regression} $\hat{b}$ minimizes the functional
\begin{equation*}
   \frac{1}{\tilde n} \sum_{i=1}^{\tilde n} ( \tilde \Y_{i} - \tilde \X_{i, \bullet} b)^{2} + 2 \lambda \| b \|_{1}.
\end{equation*}
We can also rewrite $z^{*}$ as
\begin{equation}\label{eq:25}
  z^{*} = \frac{1}{\tilde n} \left\| \tilde \X^{\intercal} (\tilde \Y - \tilde \X b^{\gamma}) \right\|_{\infty}
\end{equation}
and define $\tilde \varepsilon := \tilde \Y - \tilde \X b^{\gamma} $.
Now let us cite the following Theorem.
\begin{theorem}[Theorem 2.2 in  \citet{van2016estimation}]
  Let $\lambda_{\varepsilon}$ satisfy $\lambda_{\varepsilon} \ge   \frac{1}{\tilde n} \left\| \tilde \X^{\intercal} \tilde \varepsilon \right\|_{\infty}$. Let $0 \le \delta < 1$ be arbitrary and define for $\lambda > \lambda_{\varepsilon}$ and all sets $S \subseteq \{1,\ldots,d\}$
\begin{align*}
  \underline{\lambda} &:= \lambda - \lambda_{\varepsilon}, \qquad \overline{\lambda} := \lambda + \lambda_{\varepsilon} + \delta \underline{\lambda}, \qquad L := \frac{\overline{\lambda}}{(1-\delta) \underline{\lambda}}, \\
\hat{\phi}^{2}(L,S)&:= \min_{\| b_{S} \|_{1} =1 , \| b_{-S} \|_{1} \le L} |S|  \frac{1}{ \tilde{n}}  \sum_{i=1}^{ \tilde{n}} ( \tilde{\X}_{i,\bullet}(b_{S}- b_{-S}))^{2}.
\end{align*}
Then for all $b $ and all sets $S$,
\begin{equation*}
  2 \delta \underline{\lambda} \| \hat b - b\|_{1} +\frac{1}{\tilde n}  \sum_{i=1}^{\tilde n} ( \tilde \X_{i,\bullet}(\hat b- b^{\gamma}))^{2}  \le   \frac{1}{\tilde n}  \sum_{i=1}^{\tilde n} ( \tilde \X_{i,\bullet}(b^{\gamma}- b))^{2} + \frac{\overline{\lambda}^{2} |S|}{\hat{\phi}^{2}(L,S)} + 4 \lambda \| b_{-S}\|_{1}.
\end{equation*}
\end{theorem}
Using the above-mentioned change of notation concludes the proof of Lemma~\ref{lemma:finite-sample-bound}.

\end{proof}

\subsection{Proof of Theorem~\ref{thm:invariance}}
\begin{proof}
Define $\mathbf{G} := \Etrain[A A^{\intercal}]$. Recall that
\begin{equation*}
I = \{ b: \Etrain[A \cdot (Y-X^\intercal b)] =0 \}
\end{equation*}
and define
\begin{align*}
J := \,  &\{ b : \text{ for all $v$ in the span of $\mathbf{M}$ we have that} \\
&     Y - X^\intercal b \text{ has the same distribution under $\mathbb{P}_{v}$ as under $\Ptrain$}\}.
\end{align*}
We will show $I = J$. For simplicity, in the following we will write $w_{b} := ((\mathrm{Id}-\mathbf{B})_{d+1, \bullet}^{-1}-  b^{\intercal}(\mathrm{Id}-\mathbf{B})_{1:d, \bullet}^{-1} )^{\intercal} $. Using the model assumptions of Section~\ref{sec:setting-notation},
\begin{align*}
 \Etrain [A \cdot (  Y - X^\intercal b)] &=  \Etrain \left[ A \cdot \left(   w_{b}^{\intercal} ( \varepsilon + \mathbf{M} A ) \right) \right] \\
&=  \Etrain \left[ A \cdot \left(  w_{b}^{\intercal}  \mathbf{M} A  \right) \right].
\end{align*}
 As $\Etrain[A A^{\intercal}] = \mathbf{G}$, it can be rewritten as $A = \mathbf{G}^{1/2} Z$ with $\Etrain[Z Z^{\intercal}] = \mathrm{Id}$. Hence,
\begin{align*}
  \Etrain [A \cdot (Y-X^\intercal b)] &= \Etrain \left[ (\mathbf{G}^{1/2} Z) \cdot  \left( w_{b}^{\intercal}  \mathbf{M} \mathbf{G}^{1/2} Z \right) \right] \\
&= \mathbf{G}^{1/2} \Etrain \left[  Z \cdot  \left( w_{b}^{\intercal} \mathbf{M} \mathbf{G}^{1/2} Z \right) \right] \\
&= \mathbf{G}^{1/2}  \left( w_{b}^{\intercal} \mathbf{M} \mathbf{G}^{1/2}\right)^{\intercal}
\end{align*}
As $\mathbf{G}$ is assumed to be invertible, $\Etrain[A \cdot(Y-X^\intercal b)] = 0$ if and only if $ w_{b}^{\intercal} \mathbf{M} = 0$. This implies
\begin{equation}\label{eq:31}
  I = \{ b : w_{b}^{\intercal} \mathbf{M} = 0 \}.
\end{equation}
Using the model assumptions of Section~\ref{sec:setting-notation}, under $\mathbb{P}_{v}$,
\begin{align*}
    Y - X^\intercal b &=  w_{b}^{\intercal} (\varepsilon+v), \\
 \text{ and under $\Ptrain$ we have } Y - X^\intercal b &= w_{b}^{\intercal} (\varepsilon+\mathbf{M}A ).
\end{align*}
The distributions of these random variables are equal for all $v \in \text{span}(\mathbf{M})$ if and only if  $ w_{b}^{\intercal} \mathbf{M} = 0$.
Hence,
\begin{align*}
J = \,  &\{ b : w_{b}^{\intercal} \mathbf{M} = 0 \}.
\end{align*}
Using equation~\eqref{eq:31} concludes the proof.
\end{proof}

\subsection{Figures for evaluating replicabilty} \label{sec:addfig}

We show here additional results for replicability of variable
  selection in the GTEx data.
\begin{figure}[ht]
  \begin{center}
  \includegraphics[width=0.48\textwidth]{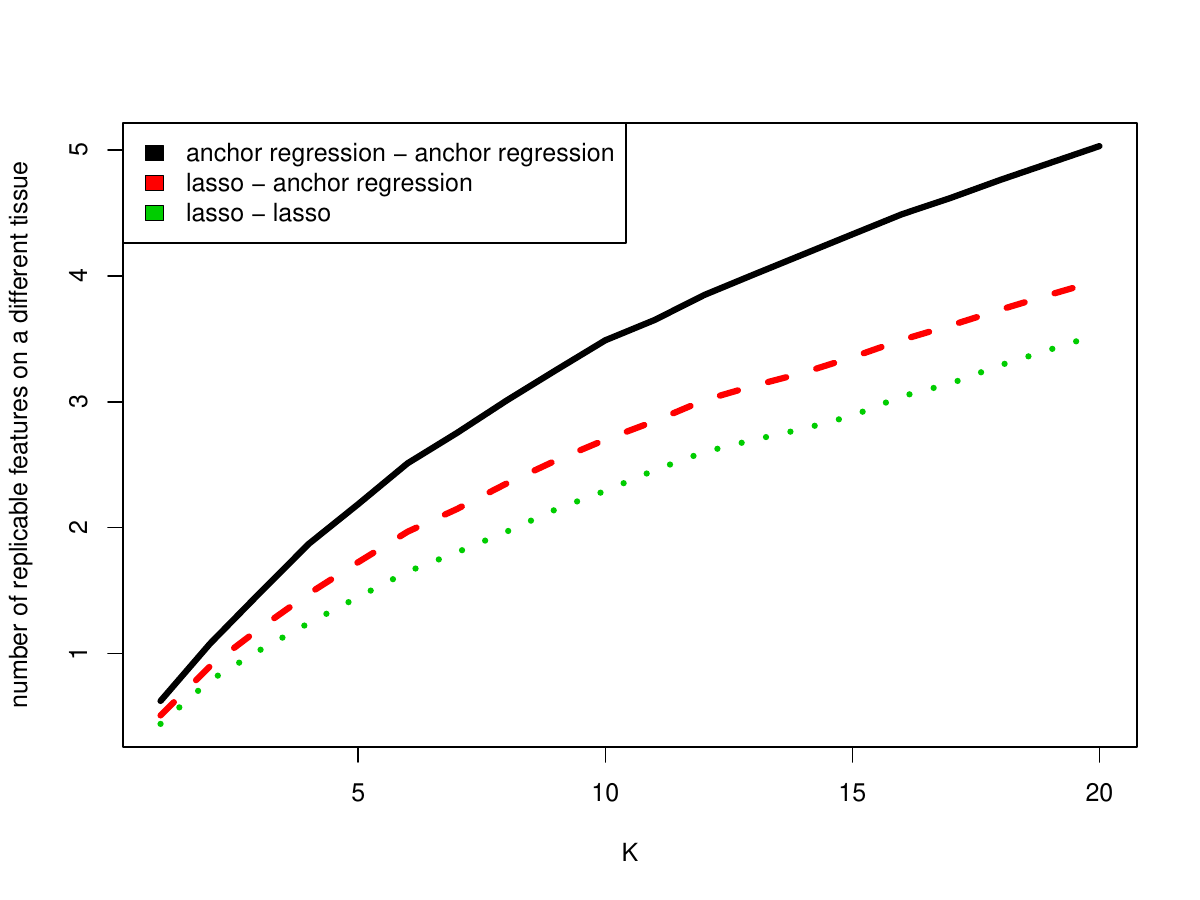}\hfill
  \includegraphics[width=0.48\textwidth]{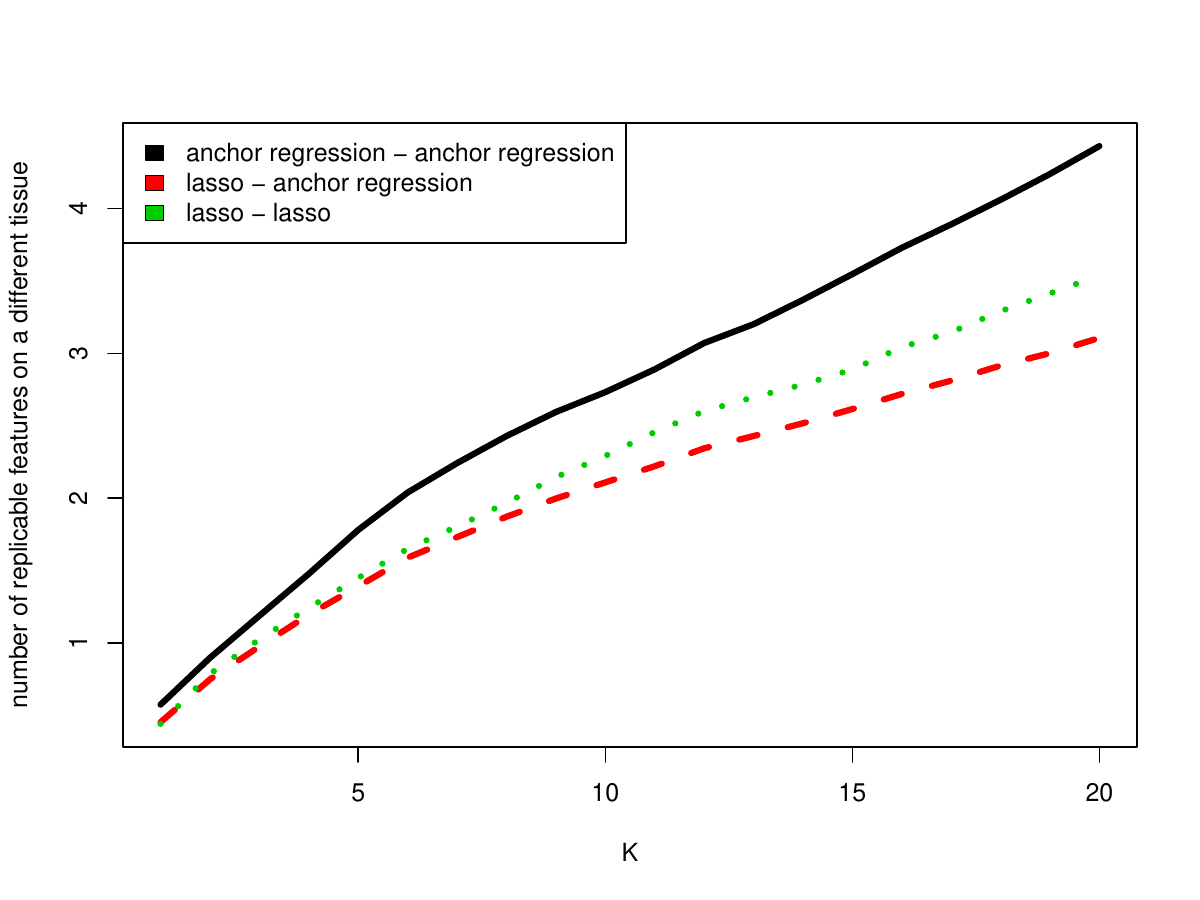}
  \end{center}
  \caption{Replicability of variable selection on GTEx
    data. Same caption as in Figure~\ref{fig:gtexcomparison}, but now with
  $a_{y,k,t} := \min_{\gamma \in [0,.25]}| \hat b_k^{\gamma,\lambda} |$ on the left, and with $a_{y,k,t} := \min_{\gamma \in [0,16]}| \hat b_k^{\gamma,\lambda} |$ on the right.}
\end{figure}

\begin{figure}[ht]
  \begin{center}
  \includegraphics[width=0.48\textwidth]{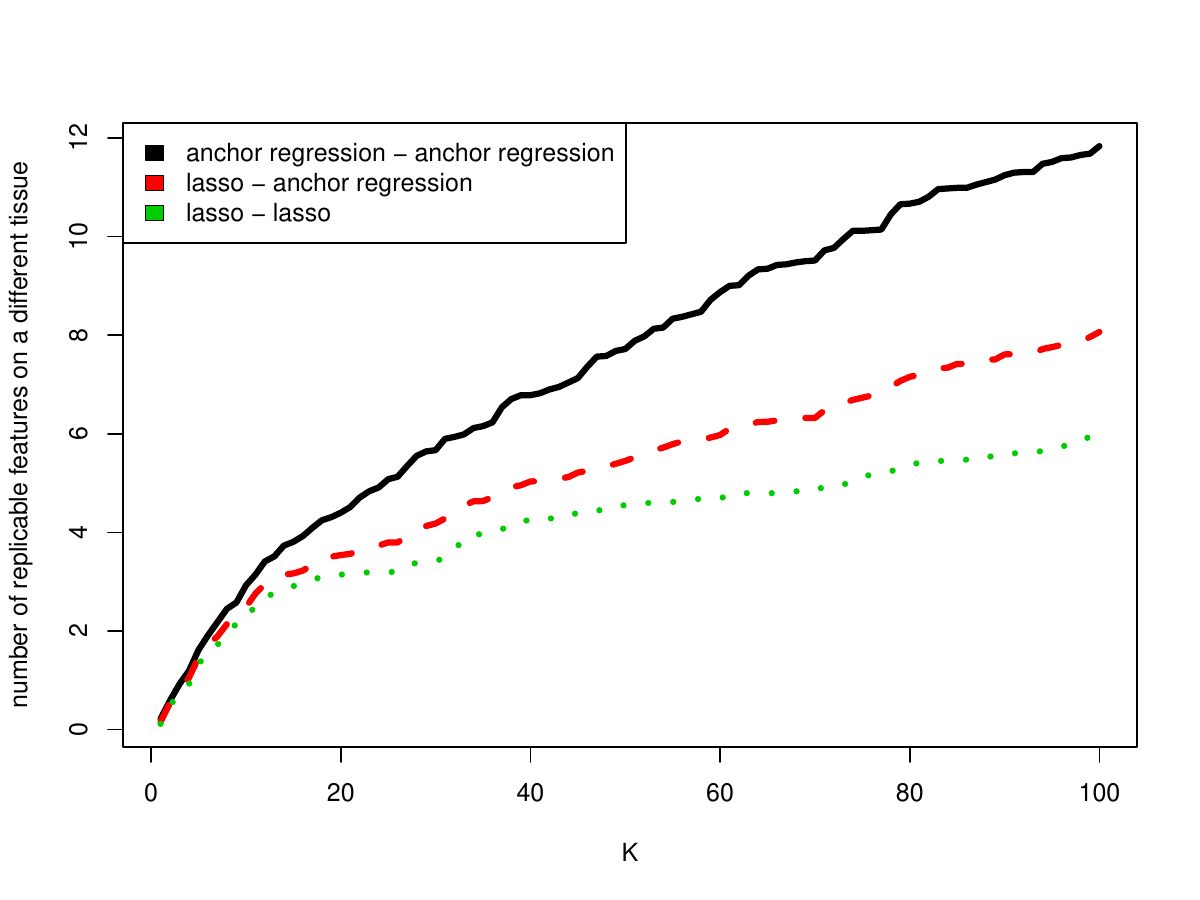}\hfill
  \includegraphics[width=0.48\textwidth]{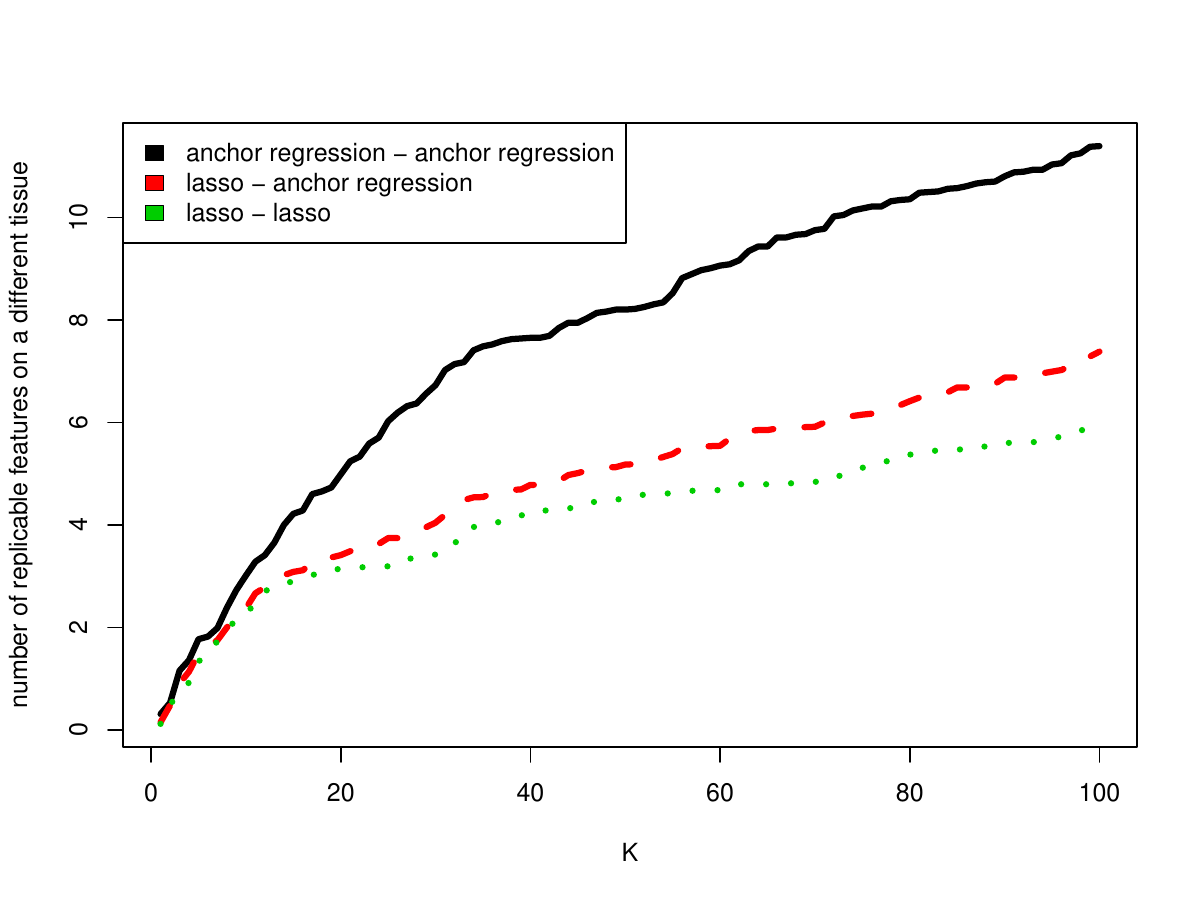}
  \end{center}
  \caption{Replicability of variable selection on GTEx
      data. Same caption as in
      Figure~\ref{fig:gtexcomparison}, but with 
  $a_{y,k,t} := \min_{\gamma \in [0,.25]}| \hat b_k^{\gamma,\lambda} |$ on
  the left, and with $a_{y,k,t} := \min_{\gamma \in [0,16]}| \hat
  b_k^{\gamma,\lambda} |$ on the right. Furthermore, the variable ranking
  is done over the 200 choices of the target variable $y$ and
    averaging the results, instead of a
  fixed target $y$.}
\end{figure}
\begin{figure}[ht]
  \begin{center}
  \includegraphics[width=0.48\textwidth]{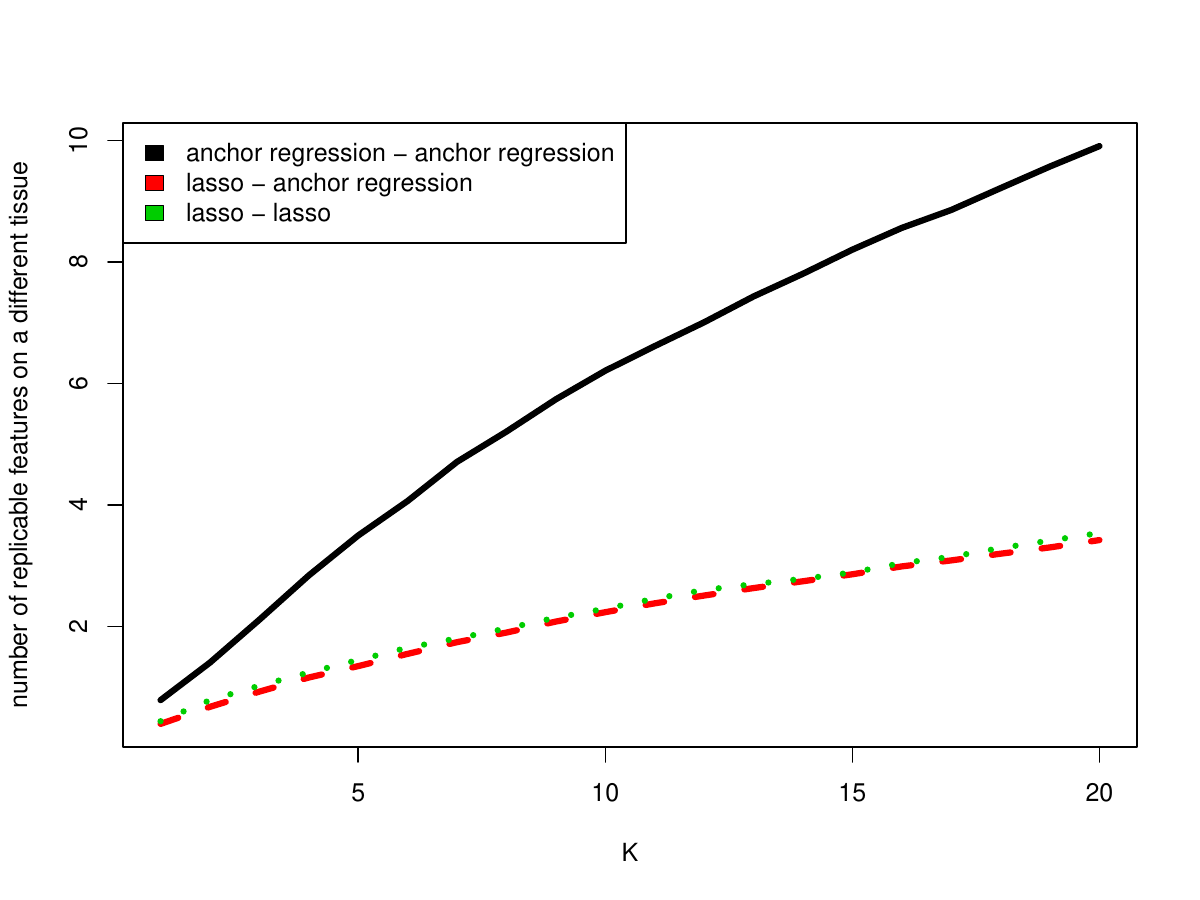}\hfill
 \includegraphics[width=0.48\textwidth]{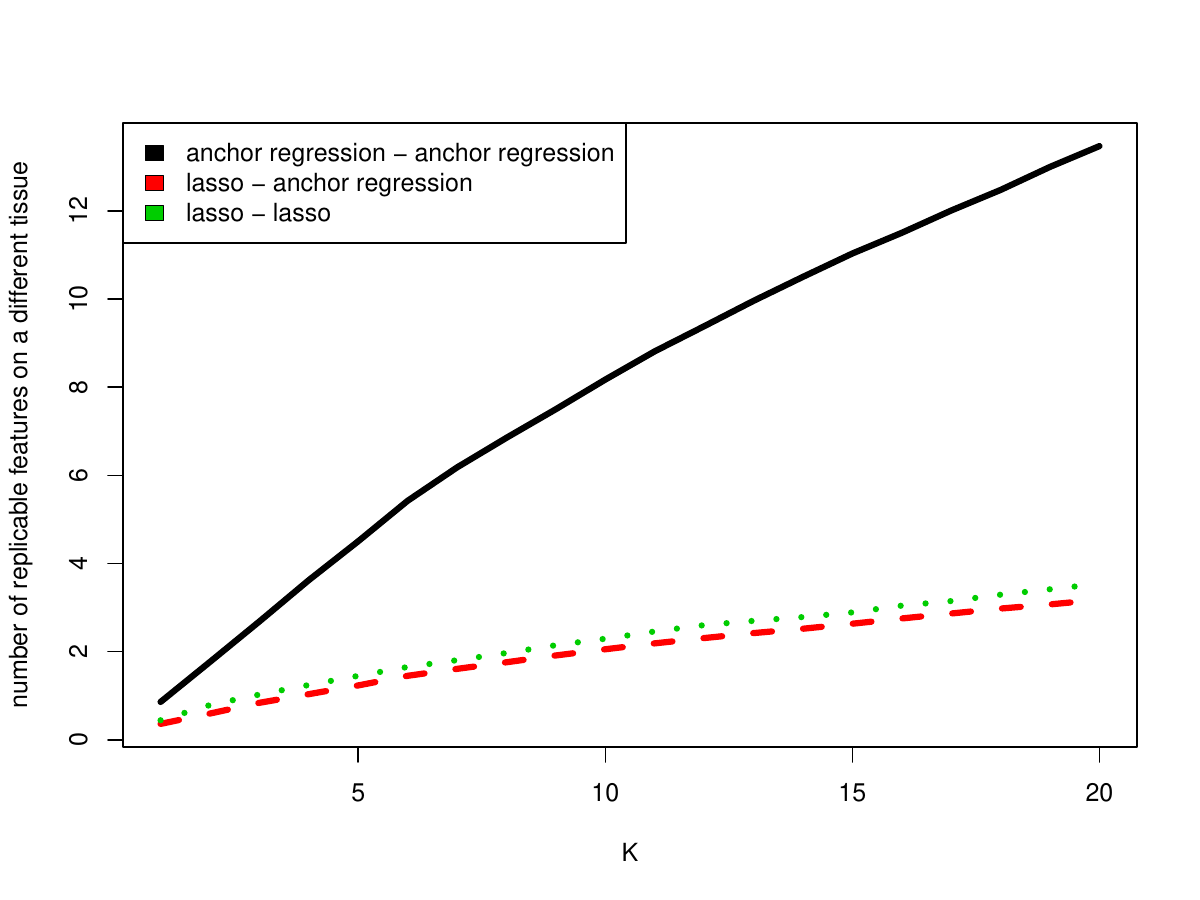}
  \end{center}
  \caption{Replicability of variable selection on GTEx
      data. Same caption as in Figure~\ref{fig:gtexcomparison}, but with $a_{y,k,t} = |\hat b^{\gamma,\lambda}_{k}|$ for $\gamma = 8$ (left) and $\gamma = 16$ (right). While the coefficients show high replicability it depends on the interpretation of the anchor whether the coefficients are scientifically meaningful quantities. This is further discussed at the end of Section~\ref{sec:repl-param-bright}.}
\end{figure}

\subsection{Figures for the bike sharing application} \label{sec:add-hour}

\begin{figure}
\begin{center}
\includegraphics[scale=0.4]{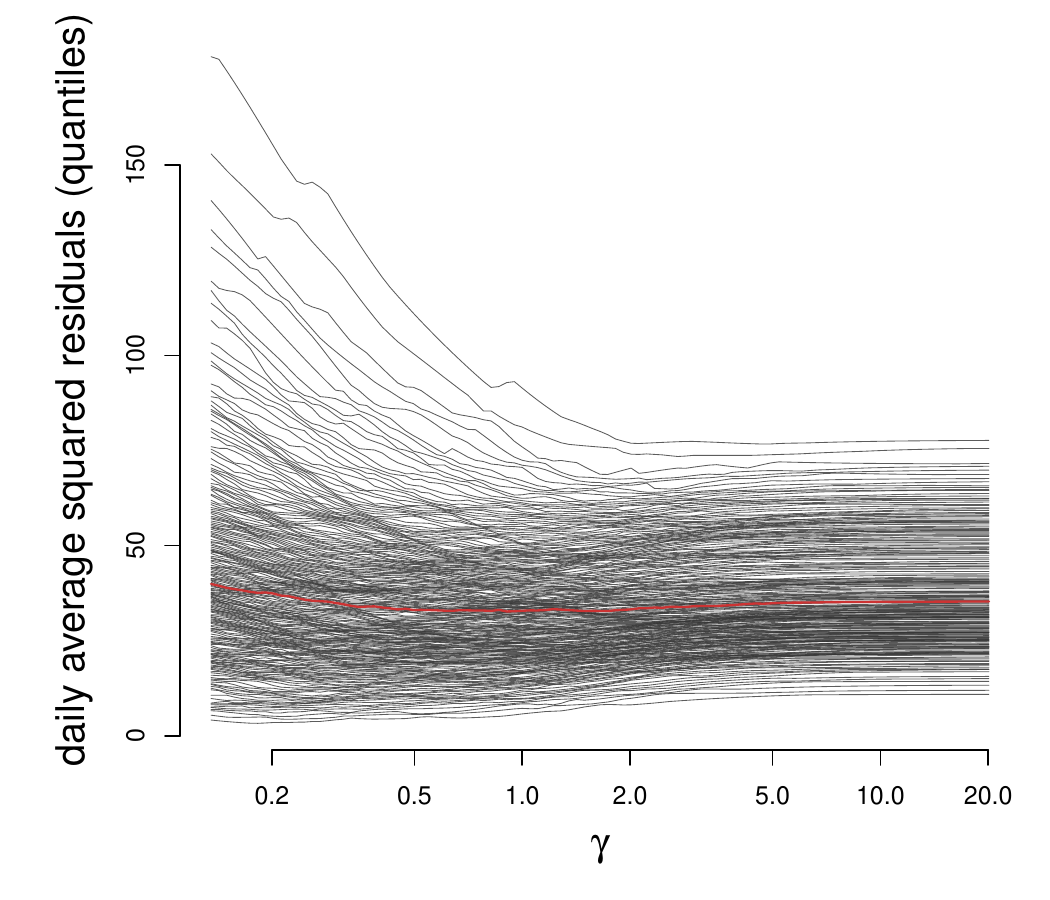}
\end{center}
\caption{
The plot is computed similarly as in Figure~\ref{fig:cvxquantiles}, but without removing the effect of working day, weekday and holiday in a pre-processing step. The two plots are very similar, i.e.\ practically it makes little difference whether the effect of working day, weekday and holiday are removed in a pre-processing step or not. }
\label{fig:pre-processing}
\end{figure}

\begin{figure}
\begin{center}
\includegraphics[scale=0.4]{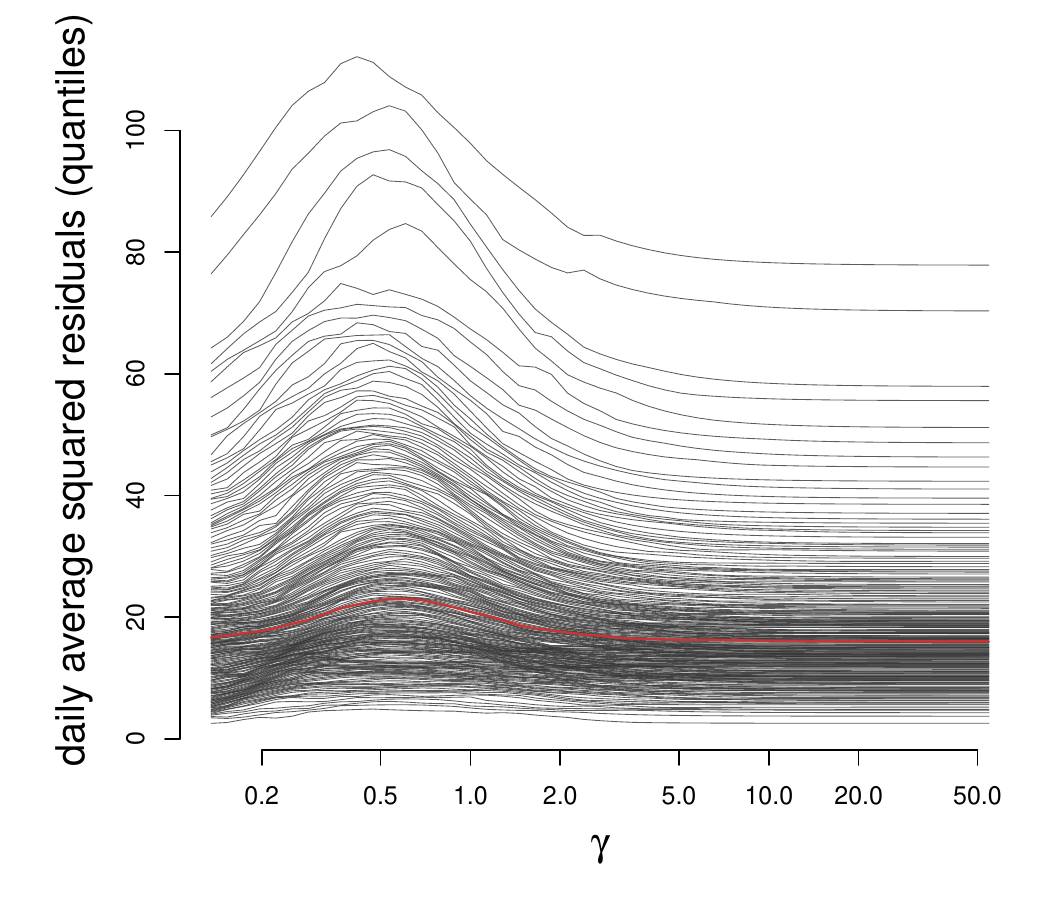}
\end{center}
\caption{
The plot is computed similarly as in Figure~\ref{fig:cvxquantiles}, but with the modified anchor procedure which is described at the end of Section~\ref{sec:bike-sharing-data}. For large quantiles of the conditional loss, $\gamma \gg 1$ outperforms $\gamma < 1$, but the relationship is not monotonous and the performance of $\gamma \approx 0$ and $\gamma = 50$ are close. }
\label{fig:hour}
\end{figure}

\begin{figure}
\begin{center}
\includegraphics[scale=0.6]{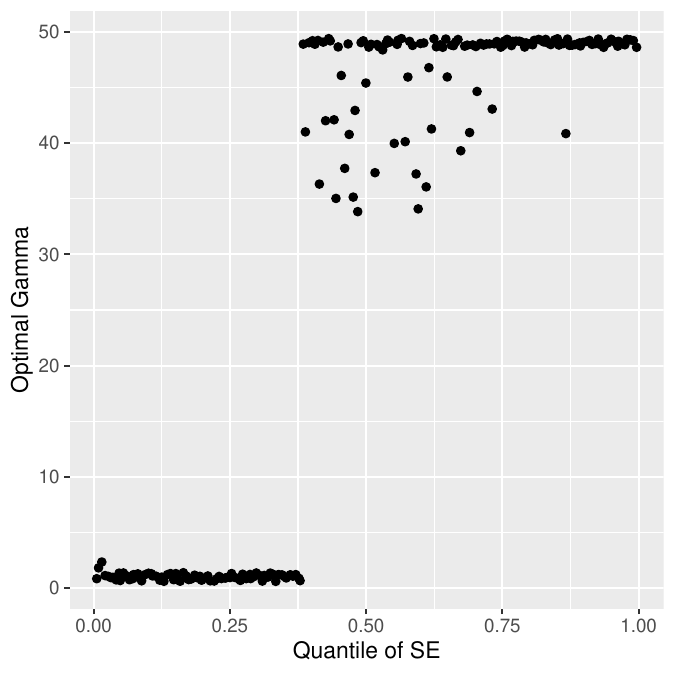}
\includegraphics[scale=0.6]{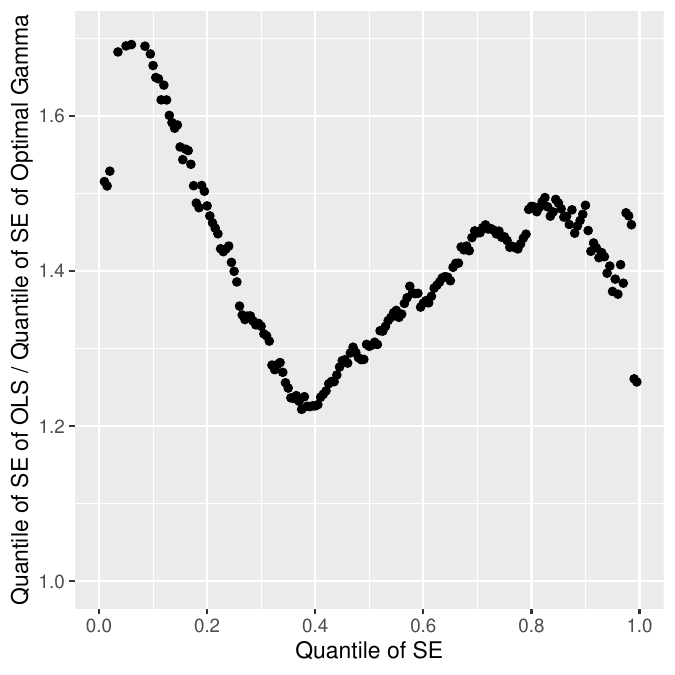}
\end{center}
\caption{
  The plots were computed similarly as in Figure~\ref{fig:optgamma}, but with the modified anchor procedure which is described at the end of Section~\ref{sec:bike-sharing-data}.  For small quantiles, $\gamma \approx 0$ is optimal, while for large quantiles $\gamma \approx 50$ is optimal. However, as can be seen in Figure~\ref{fig:hour}, the performance of $\gamma = 0$ and $\gamma \approx 50$ are close. The anchor regression procedure performs better than ordinary least-squares for all considered quantiles.} \label{fig:optgamma_hour}
\end{figure}

\end{document}